\documentclass[a4paper, 12 pt, twoside, reqno]{amsart}

\usepackage[english]{babel}
\usepackage[utf8]{inputenc}

\usepackage{setspace}
\usepackage[euler-digits,euler-hat-accent]{eulervm}
\usepackage{times}
\usepackage{amsmath}
\usepackage{amsfonts}
\usepackage{amssymb}
\usepackage{amsmath}
\usepackage{amsthm}
\usepackage{graphicx}
\usepackage{array}
\usepackage{color}
\usepackage{mathrsfs}
\usepackage{hyperref}
\usepackage{eucal}
\usepackage{esint} %%% for \fint
\usepackage{tikz}
\usepackage{upgreek}
\usepackage{enumitem}

\allowdisplaybreaks

\theoremstyle{plain}
\newtheorem{theorem}{Theorem}[section]
\newtheorem{corollary}[theorem]{Corollary}
\newtheorem{proposition}[theorem]{Proposition}
\newtheorem{lemma}[theorem]{Lemma}

\theoremstyle{definition}
\newtheorem{definition}[theorem]{Definition}

\theoremstyle{remark}
\newtheorem{remark}[theorem]{Remark} 
\newtheorem{example}[theorem]{Example}

\numberwithin{equation}{section}
\numberwithin{figure}{section}
\numberwithin{table}{section}

\newcommand{\R}{\mathbb{R}}
\newcommand{\N}{\mathbb{N}}
\newcommand{\C}{\mathbb{C}}                           

\newcommand{\Z}{\mathbb{Z}}

\newcommand{\s}[1]{\CMcal{#1}}
\newcommand{\f}[1]{\mathcal{#1}}                  
\newcommand{\bb}[1]{\mathscr{#1}}
\newcommand{\rr}[1]{\mathfrak{#1}}
\newcommand{\n}[1]{\mathbb{#1}}

\newcommand{\nnnorm}[1]{\lvert\hspace{-0.08em}\rvert\hspace{-0.08em}\rvert#1\lvert\hspace{-0.08em}\rvert\hspace{-0.08em}\rvert}

\newcommand{\ketbra}[2]{|#1\rangle\langle#2|}

\newcommand{\expo}[1]{\,\mathrm{e}^{#1}\,}                 

\newcommand{\dd}{\,\mathrm{d}}
\newcommand{ \ii}{\,\mathrm{i}\,}

\newcommand{\virg}[1]{\lq\lq#1\rq\rq}                \newcommand{\ie}{\textsl{i.\,e.\,}}
\newcommand{\eg}{\textsl{e.\,g.\,}}
\newcommand{\cf}{\textsl{cf}.\,}
\newcommand{\etc}{\textsl{etc}.\,}

\DeclareMathOperator{\Tr}{Tr}

\DeclareMathOperator{\esssup}{ess\,sup}

\hypersetup{
pdftoolbar=true,        % show Acrobat’s toolbar?
pdfmenubar=true,        % show Acrobat’s menu?
pdffitwindow=true,     % window fit to page when opened
pdfstartview=true,    % fits the width of the page to the window
pdftitle={NCG-Landau-II},    % title
pdfauthor={Giuseppe De Nittis},     % author
breaklinks=true, %permet le retour a la ligne dans les liens trop longs
colorlinks=true,       % false: boxed links; true: colored links
linkcolor=purple,         % color of internal links (change box color 
citecolor=teal, % color of links to bibliography
urlcolor=blue, %couleur des hyperliens
bookmarksopen=true, %si les signets Acrobat sont crees,
filecolor=magenta,      % color of file links
}

\begin{document}

\title[Dixmier trace and  the IDOS of perturbed magnetic operators]{
Dixmier trace and  the DOS of perturbed magnetic operators}

\author[F. Belmonte]{Fabi\'an Belmonte}

\address[F. Belmonte]{Departamento de Matem\'aticas, Universidad Cat\'olica del Norte, Antofagasta, Chile}
\email{fbelmonte@ucn.cl}

\author[G. De~Nittis]{Giuseppe De Nittis}

\address[G. De~Nittis]{Facultad de Matemáticas \& Instituto de Física,
  Pontificia Universidad Católica de Chile,
  Santiago, Chile.}
\email{gidenittis@mat.uc.cl}

\vspace{2mm}

\date{\today}

\begin{abstract}
The main goal of this work is to provide a description of the {trace per unit volume} in terms of the {Dixmier trace} (regularized by the resolvent of the harmonic oscillator) for a large class of two-dimensional \emph{magnetic operators} perturbed by (homogeneous) {potentials}. One of the payoffs of this result is the possibility of reinterpreting the {density of states} (DOS) of these perturbed magnetic systems via the  {Dixmier trace}, and taking advantage of the fact that this quantity can be conveniently calculated on the basis of the Laguerrre functions that diagonalize the harmonic oscillator.

 \medskip

\end{abstract}

\maketitle

\tableofcontents

\section{Introduction}\label{sec:Intr0}
The main goal of this work is to provide a description of the \emph{trace per unit volume} in terms of the \emph{Dixmier trace} (regularized by the resolvent of the harmonic oscillator) for a large class of two-dimensional \emph{magnetic operators} perturbed by \emph{potential}. The results exposed in this work extend the achievements obtained in \cite{belmonte-denittis-22,denittis-gomi-moscolari-19,denittis-sandoval-00,denittis-sandoval-21} for the case of \emph{unperturbed} magnetic operators, and are complementary to similar results derived in \cite{Bellissard-03,bellissard-elst-schulz-baldes-94} 
for the case of discrete operators, and in \cite{azamov-mcdonald-sukochev-zanin-19} for the case of the free Laplacian perturbed by potentials. Since the trace per unit of volume is the crucial ingredient for the definition of the \emph{integrate density of states} (IDOS), the results obtained in this work supply a solid background for the use of the Dixmier trace in the spectral analysis of perturbed magnetic operators.

\medskip

The first ingredient which enter in the description of our main results is the trace per unit of volume. Let  $\Lambda_{n} \subseteq \R^d$ be
an
 increasing sequence of compact subsets such that $\Lambda_n\nearrow\R^d$ and which satisfies the \emph{F{\o}lner condition} (see \cite{greenleaf-69} for more details). 
 Let $\chi_{\Lambda_n}$ be the projection defined as the multiplication operator by the characteristic function of $\Lambda_n$.
 A bounded operator $S$ acting on $L^2(\R^d)$ admits the trace per unit  volume (with respect  to the  F{\o}lner sequence $\Lambda_n$) if the limit
 \begin{equation}\label{eq:TUV}
\s{T}_{\rm u.v.}(S)\;:=\;\lim_{n\to+\infty} \frac{1}{|\Lambda_n|}{\Tr}_{L^2(\R^2)}( \chi_{\Lambda_n}S \chi_{\Lambda_n} )\;
 \end{equation}
 exists. The operators that admit a trace per unit volume independently of the peculiar election of the F{\o}lner sequence are particularly relevant in physics since they represent observables with \virg{good} thermodynamic properties. In the presence of perturbations by potentials 
which are \emph{homogeneous} with respect to the spatial translations one is forced to replace a single operator $S$ with a \emph{covariant} family of (measurable) operators $S:=\{S_\omega\}_{\omega\in\Omega}$ where the \emph{configuration space} $\Omega$ is a nice topological space endowed with an $\R^d$-action and an ergodic probability measure $\n{P}$ (see Section \ref{Sec:pot}). In this case, the
correct formula for the trace per unit  volume is given by
\begin{equation}\label{eq:tr_un_vol}
 \s{T}_{\rm u.v.}(S)\;:=\;\n{E}\left[ 
\s{T}_{\rm u.v.}(S_\omega)\right]
\end{equation}
where, following  a consolidated tradition, one uses the notation
\[
\n{E}[f]\;:=\; \int_\Omega\dd\n{P}(\omega)\; f(\omega)
\]
for the \emph{expectation} (probabilistic average)
of the measurable function $f:\Omega\to\C$. The trace per unit volume is the central tool for the definition of the IDOS of a self-adjoint (bounded from below) operator $H$.
Let $\chi_{(-\infty,\epsilon]}(H)$ be the spectral projection of $H$ associated with the spectral interval  $(-\infty,\epsilon]$. Then the IDOS of $H$ is the function $N_H:\R\to [0,+\infty]$  defined as
\begin{equation}\label{eq:def_idos}
N_H(\epsilon)\;:=\; \s{T}_{\rm u.v.}\left(\chi_{(-\infty,\epsilon]}(H)\right)\;.
\end{equation}
The \emph{density of states} (DOS)  of $H$ is the
 \emph{Lebesgue-Stieltjes measure} $\mu_H$ associated to $N_H$. There is an extremely extensive literature devoted to the study of the existence and regularity of the IDOS and the DOS of a given operator $H$. The interested reader is referred to the classic monographs \cite{carmona-lacroix-90,pastur-figotin-92}, or the more recent book \cite{veselic-08}. It is also worth mentioning that the trace per unit volume enters  in a crucial way in the \emph{Kubo  formula} which is one of the main result for the study of  transport phenomena in the
linear response regime \cite{bellissard-schulz-baldes-98,elgart-schlein-04,bouclet-germinet-klein-schenker-05,denittis-lein-book,henheik-teufel-21}. 

 \medskip
 
 The second ingredient we need is the Dixmier trace.
 This trace was invented by Dixmier as an example of a non-normal
 trace \cite{dixmier-66}.
 There are several standard references
 devoted to the theory of the Dixmier trace \cite{connes-94,connes-moscovici-95,gracia-varilly-figueroa-01,lord-sukochev-zanin-12,alberti-matthes-02}, but for the unfamiliar reader we will provide here a brief summary. Let $\mu_n(T)$  be the sequence of the \emph{singular values} of a compact operator $T$, \ie    the eigenvalues of  $|T|:=\sqrt{T^*T}$,
 listed in decreasing order and repeated according to their multiplicity. Consider the  new sequence
 \begin{equation}\label{eq:partial_gamma}
\gamma_N(T)\;:=\;\frac{1}{\log(N)}\sum_{n=0}^{N-1}\mu_n(T)\;,\qquad N>1\;.
\end{equation}
 Then
  $T$ is in the \emph{Dixmier ideal} $\rr{S}^{1^+}$ if its \emph{(Calder\'on)} norm
\begin{equation}\label{eq:clad_norm}
\lVert T\rVert_{1^+}\;:=\;\sup_{N>1}\ \gamma_N(T)\;<\;+\infty
\end{equation}
is finite. $\rr{S}^{1^+}$ is a two-sided self-adjoint ideal that is closed with respect to the norm~\eqref{eq:clad_norm}, but not
with respect to the operator norm. Therefore, every  $T\in\rr{S}^{1^+}$ defines 
by means of of \eqref{eq:partial_gamma} a sequence $\gamma_N(T)$ in $\ell^\infty(\N)$.
To define a trace functional with domain the  ideal $\rr{S}^{1^+}$ one needs to choose a \emph{generalized scale-invariant limit} ${\rm Lim}: \ell^{\infty}(\mathbb{N}) \to \mathbb{C}$ and the associated {Dixmier trace} for a positive element  is defined as
\begin{equation}\label{eq:dix_norm-00}
  {\Tr}_{{\rm Dix},{\rm Lim}}(T)\;: =\; {\rm Lim}\big[ \{ \gamma_{N}(T)\}_{N} \big]
  \;,\qquad T \in \rr{S}^{1^+}\;,\;\; T\geqslant0\;.
\end{equation}
This definition extends to non-positive elements of  $\rr{S}^{1^+}$ by linearity. The
resulting trace
is continuous with respect to the  norm~\eqref{eq:clad_norm}, \ie
$|{\Tr}_{{\rm Dix},{\rm Lim}}(T)|\leqslant \|T\|_{1^+}$. 
 {An element $T\in \rr{S}^{1^+}$ is called \emph{measurable} if  
the value of \eqref{eq:dix_norm-00} is independent of the
  choice of the generalized scale-invariant limit ${\rm Lim}$. For a positive element $T\geqslant 0$ this is equivalent to the convergence of a certain Ces\`{a}ro mean of $\gamma_N(T)$ \cite[Chap.~4, Sect.~2, Proposition 6]{connes-94}. In particular, for a $T\geqslant 0$ such that $\gamma_N(T)$ is convergent, one has that $T$ is measurable and 
  \begin{equation}\label{eq:recip_Dix_Tr}
  {\Tr}_{{\rm Dix}}(T)\;: =\;
  \lim_{N\to\infty}\left(\frac{1}{\log(N)}\sum_{n=0}^{N-1}\mu_n(T)\right)\;,
  \end{equation}
independently of the election of the generalized scale-invariant limit.}
The set of measurable operators $\rr{S}^{1^+}_{\rm m}$ is
 a closed subspace of
$\rr{S}^{1^+}$ (but not an ideal) which is invariant under conjugation by bounded invertible operators. This is the class of operators in which will be mainly interested. It is worth mentioning at the end of this short presentation that the Dixmier trace enters as a crucial ingredient for the construction of the \emph{quantize calculus} and the definition of the \emph{Chern character} in Connes' \emph{non-commutative geometry} \cite{connes-94,connes-moscovici-95,gracia-varilly-figueroa-01}.
 
\medskip 
 
It is now the moment to introduce the systems to which our results apply. 
In the Hilbert space $L^2(\R^2)$,  let $\{\psi_{n,m}\}\subset L^2(\R^2)$, with
$n,m\in\N_0:=\{0\}\cup\N$, be the orthonormal basis given by the Laguerre functions \eqref{eq:lag_pol}.
Let us introduce the family $\{\Upsilon_{j\mapsto k}\;|\, (j,k)\in \N_0^2\}$ of \emph{transition operators}   defined by 
\begin{equation}\label{eq:intro:basic_op}
\Upsilon_{j\mapsto k}\psi_{n,m}\;:=\;\delta_{j,n}\;\psi_{k,m}\;,\qquad k,j,n,m\in\N_0\;.
\end{equation}
A direct computation shows that  \cite[Proposition 2.10]{denittis-sandoval-00}
\begin{equation}\label{eq:rel_alg}
(\Upsilon_{j\mapsto k})^*\;=\;\Upsilon_{k\mapsto j}\;,\qquad \Upsilon_{j\mapsto k}\Upsilon_{m\mapsto n}\;=\;\delta_{j,n}\Upsilon_{m\mapsto k}
\end{equation}
for every $j,k,n,m\in \N_0^2$. 
In view of the relations \eqref{eq:rel_alg} the transition operators 
can be chosen as generators of
 a $C^*$-algebra. Let 
\begin{equation}\label{eq:equal_C-Ups}
\bb{C}_B\;=\;C^*(\{\Upsilon_{j\mapsto k}\;|\;k,j\in\N_0\})
\end{equation}
 be the (non-unital) $C^*$-algebra generated inside the bounded operators
$\bb{B}(L^2(\R^2))$ by  the norm closure of polynomials in the generators $\Upsilon_{j\mapsto k}$. 
We will refer to  $\bb{C}_B$ as the $C^*$-algebra of \emph{(unperturbed) magnetic operators}. Such a name is justified by the fact that the Landau Hamiltonian
$H_{B}$ (for a \emph{magnetic field} $B>0$) defined in \eqref{eq:intro_LH} is \emph{affiliated} with $\bb{C}_B$. More precisely, it turns out that the \emph{Landau projections} $\Pi_j$ (the spectral projections of $H_B$) are contained in $\bb{C}_B$ in view of the equality $\Pi_j=\Upsilon_{j\mapsto j}$. The algebra $\bb{C}_B$ has been studied extensively in \cite{denittis-sandoval-00,denittis-sandoval-21} and Section \ref{sec:unp_magn_alg} provides
 a brief presentation of some important aspects concerning this algebra and its enveloping von Neumann algebra $\bb{M}_B$.

\medskip

The algebra of magnetic operators provides information only on the unperturbed dynamics in the presence of a magnetic field $B$. In order to also consider the effect of homogeneous perturbations (periodic, random, \etc) one has to enlarge the algebra $\bb{C}_B$. This can be done by considering 
an ergodic dynamical system $(\Omega,\R^2,\rr{t},\n{P})$, 
 which encodes the information about the possible configuration of the perturbations and their transformation under the spatial translations implemented by the $\R^2$-action $\rr{t}$. As described in detail in  Section \ref{Sec:pot}, to each continuous function $g\in C(\Omega)=:\bb{A}_\Omega$ and any point $\omega\in\Omega$, one can define the multiplication operator 
  $M_{g,\omega}\in \bb{B}(L^2(\R^2))$  given by
$$
\big(M_{g,\omega}\varphi\big)(x)\;:=\;g\big(\rr{t}_{x}(\omega)\big)\varphi(x)\;,\qquad \varphi\in L^2(\R^2)\;.
$$ 
Being a multiplication operator,  $M_{g,\omega}$ can be interpreted as a potential when $g$ is real-valued. We will focus on the products of the type $A M_{g,\omega}$ with $A\in\bb{C}_B$ and $g\in \bb{A}_\Omega$. Let us denote with $(\bb{C}_B\cdot\bb{A}_\Omega)_\omega$ the linear space of finite linear combinations of such products and with 
$[[\bb{C}_B\cdot\bb{A}_\Omega]]_\omega$ its norm-closure inside 
 $\bb{B}(L^2(\R^2))$. One is tempted to see $[[\bb{C}_B\cdot\bb{A}_\Omega]]_\omega$ as a subspace of the minimal  $C^*$-algebra $\rr{A}_{B,\omega}$ generated by $\bb{C}_B$ and the potentials associated to  
  $\bb{A}_\Omega$. However, this is unnecessary, in view of the following result:
   \begin{theorem}\label{teo:ident1}
 For every $\omega\in\Omega$ the space $[[\bb{C}_B\cdot\bb{A}_\Omega]]_\omega$ is a $C^*$-algebra, and in particular
\begin{equation}\label{eq:caract1}
 [[\bb{C}_B\cdot\bb{A}_\Omega]]_\omega\;=\;\rr{A}_{B,\omega}\;.
\end{equation}
 Moreover, $\rr{A}_{B,\omega}$ coincides with a representation (labelled by $\omega$) of the twisted crossed product $\bb{A}_\Omega\rtimes_B\R^2$\;.
 \end{theorem}
  \medskip
  
  \noindent
  Theorem \ref{teo:ident1} is the first main result of this work. 
  Although at first glance it may appear somewhat surprising, it
  has the same taste as other results present in the literature like \cite[Theorem 1.1]{georgescu-iftimovici-02} and \cite[Lemma 3.2]{georgescu-iftimovici-06} (but with a totally distinct strategy of the proof).
As shown in Section \ref{sec:pert_magn},
 the $C^*$-algebra $\rr{A}_{B,\omega}$ contains the resolvents of the magnetic operators perturbed by potentials with the reference point, or better the \emph{origin}, fixed by $\omega$. For this reason, we will refer to $\rr{A}_{B,\omega}$ as the algebra of perturbed magnetic operators.
However, to remove the unnecessary and unphysical dependence on the choice of a particular origin $\omega$ it is customary to consider instead of a single operator $S_\omega\in \rr{A}_{B,\omega}$ a complete family $S:=\{S_\omega\}_{\omega\in\Omega}$ whose elements are subject to a suitable covariance relation under translations (see eq. \eqref{eq:cov_cond_V}). This leads us to consider the 
 \emph{full} algebra of perturbed magnetic operators
\begin{equation}\label{eq:int_rep_1}
\rr{A}_{B,\Omega}\;:=\;\int^\oplus_{\Omega}\dd\n{P}(\omega)\;\rr{A}_{B,\omega}\;.
\end{equation}
 acting on the \emph{direct integral} Hilbert space \cite[Part II, Chapter 1]{dixmier-81}
 \begin{equation}\label{eq:dir_int<_hil}
 \s{H}_\Omega\;:=\;\int^\oplus_{\Omega}\dd\n{P}(\omega)\;L^2(\R^2)\;\simeq\;L^2(\Omega,\n{P})\otimes L^2(\R^2)\;.
 \end{equation}
The $C^*$-algebra $\rr{A}_{B,\Omega}$, along with its 
enveloping von Neumann algebra $\rr{M}_{B,\Omega}$, are the right objects to study the properties of magnetic systems perturbed by homogeneous potentials.  
This is perfectly consistent with the existing literature. In fact  Proposition \ref{prop:int_rep_pair} shows that 
  $\rr{A}_{B,\Omega}$ is a  faithful and non-degenerate
  (left-regular) representation of the  \emph{twisted crossed product} $\bb{A}_\Omega\rtimes_B\R^2$ associated to the dynamical system $(\Omega,\R^2,\rr{t},\n{P})$  \cite{busby-smith-70,pedersen-79,packer-raeburn-89,williams-07}. 
  The $C^*$-algebra $\bb{A}_\Omega\rtimes_B\R^2$ plays a privileged and relevant role in the study of topological properties of condensed matter systems since the seminal paper \cite{Bellissard-93},  where it was named the \emph{noncommutative Brillouin zone}.
The identification of   $\rr{M}_{B,\Omega}$ as the von Neumann algebra associated to the 
  faithful 
  (left-regular) representation of $\bb{A}_\Omega\rtimes_B\R^2$ has a further important consequence. In fact it is known that $\rr{M}_{B,\Omega}$ can be endowed with a unique faithful and semi-finite norma trace $\tau_{\n{P}}$ defined on a dense ideal $\rr{I}_{B,\Omega}\subset \rr{M}_{B,\Omega}$. The proof of this fact (see Theorem \ref{theo_trac}) presented here follows very closely the construction in \cite{lenz-99}, and relies on the \emph{Hilbert algebra} structure underlying the $C^*$-algebra 
  $\rr{A}_{B,\Omega}$ as explained in Section \ref{sec:tr_u_v} and Appendix \ref{sec_Hilb_str}. 
 This trace is indeed closely related with the trace per unit volume \eqref{eq:tr_un_vol}. As discussed in Appendix \ref{sec:tra_UV} any $S\in\rr{I}_{B,\Omega}$ admits trace per unit volume 
  (independent of the  election of the F{\o}lner sequence) and the following equality holds true
\begin{equation}\label{eq:intro_001}
\tau_{\n{P}}(S)\;=\; 2\Lambda_B\;\s{T}_{\rm u.v.}(S)\;, 
\end{equation}
where $\Lambda_B:=\pi\ell^2$  has the physical meaning of the area of the \emph{magnetic disk} of radius $\ell$, and $\ell={B}^{-\frac{1}{2}}$ when all the physical constants are set equal to 1.
  
\medskip

We still need a small effort to introduce our second main result.  For $q\in\N$ consider 
  the space $\ell^q(\N_0^2, \bb{A}_\Omega)$ of $q$-summable sequences, \ie $\{g_{n,m}\}\subset\bb{A}_\Omega$
 if and only if $\sum_{(n,m)\in \N_0^2}\|g_{n,m}\|^q_\infty<\infty$.
  Associated to this let us introduce the (formal)  space
\begin{equation}\label{eq:exp_op}
\rr{L}^q_{B,\Omega}\;:=\;\left.\left\{S\;=\;\sum_{(n,m)\in\N_0^2}\Upsilon_{n\mapsto m}M_{g_{n,m}}\;\right|\;\{g_{n,m}\}\in \ell^q(\N_0^2, \bb{A}_\Omega)\right\}\;.
\end{equation}
Here  $M_{g_{n,m}}=\{M_{g_{n,m},\omega}\}_{\omega\in\Omega}$ is the bounded operator on $\bb{B}(\s{H}_\Omega)$ generated by the collection of all the potentials $M_{g_{n,m},\omega}\in\bb{B}(L^2(\R^2))$
associated to $g_{n,m}\in \bb{A}_\Omega$ and $\omega\in\Omega$.  Only  the cases $q=1,2$ are relevant for the aims of this work, and for the needs of this introduction, we focus only on $q=1$. The relevance of $\rr{L}^1_{B,\Omega}$ is that it is norm-dense in $\rr{A}_{B,\Omega}$ and strongly-dense in $\rr{M}_{B,\Omega}$ as consequence of Theorem \ref{teo:ident1}. Moreover, one can prove that $\rr{L}^1_{B,\Omega}\subset \rr{I}_{B,\Omega}$ is contained inside the ideal of definition of the trace $\tau_{\n{P}}$ (see Proposition \ref{prop:_pre_trac_dix}).
 
\medskip

Just as a last bit of information let us introduce the  \emph{harmonic oscillator} $Q$ on $L^2(\R^2)$ described in terms of its  diagonalization on the basis of the Laguerre functions $\psi_{n,m}$ as
\[
Q\psi_{n,m}\;:=\;(n+m+1)\psi_{n,m}\;,\qquad (n,m)\in \N^2_0\;.
\] 
The expression
\begin{equation}\label{eq:op-Q-2}
 Q_{\lambda}^{-s}\;:=\;(Q+\lambda{\bf 1})^{-s}
\end{equation}
  defines a compact operator on $L^2(\R^2)$ for every 
$s>0$ and $\lambda>-1$.

 \medskip
 With all the information provided above, we can present a precise statement of the second main result achieved in this work.
\begin{theorem}\label{theo:main_dix_eq_pot}
Let $S=\{S_\omega\}_{\omega\in\Omega}\in\rr{L}^1_{B,\Omega}$. Then both $Q_{\lambda}^{-1}S_\omega$ and $S_\omega Q_{\lambda}^{-1}$ are measurable operators, \ie elements of  $\rr{S}^{1^+}_{\rm m}$, and
\begin{equation}\label{eq:traXXX_III_XX_z}
\begin{aligned}
\tau_{\n{P}}(S) \;=\; \n{E}\left[{\Tr}_{\rm Dix}\big(Q_{\lambda}^{-1}S_\omega\big)\right]\;&=\;\n{E}\left[{\Tr}_{\rm Dix}\big(S_\omega Q_{\lambda}^{-1}\big)\right]\end{aligned}
\end{equation}
independently of $\lambda>-1$.
\end{theorem}

\medskip

The proof of Theorem \ref{theo:main_dix_eq_pot} is postponed to Section \ref{sect:dixmier_gen_elem} and is based on a series of technical results (Lemmas \ref{lem:densiti_set_trace_0}, \ref{lem:densiti_set_trace_1} and Propositions \ref{lem:densiti_set_trace}, \ref{prop:densiti_set_trace_02}) aimed to prove that the claim holds for the building block operators
$\Upsilon_{j\mapsto k} M_{g,\omega}$. Then, the final result is obtained by a continuity argument in the topology induced by the 
Calder\'on norm \eqref{eq:clad_norm}. It is also important to point out that a crucial ingredient of the proof is the \emph{scaling-limit formula} for certain combinations of the Laguerre polynomials described in Appendix \ref{app_scal_lim}. In particular the content of Lemma \ref{lemm_scal_lim} extends and completes a partial result initialy proved in \cite[Lemma 3]{hupfer-leschke-warzel-01}.

\medskip
 
By combining the content of  Theorem \ref{theo:main_dix_eq_pot}, the equality \eqref{eq:intro_001} and the definition \eqref{eq:def_idos} one infers the possibility of representing the IDOS of a large class of elements in  $\rr{M}_{B,\Omega}$ in terms of the Dixmier trace regularized by the operator 
 $Q_{\lambda}^{-1}$. This observation paves the way to extend the results obtained in \cite{belmonte-denittis-22} for unperturbed magnetic operators to magnetic operators perturbed by homogeneous potentials. In particular, it is expected to obtain an extension of the \emph{residue formula} \cite[Theorem 1.1]{belmonte-denittis-22} and 
 \emph{energy shell formula} \cite[Theorem 1.2]{belmonte-denittis-22} for the computation of the IDOS in the presence of perturbations. There is also a second aspect that deserves further investigation. The equality \eqref{eq:traXXX_III_XX_z} is established only for elements in the dense subspace $\rr{L}^1_{B,\Omega}$, but one is tempted to believe that the equality must occur on the whole domain
of definition $\rr{I}_{B,\Omega}$ of the trace $\tau_{\n{P}}$. However, this question is still unresolved even in the unperturbed case (see \eg \cite[Remark 2.29]{denittis-sandoval-00}).

\medskip

 It is worth spending a few words in this presentation for 
 a comparison with the stimulating paper \cite{azamov-mcdonald-sukochev-zanin-19}. 
Theorem \ref{theo:main_dix_eq_pot} can be interpreted as  the \virg{magnetic version} of \cite[Theorem 1.1]{azamov-mcdonald-sukochev-zanin-19} in the special case $d=2$. Aside from the obvious similarity, there is a relevant difference. The \virg{weight} introduced in \cite[Theorem 1.1]{azamov-mcdonald-sukochev-zanin-19} is the multiplication operator by the function $(1+|x|^2)^{-1}$ which is evidently not compact. 
On the contrary, the \virg{weight} $Q_{\lambda}^{-1}$ used in Theorem \ref{theo:main_dix_eq_pot} is a compact operator which is diagonalized on the basis of the  {generalized Laguerre functions} $\psi_{n,m}$ defined in \eqref{eq:lag_pol}.
The latter fact provides a significant computational advantage. 

\medskip

Finally the compactness of $Q_{\lambda}^{-1}$, along with the content of Theorem \ref{theo:main_dix_eq_pot} has relevance in the study of the topological interpretation of certain magnetic phenomena like the \emph{quantum Hall effect} (QHE). Let us start by claiming that any result that establishes an equality between the trace per unit volume and the Dixmier trace provides at the same time a bridge between the study of certain thermodynamics quantities, like the transport coefficients, and 
the topological invariants inside the Connes' quantized calculus. This is  
 exactly the spirit of the \emph{Kubo-Chern duality} proved in the seminal paper \cite{bellissard-elst-schulz-baldes-94},  which represents the most exhaustive explanation of the topological nature of QHE for discrete (random) systems. The strategy of \cite{bellissard-elst-schulz-baldes-94}
has been repeated in  \cite {denittis-sandoval-00,denittis-sandoval-21} 
obtaining a {Kubo-Chern duality} for continuous and \emph{unperturbed} magnetic systems based on the use of a (magnetic) Dirac operator with \emph{compact} resolvent proportional to $Q_{\lambda}^{-\frac{1}{2}}$.
The compactness of the resolvent of the Dirac operator is a new important ingredient in the construction of a non-commutative geometry for magnetic systems since it permits a natural discretization (and approximations) of the formulas for the calculation of the topological invariants. Theorem \ref{theo:main_dix_eq_pot} paves the way to use the same Dirac operator 
also in the case of magnetic systems perturbed by potentials. This question is currently a matter of investigation.

% \medskip
% 
% \noindent
%{\bf Structure of the paper.}
%In {\bf Section~\ref{sec:BG_mat}} we will introduce the background material about magnetic operators necessary as the starting point for the  formulation of our main results.  
%{\bf Section~\ref{sect:pre-res-form}}
%contains the proof of the 
% {residue-type formula} anticipated in Theorem \ref{theo:1st-RF} and {\bf Section~\ref{sec:shell_F}} contains the  the proof of  the energy shell formula provided in Theorem \ref{theo:2st-ES}.
%{\bf Section~\ref{Sect_DOS}} concerns with the construction 
%of the IDOS and of the DOS and contains the proof of the spectral formula \eqref{eq:spec_form}, and consequently of 
%Theorem \ref{corol:main_dix_eq_03}. 
%Possible generalizations of our main results  along with the related open problems  are briefly discussed in 
%{\bf Section~\ref{sec:pos_ghen}}.
% {\bf Appendix~\ref{sec:ser-von}} contains some useful results about the elements of the magnetic algebra $\bb{M}$
%which are needed for the proof of Lemma \ref{lemm_absorb}.
% Finally, {\bf Appendix~\ref{sec:dix_tr_2}} contains a very short introduction to the Dixmier trace
% in order to make this  work self-contained.
% 
% 
%
%
%

 \medskip
 
 \noindent
{\bf Acknowledgements.}
GD's research is supported by the grant \emph{Fondecyt Regular - 1230032}. This work would not have been possible without the stimulating ideas that J. Bellissard shared with GD way back in 2012 in a hotel of Sendai.
% GD would like to cordially thank Massimo Moscolari  for his help in the proof of Theorem \ref{theo:1st-RF}.

%------------------------------------------------------------------------%
\section{The algebra of unperturbed magnetic operators}
\label{sec:unp_magn_alg}
In this section, we provide some more information about the algebra of 
 unperturbed magnetic operators introduced in Section \ref{sec:Intr0}. A more complete exposition is contained in \cite{denittis-sandoval-00}.
 
 \medskip

The \emph{(dual) magnetic translations}\footnote{The name magnetic translations is common in the condensed matter community since the works of Zak \cite{zak1,zak2}. Mathematically, they are also known as Weyl systems.}
on $L^2(\R^2)$ are the unitary operators 
defined by
\begin{equation}\label{eq:intro_000}
  \left(V(a)\psi\right)(x)\; =\; \expo{\ii\frac{x\wedge a}{2\ell^{2}}}\;\psi(x-a)\;, \qquad a \in \R^2\;,\quad \psi \in L^2(\R^2)\;
 \end{equation}
where $x\wedge a :=x_{1}a_{2} -
    x_{2}a_{1}$, for all $x=(x_1,x_2)$ and $a=(a_1,a_2)$.
A direct 
  computation shows that 
  \[
  \begin{aligned}
    V(a)V(b)\;&=\; \expo{\ii\frac{b\wedge a}{2\ell^{2}}}\; V(a+b)
  \end{aligned}
  \;,\qquad a,b \in \R^2\;.
\]
The parameter $\ell>0$ will be interpreted as the \emph{magnetic length} of the system and physically  it is proportional to $B^{-\frac{1}{2}}$, where $B>0$ is the strength of a constant magnetic field perpendicular to the plane $\R^2$. Therefore, $\ell \to \infty$ represents the \emph{singular} limit of a vanishing magnetic field 
\cite[Remark 2.2]{denittis-sandoval-00}.
Let $\bb{V}:=C^*(\{V(a)\;|\;a\in\R^2\})$ be the $C^*$-algebra generated by the unitaries $V(a)$. The \emph{magnetic von Neumann algebra} $\bb{M}_B$ (or magnetic algebra, for short) is by definition the {commutant} of $\bb{V}$  \cite[Proposition 2.18]{denittis-sandoval-00} \ie,
\begin{equation}\label{int_VNA}
\bb{M}_B\;:=\;\bb{V}'\;=\;\{A\in\bb{B}(L^2(\R^2))\;|\; AB-BA=0\;,\;\;\forall\; B\in\bb{V} \}\:.
\end{equation}
The name magnetic algebra is justified by the fact the \emph{Landau Hamiltonian} 
\begin{equation}\label{eq:intro_LH}
  H_{\rm L}\;: =\; \frac{1}{2} \left(-\ii \ell\frac{\partial}{\partial x_{1}}  - \frac{1}{2 \ell} x_{2}\right)^2\;+\;
  \frac{1}{2} \left(-\ii \ell\frac{\partial}{\partial x_{2}}  + \frac{1}{2 \ell} x_{1}\right)^2\;,
\end{equation}
is affiliated to $\bb{M}_B$ (\cf \cite[Section 2.3]{denittis-sandoval-00}).

\medskip

Consider  the Hilbert space $L^2(\R^2)$ and let $\{\psi_{n,m}\}\subset L^2(\R^2)$, with
$n,m\in\N_0$, be the orthonormal 
\emph{Laguerre
  basis} defined by 
\begin{equation}\label{eq:lag_pol}
\psi_{n,m}(x)\;:=\;\psi_0(x)\ \sqrt{\frac{n!}{m!}}\left[\frac{x_1-\ii x_2}{\ell\sqrt{2}}\right]^{m-n}L_{n}^{(m-n)}\left(\frac{|x|^2}{2\ell^2}\right)\; ,
\end{equation}
where
\begin{equation}\label{eq:lag_pol-00}
  L_n^{(\alpha)}\left(\xi\right)\;:=\;\sum_{j=0}^{n}\frac{(\alpha+n)(\alpha+n-1)\ldots(\alpha+j+1)}{j!(n-j)!}\left(-\xi\right)^j\;,\quad\alpha,\xi\in \R
\end{equation}
are the {generalized Laguerre polynomial} of degree $n$ (with the usual convention $0!=1$) and 
\begin{equation}\label{eq:herm1}
\psi_{0,0}(x)\;=\;\psi_{0}(x)\;:=\;\frac{1}{\sqrt{2\pi}\ell}\ \expo{-\frac{|x|^2}{4\ell^2}}\;.
\end{equation}

\medskip

The  magnetic  algebra $\bb{M}_B$ defined by \eqref{int_VNA} coincides with the enveloping von Neumann algebra of the $C^*$-algebra $\bb{C}_B$ defined by \eqref{eq:equal_C-Ups}, \ie
 $
\bb{M}_B=\bb{C}_B''
$ 
where on the right-hand side one has the bicommutant  of  $\bb{C}_B$ \cite[Proposition 2.18]{denittis-sandoval-00}. Therefore,  $\bb{M}_B$ is generated by the weak limits of the polynomials of the {transition operators}   defined by 
\eqref{eq:intro:basic_op}.

 %------------------------------------------------%
 \section{The algebra of potentials}\label{Sec:pot}
This section is devoted to the description of perturbative potentials from an algebraic perspective.  Following a commonly accepted point of view, the perturbations by \emph{homogeneous} potentials are introduced through an \emph{ergodic dynamical system} given by the quadruple $(\Omega,\R^2,\rr{t},\n{P})$, where:
\begin{itemize}
\item $\Omega\neq\emptyset$ is a compact and metrizable (hence separable) space\footnote{With this requirements $\Omega$ turns out to be a \emph{Polish} space and the pair $(\Omega,\mathtt{Bor}(\Omega))$ is a \emph{standard} Borel space.} 
  endowed with its  Borel $\sigma$-algebra $\mathtt{Bor}(\Omega)$;
  \vspace{1mm}
\item
  $\mathbb{P}$ is a (complete) probability Borel measure\footnote{\label{note010}With this assumption the triple $(\Omega,\mathtt{Bor}(\Omega),\n{P})$ is a \emph{standard probability} space.
 From  \cite[Th{\'e}or{\`e}me 4-3]{delarue-93} it turns out  that the associated Hilbert space $L^2(\Omega,\n{P})$ is separable.}, namely  $\mathbb{P}(\Omega)=1$;
  \vspace{1mm}
\item
  The support of $\mathbb{P}$ coincides with the whole space $\Omega$\footnote{This means that any open set $\Sigma\subset\Omega$ has a positive measure $\n{P}(\Sigma)>0$.};
  \vspace{1mm}
\item
  $\rr{t}:\R^2\to\mathrm{Homeo}(\Omega)$ is a representation of the group $\R^2$ by homeomorphisms of the space $\Omega$ such that the group-action $\rr{t}:\R^2\times\Omega\to\Omega$ given by $(a,\omega)\mapsto\rr{t}_a(\omega)$ is jointly continuous;
  \vspace{1mm}
\item The measure $\n{P}$ is invariant by $\rr{t}$, \ie $\n{P}(\Sigma)=\n{P}(\rr{t}_a(\Sigma))$ for all $a\in\R^2$ and $\Sigma\in\mathtt{Bor}(\Omega)$;
  \vspace{1mm}
\item The measure $\n{P}$ is ergodic, \ie if $\Sigma\in\mathtt{Bor}(\Omega)$ satisfies $\rr{t}_a(\Sigma)=\Sigma$ for all $a\in\R^2$, then
  $\n{P}(\Sigma)=1$ or $\n{P}(\Sigma)=0$.
 \end{itemize}

\medskip

In the following $\Omega$ will be called the \emph{hull of potentials.}
Let 
 $\bb{A}_\Omega:={C}(\Omega)$ be the   $C^\ast$-algebra of the continuous complex-valued functions on $\Omega$ 
endowed with the usual norm of the uniform convergence $\|\ \|_\infty$. This is a commutative $C^\ast$-algebra which is also separable since $\Omega$ is metrizable \cite[Theorem 2.4]{chou-12}. 
The fact that the group action $(a,\omega)\mapsto\rr{t}_a(\omega)$ is \emph{jointly} continuous implies the following relevant fact:
\begin{lemma}[{\cite[Lemma 2.5]{williams-07}}]\label{lemma:automorph}
For each $a\in\R^2$ the map $T_a:\bb{A}_\Omega\to\bb{A}_\Omega$ defined by
$T_a(g)(\omega):=g(\rr{t}_{-a}(\omega))$, for all  $g\in\bb{A}_\Omega$, is an automorphism of  $\bb{A}_\Omega$. 
Moreover, the map $T:\R^2\to\text{\upshape Aut}(\bb{A}_\Omega)$, which assigns to each vector $a\in\R^2$ the related automorphism $T_a$, defines an action of $\R^2$ on $\bb{A}_\Omega$ which is strongly continuous. 
\end{lemma}
According to the standard terminology, the triple  $(\bb{A}_\Omega,\R^2,T)$ defines a  \emph{separable $C^\ast$-dynamical system} over the hull $\Omega$  \cite[Sect. 7.4]{pedersen-79}. The action of $\R^2$ endows $\bb{A}_\Omega$  also of a differential structure and one can define a dense subalgebra $\text{Diff}(\bb{A}_\Omega)\subset \bb{A}_\Omega$ made of 
\emph{differentiable} elements. This fact is discussed in detail in Appendix \ref{sec_dif_pot} and is clarified here for possible applications. 
 
 \medskip
 
Let us now introduce a family of representations of $\bb{A}_\Omega$ on the Hilbert space $L^2(\R^2)$. For every $\omega\in\Omega$, let $\pi_\omega:\bb{A}_\Omega\to \bb{B}(L^2(\R^2))$ be the map defined by $\pi_\omega:g\mapsto M_{g,\omega}$ where  $M_{g,\omega}$ is the multiplication operator given by
$$
\big(M_{g,\omega}\varphi\big)(x)\;:=\;g\big(\rr{t}_{x}(\omega)\big)\varphi(x)\;,\qquad \varphi\in L^2(\R^2)\;.
$$ 
The map  $\pi_\omega$ is evidently a homomorphism of $C^*$-algebras, \ie a representation. 

\medskip

From its very definition one hast that
$$
\|\pi_\omega(g)\|\;=\;\|M_{g,\omega}\|\;=\;\sup_{x\in{\R^2}}\big|g\big(\rr{t}_{x}(\omega)\big)\big|
\;=\;\sup_{\omega'\in{\rm Orb}(\omega)}\big|g\big(\omega'\big)\big|\;\leqslant\;\|g\|_\infty
$$
where ${\rm Orb}(\omega):=\{\rr{t}_{x}(\omega)\;|\; x\in\R^2\}$ is the \emph{orbit} of $\omega$. 

\medskip

Let us recall that a  point $\omega\in\Omega$ it is called \emph{transitive} if it has dense orbit, namely if
$\overline{{\rm Orb}(\omega)}=\Omega$. The  the dynamical system $(\Omega,\R^2, \rr{t})$ is called \emph{minimal} if every point $\omega\in\Omega$ is transitive.
\begin{proposition}\label{prop:faith1}
If $\omega\in\Omega$ is transitive then representation $\pi_\omega$ is faithful. If $(\Omega,\R^2, \rr{t})$ is {minimal}, then every representation $\pi_\omega$ is faithful.
\end{proposition}
\proof
Let us recall that for a $C^*$-algebra a representation is faithful if and only if it is isometric \cite[Proposition 2.3.3]{bratteli-robinson-87}. 
Since
\[
\|\pi_\omega(g)\|\;=\;\sup_{\omega'\in{\rm Orb}(\omega)}\big|g\big(\omega'\big)\big|\;=\;\sup_{\omega'\in\overline{{\rm Orb}(\omega)}}\big|g\big(\omega'\big)\big|
\]
in view of the continuity of $g$, one infers that the transitivity of $\omega$ implies the isometry of $\pi_\omega$. The claim about the minimality follows immediately.
\qed

\medskip

The $C^*$-algebra
$$
\pi_\omega(\bb{A}_\Omega)\;\subset\;\bb{B}\big(L^2(\R^2)\big)
$$
 will be called the \emph{algbera of $\omega$-potentials}. Let $M_{g,\omega}\in \pi_\omega(\bb{A}_\Omega)$ and $V(a)$ the magnetic translation   \eqref{eq:intro_000}. The one can check that
 $V(a)^*M_{g,\omega}V(a)=M_{g,\rr{t}_a(\omega)}$, or
 $$
 V(a)^*\;\pi_\omega(\bb{A}_\Omega)\;V(a)\;=\;
 \pi_{\rr{t}_a(\omega)}(\bb{A}_\Omega)\;,\qquad \forall\; a\in\R^2\;, \quad
 \forall\;\omega\in\Omega\;.
 $$
 This formula is known as \emph{covariance relation}.

 \medskip

 Let us consider the direct integral~
 $ \s{H}_\Omega$
defined by equation \eqref{eq:dir_int<_hil}.
 In view of Note \ref{note010} it follows that $\s{H}_\Omega$ is separable.
The algebra of bounded operators on $\s{H}_\Omega$ will be denoted with $\bb{B}(\s{H}_\Omega)$. However, 
in order to benefit from the direct integral structure of $\s{H}_\Omega$, it is  more convenient to focus on the sub-algebra of  \emph{decomposable} operators $\bb{D}(\s{H}_\Omega)\subset \bb{B}(\s{H}_\Omega)$~\cite[Part II, Chapter 2]{dixmier-81}. An element $A\in\bb{D}(\s{H}_\Omega)$ can be thought of as a family $A:=\{A_\omega\}_{\omega\in\Omega}$ so that:
\begin{itemize}
\item[i)] $A_\omega\in \bb{B}(L^2(\R^2))$ for almost all $\omega\in\Omega$ and the maps 
$$
\Omega\;\ni\;\omega\;\longmapsto\; \langle\phi, A_\omega\varphi \rangle_{L^2}\;\in\;\C
$$ are measurable for every choice of $\phi,\varphi\in L^2(\R^2)$;  
  \vspace{1mm}
\item[ii)] the family $A:=\{A_\omega\}_{\omega\in\Omega}$ is essentially bounded in the sense
  that 
  $$
  \esssup_{\omega\in\Omega}\lVert A_\omega\rVert\;<\;\infty\;.
  $$
\end{itemize}
The set $\bb{D}(\s{H}_\Omega)$ is a von Neumann sub-algebra of $\bb{B}(\s{H}_\Omega)$ and the norm of an element $A\in\bb{D}(\s{H}_\Omega)$ is given by
\begin{equation}\label{eq:rand_op_norm_direct}
  \lVert A \rVert\;:=\; \esssup_{\omega\in\Omega}\lVert A_\omega\rVert\;.
\end{equation}
Decomposable operators act on vectors of  $\s{H}_\Omega$ fiberwise. More precisely consider
 a decomposable operator $A:=\{A_\omega\}_{\omega\in\Omega}$ and  a vector
$\hat{\varphi}:=\{\varphi_\omega\}_{\omega\in\Omega}\in\s{H}_\Omega$, then
\begin{equation}\label{eq:rand_op}
  A\hat{\varphi}\;=\;\{(A\hat{\varphi})_\omega\}_{\omega\in\Omega}\;=\;\{A_\omega\varphi_\omega\}_{\omega\in\Omega}\;,
\end{equation}
and this justifies the notation
$$
A\;=\;\int^\oplus_{\Omega}\dd\n{P}(\omega)\;A_\omega\;.
$$
Let $A,B\in \bb{D}(\s{H}_\Omega)$ be two decomposable operators.
It is worth recalling that the condition $A=B$ is equivalent to $A_\omega=B_\omega$ for almost all $\omega\in\Omega$. 

\medskip

Let $g\in \bb{A}_\Omega$ and consider the family $M_g:=\{\pi_\omega(g)\}_{\omega\in\Omega}$. It follows that $M_g$ is a decomposable operator over  $\s{H}_\Omega$ such that $\|M_g\|\leqslant \|g\|_\infty$. Moreover the map $\pi:\bb{A}_\Omega\to\bb{D}(\s{H}_\Omega)$ defined by
$$
\pi(g)\;:=\;M_g\;=\;\int^\oplus_{\Omega}\dd\n{P}(\omega)\;\pi_\omega(g)\;=\;\int^\oplus_{\Omega}\dd\n{P}(\omega)\;M_{g,\omega}
$$
provides a representation of $\bb{A}_\Omega$ over $\s{H}_\Omega$ \cite[Section 8.1]{dixmier-77}.
\begin{proposition}\label{prop:rep_pair}
The representation $\pi$  is faithful and non-degenerate.
\end{proposition}
\proof
Assume that $M_g=0$. Therefore, $M_{g,\omega}=0$  for $\n{P}$-almost all $\omega$ and in turn
\[
 \|M_{g,\omega}\|
\;=\;\sup_{\omega'\in{\rm Orb}(\omega)}\big|g\big(\omega'\big)\big|\;=\;0
\]
Therefore, the restriction of $g$ to ${\rm Orb}(\omega)$ vanishes   for $\n{P}$-almost all $\omega$. However, the continuity of $g$ implies that if $g(\omega_0)\neq 0$, then there is an open neighborhood $\Sigma$ of $\omega_0$ such that $g(\omega)\neq 0$, for every $\omega\in \Sigma$. This means that the restriction of $g$ to ${\rm Orb}(\omega)$ does not vanish identically for every $\omega\in \Sigma$. Since $\n{P}(\Sigma)>0$, we get a contradiction. Hence $g=0$ on $\Omega$, \ie $\pi$ is faithful. Moreover, $C(\Omega)$ is unital, thus $\pi$ is obviously non-degenerate.
\qed 

\medskip

We will refer to 
\begin{equation}\label{eq:full_pot}
\bb{P}_\Omega\;:=\;\pi(\bb{A}_\Omega)\;\subset\;\bb{D}(\s{H}_\Omega)\;\subset\;\bb{B}\big(\s{H}_\Omega\big)
\end{equation}
as
the \emph{(full) algebra of potentials}.  

\medskip

Let us end this section with an ingredient that will be useful in Sections \ref{sec:cross_prod} and \ref{sec:tr_u_v}. By mimicking equation \eqref{eq:intro_000} we can introduce the \emph{(direct) magnetic translations} given by 
\begin{equation}\label{eq;dir_MT}
  \left(U(a)\psi\right)(x)\; =\; \expo{-\ii\frac{x\wedge a}{2\ell^{2}}}\;\psi(x-a)\;, \qquad a \in \R^2\;,\quad \psi \in L^2(\R^2)\;.
 \end{equation}
A direct computation shows that 
 \[
  \begin{aligned}
    U(a)U(b)\;&=\; \expo{\ii\frac{a\wedge b}{2\ell^{2}}}\; U(a+b)
  \end{aligned}
  \;,\qquad a,b \in \R^2\;
\]
Moreover, one has that  $U(a)V(b)=V(b)U(a)$ for every $a,b\in\R^2$, namely the $U(a)$'s and the $V(b)$'s form two families of mutually commuting unitary operators. 
The operators $U(a)$'s can be lifted to operators on
the direct integral $\s{H}_\Omega$. For every $a\in\R^2$ let $\rr{U}(a)$ be the
operator defined by
\begin{equation}\label{eq:magn_tral_dir}
 \rr{U}(a)\hat{\varphi}\;=\;\{(\rr{U}(a)\hat{\varphi})_\omega\}_{\omega\in\Omega}\;=\;\{U(a)\varphi_{\omega}\}_{\omega\in\Omega}
\end{equation}
for every  $\hat{\varphi}:=\{\varphi_\omega\}_{\omega\in\Omega}\in\s{H}_\Omega$. From the definition,
it results that the operators $\rr{U}(a)$'s   preserve the fibers of
$\s{H}_\Omega$, \ie
they are decomposable. It is a matter of a direct check to prove that
 the $\rr{U}(a)$'s are
still unitary and meet  the composition rule
\begin{equation}\label{eq:magn_tral2_dir}
 \rr{U}(a)\rr{U}(b)\;=\;\expo{\ii\frac{a\wedge b}{2\ell^2}}\;\rr{U}(a+b)\;,\qquad\quad a,b\in\R^2\,.
\end{equation}
In other words, they provide a strongly continuous, projective unitary representation of
the group $\R^2$ on the Hilbert space $\s{H}_\Omega$. 

\medskip

If $g\in\bb{A}_\Omega$ then direct check shows that
\[
U(a)^*\;\pi_\omega(g)\;U(a)\;=\;
 \pi_{\rr{t}_a(\omega)}(g)\;,\qquad \forall\; a\in\R^2\;, \quad
 \forall\omega\in\Omega\;.
\] 
This means that also the  ${U}(a)$'s provide a covariance relation for the elements of $\bb{A}_\Omega$. 
Since $\pi_{\rr{t}_a(\omega)}(g)=\pi_\omega(T_a g)$, 
with  $T_a:\bb{A}_\Omega\to\bb{A}_\Omega$ the map  described in Lemma \ref{lemma:automorph}, one obtains that $U(a)^*\pi_\omega(g)U(a)=\pi_{\omega}(T_a g)$ for every $a\in\R^2$, $\omega\in\Omega$ and $g\in \bb{A}_\Omega$. In terms of the direct integral representation $\pi$,
and of the lifted  magnetic translations $\rr{U}(a)$ the latter equation reads
\[
\rr{U}(a)^*\pi(g)\rr{U}(a)\;=\;\pi(T_a g)\;,\qquad \forall\; a\in\R^2\;, \quad
 \forall\; g\in \bb{A}_\Omega\;.
\]

\begin{remark}[Representing pair]\label{rk_rep_pair}
 In view of 
Proposition \ref{prop:rep_pair} the representation $\pi$ is faithful and non-degenerate.
Moreover, the mapping $a\mapsto \rr{U}(a)$ is a strongly continuous (hence weakly measurable) projective unitary representation of  $\R^2$ on $\s{H}_\Omega$. In summary the 
pair $(\pi,\rr{U})$ given by the
representation $\pi$ along with  (direct) magnetic translations $\R^2\ni a\mapsto \rr{U}(a)$ is a \emph{(twisted) representing pair} for 
the $C^\ast$-dynamical system $(\bb{A}_\Omega,\R^2,T)$  according to the definition  \cite[p. 511]{busby-smith-70}.  
\hfill $\blacktriangleleft$
\end{remark}

%------------------------------------------------%
\section{The algebra of perturbed magnetic operators}\label{sec:pert_magn}
In this section, we combine the two algebras described in Sections \ref{sec:unp_magn_alg} and \ref{Sec:pot} to obtain a larger algebra 
suitable to the study of perturbed magnetic systems.

\medskip

Let us start by introducing the algebra of magnetic operators perturbed by  
$\omega$-potentials. By taking inspiration from \cite{georgescu-iftimovici-02,georgescu-iftimovici-06}, we will introduce the linear space
\[
(\bb{C}_B\cdot\bb{A}_\Omega)_\omega\;:=\;\left.\left\{\sum_{n=1}^NA_n\pi_\omega(g_n)\;\right|\; A_n\in\bb{C}_B\;, g_n\in\bb{A}_\Omega\right\}
\]
of finite linear combinations of the products of magnetic operators 
$A_n$ and $\omega$-potentials $\pi_\omega(g_n)=M_{g,\omega}$.
Evidently, $(\bb{C}_B\cdot\bb{A}_\Omega)_\omega\subset \bb{B}(L^2(\R^2))$ and we will denote its closure by
\[
[[\bb{C}_B\cdot\bb{A}_\Omega]]_\omega\;:=\;\overline{(\bb{C}_B\cdot\bb{A}_\Omega)_\omega}^{\;\|\;\|}\;\subset\;\bb{B}(L^2(\R^2))
\] 
Since $\bb{A}_\Omega$ is unital one has that $\bb{C}_B\subset [[\bb{C}_B\cdot\bb{A}_\Omega]]_\omega$.
Let us introduce also the minimal $C^*$-algebra generated in $\bb{B}(L^2(\R^2))$ by the products $AM_{g,\omega}$ (and their adjoints  $M_{\overline{g},\omega}A^*$), denoted here by
\[
 \rr{A}_{B,\omega}\;:=\;C^*(\{AM_{g,\omega}\;|\;A\in\bb{C}_B\;, M_{g,\omega}\in\pi_\omega(\bb{A}_\Omega)\})\;\subset\; \bb{B}(L^2(\R^2))\;.
\]
At first glance
\[
[[\bb{C}_B\cdot\bb{A}_\Omega]]_\omega\;\subseteq\;\rr{A}_{B,\omega}\;\subset\;\bb{B}(L^2(\R^2))\;,
\]
but, as anticipated in Section \ref{sec:Intr0}, the first inclusion is indeed an equality.

 \proof[Proof of Theorem \ref{teo:ident1}]
By definition $\rr{A}_{B,\omega}$ is  the smallest $C^*$-algebra containing the products $AM_{g,\omega}$, and in turn $(\bb{C}_B\cdot\bb{A}_\Omega)_\omega$. Therefore, if one can prove that $[[\bb{C}_B\cdot\bb{A}_\Omega]]_\omega$ is a $C^*$-algebra then the equality \eqref{eq:caract1} follows immediately. The latter claim follows by looking at the $C^*$-algebra $\bb{A}_\Omega\rtimes_B\R^2$ introduced in Section \ref{sec:cross_prod}. One has that  every $\omega\in\Omega$ defines a representation $\pi_\omega:\bb{A}_\Omega\rtimes_B\R^2\to \bb{B}(L^2(\R^2))$. Moreover, Proposition \ref{prop:com_teo1} assures that there is a dense subset $\s{D}\subset\bb{A}_\Omega\rtimes_B\R^2$
such that 
\begin{equation}\label{con_theo}
\pi_\omega(\s{D})\;\subset\; (\bb{C}_B\cdot\bb{A}_\Omega)_\omega\;\quad\text{and}\quad\;\overline{\pi_\omega(\s{D})}^{\;\|\;\|}\;=\;[[\bb{C}_B\cdot\bb{A}_\Omega]]_\omega\;.
\end{equation}
 In view of the continuity of $\pi_\omega$ (being a $C^*$-algebra representation) one then infers that 
$\pi_\omega(\bb{A}_\Omega\rtimes_B\R^2)=[[\bb{C}_B\cdot\bb{A}_\Omega]]_\omega$, proving that 
$[[\bb{C}_B\cdot\bb{A}_\Omega]]_\omega$ is indeed a $C^*$-algebra.
Finally, equality \eqref{eq:caract1} can also be read as $\rr{A}_{B,\omega}=\pi_\omega(\bb{A}_\Omega\rtimes_B\R^2)$.
 \qed
 
 \medskip
 
We will refer to the $C^*$-algebra $\rr{A}_{B,\omega}$ as the algebra of \emph{$\omega$-perturbed magnetic operators}. We will justify this name below by considering some specific example, but to do this we will make use of the following result \cite[Proposition 2.2.7]{bratteli-robinson-87}:
\begin{lemma}\label{lemma:inv_sub}
Let $\bb{C}\subset \bb{B}$ two unital $C^*$-algebras, and $A\in\bb{C}$. If $A$ is invertible in $\bb{B}$ then $A^{-1}\in \bb{C}$.
\end{lemma}

\medskip

Let $H$ be a (not necessarily bounded) self-adjoint operator on $L^2(\R^2)$ such that $R(H):=(H-\ii{\bf 1})^{-1}\in\bb{C}_B$. An example of this type is given by the  Landau Hamiltonian \eqref{eq:intro_LH}. A standard argument shows that $R(H)\in \bb{C}_B$ implies that 
 $R_z(H):=(H-z{\bf 1})^{-1}\in\bb{C}_B$ for every $z\in\rho(H)$. In fact from the first resolvent identity one infers that $R_z(H)W_z(H)=R(H)$ with $W_z(H):={\bf 1}+(\ii-z)R(H)\in \bb{C}_B^+$,
where we denoted with $\bb{C}_B^+:=\C{\bf 1}+\bb{C}_B$ the standard unitalization of  
 $\bb{C}_B$ inside $\bb{B}(L^2(\R^2))$.
The   operator $W_z(H)$ admits a bounded inverse given by $W_z(H)^{-1}=(H-\ii{\bf 1})R_z(H)\in \bb{B}(L^2(\R^2))$. However, Lemma \ref{lemma:inv_sub} implies that $W_z(H)^{-1}\in\bb{C}_B^+$ and in turn $R_z(H)=R(H)W_z(H)^{-1}\in\bb{C}_B$.
Now let $V\in C^{\rm u}_{\rm b}(\R^2)$ be a uniformly continuous bounded function
and assume that $V:\R^2\to\R$. We will see $V$ as the multiplication operator on  $L^2(\R^2)$ defined by $(V\psi)(x):=V(x)\psi(x)$ and we will refer to it as the \emph{potential} $V$. As a consequence of the Kato-Rellich theorem \cite[Theorem X.12]{reed-simon-II} the \emph{perturbed} operator
\[
H_V\;:=\;H+V
\]
 is self-adjoint on the same domain of $H$. Let $z\in\rho(H)\cap\rho(H_V)$ be in the intersection of the two resolvent sets (\eg $z=\ii$). Then the second resolvent identity provides
$G_z(H,V)R_z(H_V)=R_z(H)$, with $G_z(H,V):= \big[{\bf 1}+R_z(H)V\big]$.
 Let us assume for the moment that there exists a {dynamical system} $(\Omega,\R^2,\rr{t})$, a continuous function $v:\Omega\to\R$ and a point $\omega\in\Omega$ such that $V(x)=v(\rr{t}_{x}(\omega))$. Then,  $R_z(H)V\in (\bb{C}_B\cdot\bb{A}_\Omega)_\omega$ and in turn $G_z(H,V)\in \rr{A}_{B,\omega}^+$. However,  $G_z(H,V)$ is invertible with inverse $G_z(H,V)^{-1}={\bf 1}-R_z(H_V)V\in \bb{B}(L^2(\R^2))$. From Lemma \ref{lemma:inv_sub} one infers that  $G_z(H,V)^{-1}\in \rr{A}_{B,\omega}^+$ and from the equation $R_z(H_V)=G_z(H,V)^{-1}R_z(H)$ one concludes that 
  $R_z(H_V)\in\rr{A}_{B,\omega}$. Finally, with the same argument used at the beginning of this paragraph one can prove that $R_z(H_V)\in\rr{A}_{B,\omega}$ for all $z\in\rho(H_V)$.
    In conclusion, we showed that the $C^*$-algebra $\rr{A}_{B,\omega}$ contains all the resolvents of the perturbed magnetic operator $H_V$, and this justifies its 
 name.
 
 \medskip

It is worth complementing the discussion above with some examples that clarify the role of the dynamical system $(\Omega,\R^2,\rr{t})$. 
\begin{example}[Constant perturbations]
Let $\Omega_\ast:=\{\ast\}$ be a singleton. Then $\bb{A}_{\Omega_\ast}=C(\{\ast\})=\C$ and in turn 
$(\bb{C}_B\cdot\bb{A}_{\Omega_\ast})_{\ast}=\bb{C}_B$. This (trivial) case describes magnetic operators perturbed by constant potentials, including the unperturbed case of an everywhere zero potential.
 \hfill $\blacktriangleleft$
\end{example}

\begin{example}[Periodic perturbations]
Let 
$\alpha,\beta\in\R^2$ be two non-aligned vectors and consider the two-dimensional lattice
$\Gamma:=\{n_1\alpha+n_2\beta\;|\;(n_1,n_2)\in\Z^2\}\simeq\Z^2$ along with the quotient $\Omega_\Gamma:=\R^2/\Gamma\simeq \n{T}^2$ homeomorphic to a two-dimensional torus. The space $\Omega_\Gamma$ inherits the action by translations of $\R^2$ on itself. More precisely, let us denote with $\R^2\ni y\mapsto[y]\in \Omega_\Gamma$ the canonical projection. Then, for every $x\in\R^2$ and $[y]\in \Omega_\Gamma$ one defines $\rr{t}_{x}([y]):=[y+x]$. The resulting dynamical system $(\Omega_\Gamma,\R^2,\rr{t})$ is evidently minimal. 
Let   $\bb{A}_{\Omega_\Gamma}:=C(\Omega_\Gamma)$, $v\in \bb{A}_{\Omega_\Gamma}$ a real-valued function and   $[y_0]\in\Omega_\Gamma$ a reference point. Then, one can define  the multiplication operator on $\bb{B}(L^2(\R^2))$ given by the function $V(x):=v([y_0+x])$. It turns out that $V$ is 
a continuous 
$\Gamma$-periodic function, and in particular $V\in C^{\rm u}_{\rm b}(\R^2)$. For that reason, we will refer to  $V$ as a \emph{$\Gamma$-periodic potential}. It is also true  that every continuous $\Gamma$-periodic function $V$ defines a $v\in \bb{A}_{\Omega_\Gamma}$ by $v([x]):=V(x)$. Moreover, in view of the minimality of the action of the translations on  $\Omega_\Gamma$ one can always use $[0]$ as the preferred reference point in the representation of the potential. In conclusion the $C^*$-algebra $\rr{A}_{B,\Gamma}:= [[\bb{C}_B\cdot\bb{A}_{\Omega_\Gamma}]]_{[0]}$ results to be the algebra associated to
  magnetic operators perturbed by 
 $\Gamma$-periodic potentials. %The torus $\Omega_\Gamma$ can be endowed with its (normalized) Haar measure $\dd\mu:=|\alpha\wedge\beta|^{-1}\dd^2 x$ (induced by the Lebesgue measure on $\R^2$) and and this permits to  consider the separable Hilbert space $L^2(\Omega_\Gamma,\mu)$. Let $\alpha^*,\beta^*\in\R^2$ be the dual vectors of the lattice $\Gamma$ defined by the relations $\alpha\cdot\alpha^*=2\pi=\beta\cdot\beta^*$ and
% $\alpha\cdot\beta^*=0=\beta\cdot\alpha^*$. The Fourier vectors $e_{r_1,r_2}([x]):=\expo{\ii(r_1\alpha^*+r_2\beta^*)\cdot[x]}$, with $(r_1,r_2)\in\Z^2$, provides an orthonormal basis of  $L^2(\n{T}^2,\mu)$. In particular, one has that 
% $e_{r_1,r_2}\in\bb{A}_{\Gamma}\subset L^2(\n{T}^2,\mu)$ for every $(r_1,r_2)\in\Z^2$.
 \hfill $\blacktriangleleft$
\end{example}

\begin{example}[Family of perturbations and their translations]
A more general situation is described in \cite[Section 2.4]{Bellissard-93}. 	 
First of all, by recalling that $L^\infty(\R^2)$ is isometrically isomorphic to the dual of $L^1(\R^2)$, one can see every element of  
  $C^{\rm u}_{\rm b}(\R^2)\subset L^\infty(\R^2)$ as a functional over $L^1(\R^2)$. Let $\bb{F}\subset C^{\rm u}_{\rm b}(\R^2)$ a family of functions such that $r_{\bb{F}}:=\sup\{\|f\|_\infty\;|\; f\in\bb{F}\}<+\infty$. Moreover, for every $a\in\R^2$ 
  and for every $f\in\bb{F}$ let $\rr{t}_{a}(f)\in C^{\rm u}_{\rm b}(\R^2)$ be the function defined by
  $\rr{t}_{a}(f)(x):=f(x-a)$. Let
 \[
 \Omega_{\bb{F}}\;:=\;\overline{\left\{\rr{t}_{a}(f)\;|\; a\in\R^2\;,\;\; f\in\bb{F}\right\}}^{\;w^*}
 \] 
 where the closure is with respect to the $\ast$-weak topology of the dual space of $L^1(\R^2)$.
Moreover, since translations do not increase the norm, one has that  $\Omega_{\bb{F}}$ is contained in the closed ball of radius ${r_{\bb{F}}}$ in $L^\infty(\R^2)$.  Since the Banach-Alaoglu theorem states that closed balls are compact with respect to the $\ast$-weak topology, it turns out that $\Omega_{\bb{F}}$ is a closed subset of a compact set, hence compact. Moreover the map $a\mapsto\rr{t}_a$ induces a continuous action of $\R^2$ on $\Omega_{\bb{F}}$. Now every point $\omega\in  \Omega_{\bb{F}}$ can be identified by a function $V_\omega\in L^\infty(\R^2)$
and one can consider the function $v: \Omega_{\bb{F}}\to\C$ defined by
$v(\omega)=V_\omega(0)$. It turns out that $V_\omega(x)=v(\rr{t}_{x}(\omega))$ and that $V_\omega\in C^{\rm u}_{\rm b}(\R^2)$ \cite[Corollary 2.4.2]{Bellissard-93}. The $C^*$-algebra $[[\bb{C}_B\cdot\bb{A}_{\Omega_{\bb{F}}}]]_\omega$ describes magnetic operators perturbed by potentials in  $\bb{F}$ and translations of them.
 \hfill $\blacktriangleleft$
\end{example}

Returning to the general discussion let us observe that in view of the fact that  $\bb{C}_B$ commute with the magnetic translations and the covariance of the representations $\pi_\omega$, one obtains the extended covariance relation
\begin{equation}\label{eq:cov_cond_V}
 V(a)^*\;\rr{A}_{B,\omega}\;V(a)\;=\;
 \rr{A}_{B,\rr{t}_a(\omega)}\;,\qquad \forall\; a\in\R^2\;, \quad
 \forall\omega\in\Omega\;.
\end{equation}
Moreover, by collecting all the $C^*$-agebras $\rr{A}_{B,\omega}$ in the direct integral one obtain the \emph{full} algebra of perturbed magnetic operators $\rr{A}_{B,\Omega}$ defined in \eqref{eq:int_rep_1}.

%------------------------------------------------%
\section{Relation with the twisted crossed product}\label{sec:cross_prod}
In this section we will describe the $C^*$-algebras $\rr{A}_{B,\omega}$ and $\rr{A}_{B,\Omega}$ 
 as representations of an \virg{abstract} $C^*$-algebra known as \emph{twisted crossed product} \cite{busby-smith-70,pedersen-79,packer-raeburn-89,williams-07}.

 \medskip

Let us start  with the Lebesgue-Bochner space
$$
\s{L}^1\;:=\;L^1(\R^2,\bb{A}_{\Omega})
$$
  of integrable functions from $\R^2$ to $\bb{A}_{\Omega}$  \cite[Appendix B]{williams-07}. This is a Banach space   with respect to the $L^1$-type norm
\begin{equation}\label{eq_L1-norm}
  \nnnorm{F}_1\;:=\;\frac{1}{2\pi\ell^2}\int_{\R^2}\dd x\;\lVert F(x)\rVert _\infty\;.
\end{equation}
We need also to introduce the  space of the compactly supported functions
over $\R^2$ with values in $\bb{A}_{\Omega}$ 
$$
\s{K}\;:=\;C_c(\R^2,\bb{A}_{\Omega})\;,
$$ 
and the space 
$$
\s{K}_0\;:=\;C_c(\Omega\times\R^2)\;,
$$ 
of the compactly supported continuous functions on
$\Omega\times\R^2$.
\begin{lemma} The Banach space  $\s{L}^1$ is separable. Moreover, one has the following inclusions 
$
\s{K}_0\subset \s{K} \subset \s{L}^1
$
and both $\s{K}_0$ and $\s{K}$ are dense in $\s{L}^1$ with respect to the norm-topology of $\nnnorm{\;}_1$. 
\end{lemma}
\proof
The inclusion $\s{K}\subset\s{L}^1$ is evident and the density of $\s{K}$ is proved in 
 \cite[Proposition B.33
]{williams-07}. The proper\footnote{This inclusion is proper in general, see \eg~\cite[Remark
  2.32]{williams-07}.} inclusion $\s{K}_0\subset\s{K}$
is  given by the identification
of  functions on $\Omega\times\R^2$ with functions from $\R^2$ to
$C(\Omega)$. For the density of $\s{K}_0$, one first observes that 
the
algebraic tensor product $C_c(\R^2)\odot \bb{A}_{\Omega}$ is
  contained in $\s{K}_0$. Then,
 \cite[Lemma 1.87]{williams-07} proves that $C_c(\R^2)\odot \bb{A}_{\Omega}$ is dense in 
 $\s{K}$ in the inductive limit topology, and therefore for the topology  induced by the norm $\nnnorm{\;}_1$. The separability of  $\s{L}^1$ follows from the density of simple functions in 
 $\s{L}^1$ \cite[Proposition B.33
]{williams-07} along with the
  the separability of $\bb{A}_{\Omega}$ and 
  $L^1(\R^2)$ (equivalently from the fact that  the Borel $\sigma$-algebra of $\R^2$ is countably generated). 
  \qed

\medskip

We will need also a second density result. For that let us introduce the linear spaces
\[
\s{F}\;:=\;F(\R^2)\odot \bb{A}_{\Omega}\;,\qquad\s{S}\;:=\;S(\R^2)\odot \bb{A}_{\Omega}
\]
where $F(\R^2)$ is the space of finite linear combinations  of Laguerre functions \eqref{eq:lag_pol}, and   $S(\R^2)$ denotes the space of Schwartz functions on $\R^2$. 
Since $F(\R^2)\subset S(\R^2)$ one infers that $\s{F}\subset \s{S}$.
Elements of $\s{F}$ or $\s{S}$ are finite linear combinations of simple tensors $(f\odot g)(\omega,x):=f(x)g(\omega)$
with  $g\in\bb{A}_{\Omega}$ and $f$ an element of $F(\R^2)$ or $S(\R^2)$, respectively.
\begin{lemma}\label{lemma:F_dens} 
The linear spaces $\s{F}$ and $\s{S}$ are dense in  $\s{L}^1$. 
\end{lemma}
\proof
In view of the inclusion $\s{F}\subset \s{S}$, it is sufficient to prove the claim only for $\s{F}$. The main point is that  the space ${F}(\R^2)$ is dense in $L^1(\R^2)$. A way to prove this  is to consider the \emph{Riesz means} of the expansion of the functions in  $L^1(\R^2)$ on the Laguerre basis $\psi_{n,m}$ 
\cite[Theorem 2.5.1]{thangavelu-93}. This implies that elements in $L^1(\R^2)\odot \bb{A}_{\Omega}$ (for instance simple functions) are arbitrarily well approximated by elements in  ${\s{F}}$ with respect to the 
norm $\nnnorm{\;}_1$. The density of $L^1(\R^2)\odot \bb{A}_{\Omega}$ in the space of the  $L^1$-Bochner integrable functions \cite[Proposition~B.33]{williams-07} concludes the proof.  \qed

\medskip

The Banach space $\s{L}^1$ can be  made into a Banach $\ast$-algebra by means of the following two operations
\begin{align}\label{eq:sub_alg_struct0} 
(F\star G)(\omega,x)\;:&=\;
\frac{1}{2\pi\ell^2}\int_{\R^2}\dd y\; F(\omega,y)\; G\left(\rr{t}_{-y}(\omega), x-y\right)\;\Theta_B(y,x)   \\  
%&=\;
%\int_{\R^d}\dd y\; f(\omega,y)\; \; \expo{\frac{\ii}{2}\Phi_B(x,y)}\;, \nonumber\\
%\nonumber\\
F^\star(\omega,x)\;:&=\;%\hat{\tau}_x[f(\omega,-x)]\;=\;
\overline{F(\rr{t}_{-x}(\omega),-x)}\label{eq:sub_alg_struct1} \;.
\end{align}
where in \eqref{eq:sub_alg_struct0}  we introduced the notation  $
\Theta_B(x,y):=\expo{\ii\frac{x\wedge y}{2\ell^2}}$. The map $\Theta_B:\R^2\times\R^2 \to \n{S}^1$ meets the   \emph{2-cocycle} condition
  \[
    \Theta_B(x,y+z) \Theta_B(y,z)\;  =\; \Theta_B(x,y) \Theta_B(x+y,z)\;,\qquad\quad x,y,z\in\R^2
  \]
and the \emph{normalization} conditions $\Theta_B(x,\pm x)=\Theta_B(0, x)=\Theta_B(x,0)=1$.
\begin{proposition}
The quadruple $\s{L}_B^1:=(\s{L}^1, \star, ^\star,\nnnorm{\;}_1)$ is a separable Banach $\ast$-algebra and 
$\s{K}_0$ and $\s{K}$ are dense subalgebras.
\end{proposition}
\proof
The product $\star$ is a twisted convolution,
hence the fact that  $f\star g$ is still in $\s{L}^1$ can be proved by the Bochner generalization of the Young's inequalities which provides
\[
 \nnnorm{F\star G}_1\;\leqslant\;  \nnnorm{F}_1  \nnnorm{G}_1 \;,\qquad F,G\in\s{L}^1\;.
\]
 A detailed calculation of the fact that 
the operations \eqref{eq:sub_alg_struct0} and  \eqref{eq:sub_alg_struct1} endowe  $\s{L}^1$ with the structure of a $\ast$-algebra is provided in \cite[Section 6]{busby-smith-70}. See also the discussion in \cite[Section
 2.3]{williams-07} and \cite[Section
 B.2.1]{williams-07} which can be easily adapted to the case of a non-trivial twisted  2-cocycle.
 \qed 
 
\medskip

One refers to $\s{L}_B^1$ as the \emph{group algebra} on $\R^2$ with values in $\bb{A}_{\Omega}$ \emph{twisted} by $\Theta_B$ \cite[Section 2]{busby-smith-70}.

\begin{remark}[Approximate identity]
  \label{rk:approx_ident} 
  The Banach $\ast$-algebra $\s{L}^1_B$ has (left) sequential approximate identities \cite[Theorem 3.2]{busby-smith-70}.  An explicit realization is provided by the family $\{\imath_n\}_{n\in\N}$ with $\imath_n(\omega,x):=n^2\frac{\pi\ell^2}{2}\chi_{{n}}(x)$ where $\chi_{{n}}$ is the characteristic function of the square $Q_n:=[-n^{-1},n^{-1}]^2\subset\R^2$. By construction $\nnnorm{\imath_n}_1=1$ for all $n\in\N$.
 If one replaces   $\chi_{n}$ with suitable smooth functions supported in $Q_{n}$ one gets$\{\imath_n\}_{n\in\N}\subset\s{K}_0$.
   \hfill $\blacktriangleleft$
\end{remark}

\medskip

The realization of the Banach $\ast$-algebra $\s{L}^1_B$ is the principal step 
towards the definition of the relevant $C^*$-algebra.
To equip $\s{L}^1_B$ with a $C^\ast$-structure we have to replace the norm $\nnnorm{\;}_1$, which does not verifies the $C^\ast$-property, with a new $C^*$-norm.  This  can be done by introducing  the \emph{universal norm} of an element  $F\in \s{L}^1_B$ as
\begin{equation}\label{eq13}
\|F\|_{\rm u}\;:=\;\sup\left\{\|\rho(F)\|_{\bb{B}(\s{H})}\;\left|\; (\rho,\s{H})\ \text{is a Hilbert space $\ast$-representation of}\ \ \s{L}^1_B\right\}\right.\;.
\end{equation}
Since every $\ast$-representation is automatically continuous \cite[Proposition 2.3.1]{bratteli-robinson-87}, the supremum is well defined and 
\begin{equation}\label{ineq_univ_norm}
\|F\|_{\rm u}\;\leqslant\; \nnnorm{F}_1\;,\qquad \forall\; F\in \s{L}^1_B\;.
\end{equation}
Equation \eqref{eq13} defines a $C^\ast$-seminorm on $\s{L}^1_B$. In fact $\|\;\|_{\rm u}$ turns out to be  a norm since there exist faithful Hilbert space representations of  $\s{L}^1_B$ (see Proposition \ref{prop:int_rep_pair}).
The resulting $C^\ast$-algebra 
\[
\bb{A}_\Omega\rtimes_B\R^2\;:=\;\overline{\s{L}^1_B}^{\;\|\;\|_{\rm u}}
\]
is called the \emph{twisted crossed product} of the dynamical system $(\Omega,\R^2,\rr{t})$ in mathematical literature \cite{busby-smith-70,pedersen-79,packer-raeburn-89,williams-07}.

\begin{remark}\label{rk:one-point}
In the special case $\Omega=\{\ast\}$, or equivalently $\bb{A}_\Omega=\C$, the 
twisted crossed product described above reduces to the twisted group $C^*$-algebra of the group $\R^2$. This particular case has been considered in \cite[Sections 2.2 \& 2.3]{denittis-sandoval-00} under the name of \emph{magnetic group $C^*$-algebra}. In more detail one has that every $f\in L^1(\R^2)$ defines a magnetic operator $K_f\in \bb{C}_B$ through the prescription
\[
(K_f\varphi)(x)\;:=\;\frac{1}{2\pi\ell^2}\int_{\R^2}\dd y\; f(x-y)\;\Theta_B(x,y)\;\varphi(y)
\]
and the mapping $K:f\mapsto K_f$ is an injective representation of the convolution algebra $L^1(\R^2)$. In fact, compared with the representation $\pi$ introduced in  \cite[p.13]{denittis-sandoval-00} one gets that 
$K_f=\pi(f^-)$ where $f^-(x):=f(-x)$.
Denoting with $\bb{G}_B$ the clausure of $L^1(\R^2)$ with respect to the universal norm one gets that $K:\bb{G}_B\to \bb{C}_B$ is an isomorphism of  $C^*$-algebras.
With the same argument of Lemma \ref{lemma:F_dens}, one deduces that $K(F(\R^2))$ is dense in 
$\bb{C}_B$. As a final remark, by comparing the proof of \cite[Proposition 2.10]{denittis-sandoval-00} with the fact that $\psi_{j,k}^-=(-1)^{j-k}\psi_{j,k}$, one infers that $K_{\psi_{j,k}}=(\sqrt{2\pi} \ell)^{-1}\Upsilon_{j\mapsto k}$.
\hfill $\blacktriangleleft$
\end{remark}

\medskip

The definition of the universal norm and inequality \eqref{ineq_univ_norm} imply two important facts. First of all  every Hilbert space $\ast$-representation of 
$\s{L}^1_B$ has a unique continuous extension to a representation of 
$\bb{A}_\Omega\rtimes_B\R^2$.
The second is the content of the following result:
\begin{lemma}\label{lemma:dens04}
Every  subspace  of $\s{L}^1_B$ 
which is dense in the topology of the norm $\nnnorm{\;}_1$,
is also dense in $\bb{A}_\Omega\rtimes_B\R^2$ with respect to the topology of the universal norm $\|\;\|_{\rm u}$.
\end{lemma}
\proof
Let   
${\f{D}}\subset \s{L}^1_B$ 
be a $\nnnorm{\;}_1$-dense   subspace. The claim follows
by observing that  for each $\epsilon>0$ and each $F\in\bb{A}_\Omega\rtimes_B\R^2$ there are a $F_0\in \s{L}^1_B$  and a $F_\epsilon\in{\f{D}}$  such that $\|F-F_0\|_{\rm u}<\epsilon/2$ and $\nnnorm{F_0-F_\varepsilon}_1<\epsilon/2$. The inequality 
$$
\|F-F_\epsilon\|_{\rm u}\;\leqslant\;\|F-F_0\|_{\rm u}+\|F_0-F_\epsilon\|_{\rm u}\;\leqslant\;\|F-F_0\|_{\rm u}+\nnnorm{F_0-F_\epsilon}_1\;\leqslant\;\epsilon
$$
 completes the proof.
\qed

\medskip

For the next result we need one more definition. Let $\s{D}$ be a linear subspace of $\bb{A}_\Omega\rtimes_B\R^2$. Then its adjoint is the linear space
\[
\s{D}^\star\;:=\;\left\{F^\star\in\bb{A}_\Omega\rtimes_B\R^2\;|\; F\in\s{D}   \right\}\;.
\]
\begin{lemma}\label{lemma:dens04-2}
If $\s{D}\subset \bb{A}_\Omega\rtimes_B\R^2$ is a dense subspace then so is also its adjoint $\s{D}^\star$.
\end{lemma}
\proof
The claim follows by the equality $\|F^\star\|_{\rm u}=\|F\|_{\rm u}$ for every $F\in \bb{A}_\Omega\rtimes_B\R^2$, which expresses the  continuity of the adjoint, and the fact that $\bb{A}_\Omega\rtimes_B\R^2$ is closed under the adjoint.
\qed

\medskip 

The importance of Lemmas \ref{lemma:dens04} and  \ref{lemma:dens04-2} relies on the fact that they  allow us to check (linear) proprieties of the full $C^\ast$-algebra $\bb{A}_\Omega\rtimes_B\R^2$ simply looking at  the 
 linear subspaces $\s{K}$, $\s{K}_0$, $\s{S}$ or $\s{F}$ (and their adjoints) which are much easier to handle. 

\medskip

We are now in position to introduce certain special representations of  $\bb{A}_\Omega\rtimes_B\R^2$. First of all let us observe that given an element $F\in \s{L}^1_B$ and a point $\omega\in\Omega$ one can define an operator $\pi_{\omega}(F)$ on $L^2(\R^2)$ by the prescription
\begin{equation}\label{eq:c01-ññ}
  (\pi_{\omega}(F)\varphi)(x)\;:=\;\frac{1}{2\pi\ell^2}\int_{\R^2}\dd y\; F\left(\rr{t}_{x}(\omega),
    x-y\right)\;\Theta_B(x,y)\;\varphi(y)\;.
\end{equation}
As a consequence of the Young’s inequality one gets
\[
\|\pi_{\omega}(F)\varphi\|\;\leqslant\;\nnnorm{F}_1\;\|\varphi\|
\]
proving that the operator norm of $\pi_{\omega}(F)$ is  bounded by $\nnnorm{F}_1$. Moreover, a direct computation shows that $\pi_{\omega}(F\star G)=\pi_{\omega}(F)\pi_{\omega}(G)$
and $\pi_{\omega}(F^\star)=\pi_{\omega}(F)^*$. This means that the mapping $\pi_\omega: F\mapsto \pi_{\omega}(F)$ provides a $\ast$-representation of $\s{L}^1_B$ into $\bb{B}(L^2(\R^2))$
which extends to a $C^*$-representation of the twisted crossed product $\bb{A}_\Omega\rtimes_B\R^2$. We will refer to $\pi_{\omega}$ as the $\omega$-\emph{left-regular representation} and it turns out that
\[
\pi_\omega(\bb{A}_\Omega\rtimes_B\R^2)\;=\;[[\bb{C}_B\cdot\bb{A}_\Omega]]_\omega\;=\;\rr{A}_{B,\omega}
\]
 with the notation introduced in Section \ref{sec:pert_magn}. In fact the following result completes the proof of
  Theorem \ref{teo:ident1}.
\begin{proposition}\label{prop:com_teo1}
There is a dense subset $\s{D}\subset \bb{A}_\Omega\rtimes_B\R^2$ that satisfies the conditions in 
\eqref{con_theo}.
\end{proposition}
\proof
Let us start by showing that $\pi_\omega(\s{F}^\star)\subset\pi_\omega(\s{S}^\star)\subset(\bb{C}_B\cdot\bb{A}_\Omega)_\omega$. The first inclusion is obvious. Concerning the second inclusion 
let $f\odot g\in \s{S}$  with $f\in S(\R^2)$ and $g\in\bb{A}_\Omega$.
Then, by definition
\[
\begin{aligned}
(\pi_\omega(f\odot g)\varphi)(x)\;&=\;
\frac{1}{2\pi\ell^2}\int_{\R^2}\dd y\; f(x-y)g\left(\rr{t}_{x}(\omega)\right)\;\Theta_B(x,y)\;\varphi(y)\\
&=\;M_{_{g,\omega}}\left(\frac{1}{2\pi\ell^2}\int_{\R^2}\dd y\; f(x-y)\;\Theta_B(x,y)\;\varphi(y)\right)\\
&=\;(M_{_{g,\omega}}K_f\varphi)(x)
\end{aligned}
\]
where $K_f\in \bb{C}_B$ is the magnetic operator associated with the symbol $f$ 
as described in Remark \ref{rk:one-point}. Since
\[
\pi_\omega\left((f\odot g)^\star\right)\;=\;(\pi_\omega(f\odot g))^*\;=\;(M_{_{g,\omega}}K_f)^*\;=\;K_{\overline{f}}M_{_{\overline{g},\omega}}
\]
one infers that $\pi_\omega(\s{S}^\star)\subset (\bb{C}_B\cdot\bb{A}_\Omega)_\omega$. Now, in view of
 Lemmas \ref{lemma:dens04}
and \ref{lemma:dens04-2} one has that $\s{F}^\star$ is dense in $\bb{A}_\Omega\rtimes_B\R^2$.
Moreover $\pi_\omega(\s{F}^\star)$ turns out to be dense in $(\bb{C}_B\cdot\bb{A}_\Omega)_\omega$
since the subspace $K(F(\R^2))$ is dense in $\bb{C}_B$ as discussed at the end of Remark \ref{rk:one-point}. Therefore,  both $\s{F}^\star$ or $\s{S}^\star$ meet the conditions in \eqref{con_theo}.
\qed

\begin{remark}\label{rk:sep_ide_C*}
The separability of $\s{L}^1_B$ implies that also $\bb{A}_\Omega\rtimes_B\R^2$, 
and in turn $\rr{A}_{B,\omega}$ are separable. Moreover, from  Remark \ref{rk:approx_ident} it follows that 
$\rr{A}_{B,\omega}$ has a sequential approximate identity given by the family 
of operators $I_{\omega,n}:=\pi_\omega(\imath_n)$ acting as
$$
(I_{\omega,n}\varphi)(x)\;=\;\left(\frac{n}{2}\right)^2\int_{Q_n+x}\dd y\; \Theta_B(x,y)\;\varphi(y)
$$
where $Q_n+x$ denotes the cube $Q_n$ translated by the vector $x$.
\hfill $\blacktriangleleft$
\end{remark}

\medskip

There is a second relevant representation of the twisted crossed product $\bb{A}_\Omega\rtimes_B\R^2$ on the Hilbert space $L^2(\Omega\times\R^2)$ given by the product measure on $\Omega\times\R^2$. For a given 
$F\in \s{L}^1_B$  one can define an operator $\pi(F)$ acting on  $\Phi\in L^2(\Omega\times\R^2)$ according to the prescription
\begin{equation}\label{eq:c01-ññ_MM}
  (\pi(F)\Phi)(\omega,x)\;:=\;\frac{1}{2\pi\ell^2}\int_{\R^2}\dd y\; F\left(\rr{t}_{x}(\omega),
    x-y\right)\;\Theta_B(x,y)\;\Phi(\omega,y)\
\end{equation}
 Exactly as above, one can check that $\pi$ defines a $\ast$-representation of the Banach $\ast$-algebra $\s{L}^1_B$, and therefore extends to a representation of the {twisted crossed product} $\bb{A}_\Omega\rtimes_B\R^2$.
We will refer to $\pi$ as the \emph{integrated left-regular representation}.

\medskip

In order to understand the relation between the representation $\pi$ and the family of representations $\{\pi_\omega\}_{\omega\in\Omega}$ it is convenient to look at the identification 
$\s{H}_\Omega\simeq L^2(\Omega\times\R^2)$ with the direct integral  
introduced in Section \ref{Sec:pot}. By identifying $\Phi\in L^2(\Omega\times\R^2)$ with an element $\hat{\phi}\in \s{H}_\Omega$ by means of $\Phi(\omega,x)=\phi_\omega(x)$, one can read equation \eqref{eq:c01-ññ_MM} in the following form:
\[
(\pi(F)\hat{\phi})_\omega(x)\;=\;(\pi_\omega(F)\phi_\omega)(x)\;.
\]
This shows that $\pi(F)$ acts as a decomposable operator on $\s{H}_\Omega$, and
\[
\pi(F)\;=\;\int^\oplus_{\Omega}\dd\n{P}(\omega)\;\pi_\omega(F)\;.
\]
By density this result extend to the full algebra $\bb{A}_\Omega\rtimes_B\R^2$. Finally, a comparison with definition \eqref{eq:int_rep_1} and Theorem \ref{teo:ident1} provide that
\[
\pi(\bb{A}_\Omega\rtimes_B\R^2)\;=\;\rr{A}_{B,\Omega}\;.
\]

\begin{proposition}\label{prop:int_rep_pair}
The representation $\pi$ defined by \eqref{eq:c01-ññ_MM} is faithful and non-degenerate.
\end{proposition}
\proof
The change $y\mapsto x-y$ of variable in  \eqref{eq:c01-ññ_MM} provides 
$$
\begin{aligned}
(\pi(F)\hat{\phi})_\omega(x)\;&=\;\frac{1}{2\pi\ell^2}\int_{\R^2}\dd y\; F\left(\rr{t}_{x}(\omega),
    y\right)\;\Theta_B(y,x)\;\phi_\omega(x-y)\\
    &=\;\frac{1}{2\pi\ell^2}\int_{\R^2}\dd y\; F\left(\rr{t}_{x}(\omega),
    y\right)\;(\rr{U}(y)\hat{\phi})_\omega(x)\;
\end{aligned}
$$
where $\rr{U}(y)$ is the direct magnetic translation introduced at the end of Section \ref{Sec:pot}.
By considering $F:\R^2\to \bb{A}_\Omega$ as a the map which associate to $y\in\R^2$ the element $F_y:=F(\cdot, y)\in\bb{A}_\Omega$,
then one can interpret $F\left(\rr{t}_{x}(\omega),
    y\right)\psi_\omega(x)$ as $(\pi_\omega(F_y)\psi_\omega)(x)$ 
for every $\omega\in\Omega$. By passing at the integrated representation one gets
 that
\[
(\pi(F)\hat{\phi})_\omega(x)\;=\;\left(\frac{1}{2\pi\ell^2}\int_{\R^2}\dd y\;\pi(F_y)\rr{U}(y)\hat{\phi}\right)_\omega(x)
\]
for every $\hat{\phi}\in\s{H}_\Omega$, and in turn 
\begin{equation}\label{eq:rep_int}
\pi(F)\;=\;\frac{1}{2\pi\ell^2}\int_{\R^2}\dd y\;\pi(F_y)\rr{U}(y)\;.
\end{equation}
This shows that  $\pi(F)$ corresponds to the \emph{integrated (regular) representation} induced by the representing pair $(\pi,\rr{U})$ described in Remark \ref{rk_rep_pair}
according to \cite[Theorem 3.3]{busby-smith-70} (see also  \cite[Section 7.7]{pedersen-79}). In particular, in \cite[Theorem 3.3]{busby-smith-70} it is proven that ${\pi}$ is a  non-degenerate representation. Moreover, it immediately follows that
$$
\| \pi(F)\|\;\leqslant\; \frac{1}{2\pi\ell^2}\int_{\R^2}\dd y\; \|\pi(F_y)\|\;=\;\nnnorm{F}_1\;.
$$
The faithfulness of the representation ${\pi}$ is a consequence of  the amenability of the group $\R^2$ \cite[Corollary 7.7.8]{pedersen-79}. \qed

%-----------%

\section{The canonical trace}\label{sec:tr_u_v}
In order to introduce the canonical trace on the $C^*$-algebra $\rr{A}_{B,\Omega}$ we need first to consider the associated von Neumann algebra.
In fact, von Neumann algebras are the natural objects where to set a solid theory of traces.
The construction of the canonical trace described bellow follows the ideas of \cite{lenz-99} and is based on the fact that there is a natural \emph{Hilbert algebra structure} underlying  $\bb{A}_\Omega\rtimes_B\R^2$ (Appendix \ref{sec_Hilb_str}).

\medskip

Let us introduce the \emph{von Neumann
algebra of perturbed magnetic operators} as the bicommutant of $\rr{A}_{B,\Omega}$, \ie
\[
\rr{M}_{B,\Omega}\;:=\;\rr{A}_{B,\Omega}''\;\subset\;\bb{B}(\s{H}_\Omega)\;.
\]
In view of the integral representation \eqref{eq:rep_int} one has the following characterization of $\rr{M}_{B,\Omega}$ \cite[Theorem 3.3 (3)]{busby-smith-70}:
\begin{proposition}\label{prop:caratVN1}
The von Neumann algebra $\rr{M}_{B,\Omega}$ coincides with the von Neumann 
subalgebra of $\bb{B}(\s{H}_\Omega)$ generated by the (full) algebra of potentials $\bb{P}_\Omega$ defined by \eqref{eq:full_pot} and by the family  $\{\rr{U}(a)\;|\; a\in\R^2\}$ of the (direct) magnetic translations defined by \eqref{eq:magn_tral_dir}.
\end{proposition}

\medskip

As for the unperturbed cases where the relevant von Neumann algebra \eqref{int_VNA} is a commutant, we can shows that also $\rr{M}_{B,\Omega}$ is the commutant inside $\bb{B}(\s{H}_\Omega)$ of certain symmetries.
For that we need a lift of the (dual) magnetic translations $V(a)$'s  on
the direct integral $\s{H}_\Omega$. For every $a\in\R^2$  and every  $\hat{\varphi}:=\{\varphi_\omega\}_{\omega\in\Omega}\in\s{H}_\Omega$, let $\rr{V}(a)$ be the
operator defined by
\begin{equation}\label{eq:magn_tral}
 \rr{V}(a)\hat{\varphi}\;=\;\{(\rr{V}(a)\hat{\varphi})_\omega\}_{\omega\in\Omega}\;=\;\{V(a)\varphi_{\rr{t}_{a}(\omega)}\}_{\omega\in\Omega}\;.
\end{equation}
 From the definition
it results that the operators $\rr{V}(a)$ do not preserve the fibers of
$\s{H}_\Omega$. In fact $\rr{V}(a)$ sends elements of the fiber over $\rr{t}_{a}(\omega)$
into elements of the fiber over $\omega $ according to the rule $
\varphi_{\rr{t}_{a}(\omega)}\mapsto(\rr{V}(a)\hat{\varphi})_\omega$ given by
\begin{equation}\label{eq:magn_tral-LL}
(\rr{V}(a)\hat{\varphi})_\omega(x)\;:=\;\expo{\ii\frac{x\wedge a}{2\ell^2}}\varphi_{\rr{t}_{a}(\omega)}(x-a)\;,\qquad x\in\R^2\;,\quad \omega\in\Omega\;.
\end{equation}
In particular, a comparison between~\eqref{eq:magn_tral}
and~\eqref{eq:rand_op} shows that the operators $\rr{V}(a)$ are
not  decomposable. However, in view of the
invariance of the measure $\n{P}$ one can check that the $\rr{V}(a)$'s are
still unitary. Moreover, they meet  the composition rule
\begin{equation}\label{eq:magn_tral2}
 \rr{V}(a)\rr{V}(b)\;=\;\expo{\ii\frac{b\wedge a}{2\ell^2}}\;\rr{V}(a+b)\;,\qquad\quad a,b\in\R^2\;,
\end{equation}
providing a projective unitary representation of
the group $\R^2$ on the Hilbert space $\s{H}_\Omega$. 
The covariance condition \eqref{eq:cov_cond_V} translate as $\rr{V}(a)\pi(g)\rr{V}(a)^*=\pi(g)$ for every $a\in\R^2$ and $g\in\bb{A}_\Omega$. Moreover, in view of the commutation relations between the $U(a)$'s and the $V(b)$'s one infers that $\rr{U}(a)\rr{V}(b)=\rr{V}(b)\rr{U}(a)$ for every $a,b\in\R^2$.
Let $\bb{V}_\Omega:=C^*(\{\rr{V}(a)\;|\;a\in\R^2\})$ be the $C^*$-algebra generated by the unitaries $\rr{V}(a)$'s and consider the commutant
\begin{equation}\label{int_VNA_Omega}
\bb{V}'_\Omega\;=\;\{A\in\bb{B}(\s{H}_\Omega))\;|\; AB-BA=0\;,\;\;\forall\; B\in\bb{V}_\Omega \}\:.
\end{equation}
\begin{proposition}
It holds true that
$\rr{M}_{B,\Omega}\;=\;\bb{V}'_\Omega \cap \bb{D}(\s{H}_\Omega)
$.
\end{proposition}
\proof
Since the intersection of von Neumann algebras is a von Neumann algebra \cite[Part II, Chapter 1, Proposition 1]{dixmier-81}, one gets that
the intersection of $\bb{V}'_\Omega$ and $\bb{D}(\s{H}_\Omega)$ is a von Neumann algebra. 
Moreover, form Proposition \ref{prop:caratVN1} one infers that $\rr{M}_{B,\Omega}\subseteq\bb{V}'_\Omega \cap \bb{D}(\s{H}_\Omega)$. 
To prove the opposite inclusion (see Corollary \ref{corol:commut}), and in turn the equality, one needs to explore  more in depth the structure of the Banach algebra $\s{L}^1_B$ and its 
integrated left-regular representation $\pi$. 
\qed

\medskip

In Section \ref{sec:cross_prod} it has been shown
  that the function space $\s{K}_0=C_c(\Omega\times\R^2)$ has the structure of a normed $\ast$-algebra with the operations defined by 
\eqref{eq:sub_alg_struct0} and \eqref{eq:sub_alg_struct1}. Indeed $\s{K}_0$ has a richer structure which is inherited from the natural inclusion of $\s{K}_0$ in the 
Hilbert space 
$L^2(\Omega\times\R^2)$.
More explicitly one can endow $\s{K}_0$ 
with the inner product
$\langle\langle\;,\;\rangle\rangle_{0}:\s{K}_0\times\s{K}_0\to\C$ defined 
 by
\begin{equation}\label{eq:inner_prod}
 \langle\langle
  F_1,F_2\rangle\rangle_{0}\;:=\; \frac{1}{2\pi\ell^2}\int_{\R^2}\dd x\int_{\Omega}\dd\n{P}(\omega)\;
  \overline{F_1(\omega,x)}F_2(\omega,x)
  \end{equation}
for all $F_1,F_2\in \s{K}_0$. The $\ast$-algebra $\s{K}_0$ endowed with 
this structure turns out to be a \emph{Hilbert algebra} according to the definition in \cite[Part I, Chapter 5]{dixmier-81} (see also \cite{nakano-50,kervin-69,rieffel-69}).
\begin{proposition}\label{prop:hilb_alg_ext}
The quadruple $(\s{K}_0,  \star, ^\star, \langle\langle\;,\;\rangle\rangle_{0})$ is a Hilbert algebra, namely:
  \begin{itemize}
  \item[(i)] $\langle\langle F_1,F_2\rangle\rangle_{0}=\langle\langle
    F_2^\star,F_1^\star\rangle\rangle_{0}$ for all $F_1,F_2\in \s{K}_0$;\vspace{1mm}
    
  \item[(ii)] $\langle\langle G\star F_1 ,F_2\rangle\rangle_{0}=\langle\langle
    F_1, G^\star\star F_2\rangle\rangle_{0}$ for all $F_1,F_2, G\in \s{K}_0$;\vspace{1mm}

  \item[(iii)] The set $\{F\star G\ |\ F,G\in \s{K}_0\}$ is dense in $L^2(\Omega\times\R^2)$;\vspace{1mm}
    
  \item[(iv)] For all $G\in \s{K}_0$ the linear maps $F\mapsto F\star G$ and $F\mapsto G\star F$ are continuous with respect to the topology induced by the inner product.
 \end{itemize}

  \end{proposition}

\medskip

The proof of this result is technical and quite standard. 
For the benefit of reader a sketch of the proof is postponed to Appendix \ref{sec_Hilb_str}.

\medskip

With the help of Proposition \ref{prop:hilb_alg_ext} for every $F\in \s{K}_0$ we can define  two
bounded operators  on $L^2(\Omega\times\R^2)$ denoted with  $R_F$ and $L_F$. Initially these operators are defined for elements  $\Phi\in \s{K}_0$ as  $R_F\Phi:= \Phi\star F$ and $L_F\varphi:= F\star\Phi$. Then, in view of Proposition \ref{prop:hilb_alg_ext} (iv) they extend by continuity to unique bounded operators on $L^2(\Omega\times\R^2)$. The identification of 
$L^2(\Omega\times\R^2)$ with the direct integral $\s{H}_\Omega$ given by $\Phi(\omega,x)=\phi_\omega(x)$ allows us  to write
$$
\begin{aligned}
(R_F\hat{\phi})_\omega(x)\;&=\;\frac{1}{2\pi\ell^2}\int_{\R^2}\dd y\; \phi_\omega(y)\; F\left(\rr{t}_{-y}(\omega), x-y\right)\; \Theta_B(y,x)\;,\\
(L_F\hat{\phi})_\omega(x)\;&=\;\frac{1}{2\pi\ell^2}\int_{\R^2}\dd y\; F(\omega,y)\; \phi_{\rr{t}_{-y}(\omega)}\left(x-y\right)\;  \Theta_B(y,x)\;,
\end{aligned}
$$
for every $\hat{\phi}\in\s{H}_\Omega$. Moreover, an  adaptation of the usual Young's convolution inequality shows that they  are bounded by
$\nnnorm{F}_1$. The family of these operators 
generates two von Neumann algebras inside $\bb{B}(\s{H}_\Omega)$ defined by
$$
\begin{aligned}
\rr{R}\;:=\; \{R_F\ |\ F\in \s{K}_0\}'' \;,\qquad
\rr{L}\;:=\; \{L_F\ |\ F\in \s{K}_0\}'' \;.
\end{aligned}
$$
The commutation theorem states that $\rr{R}'=\rr{L}$ and consequently 
$\rr{L}'=\rr{R}$ \cite[Part I, Chapter 5, Theorem 1]{dixmier-81}. As the last ingredient let us consider the operator $J$ on $\s{H}_\Omega$ defined by
$(J\hat{\phi})_{\omega}(x)\;:=\;\phi_{\rr{t}_{x}(\omega)}(x)$.
A direct computation shows that $J$
is a unitary operator, \ie $J^*=J^{-1}$. Indeed,   from the invariance of the measure $\n{P}$ one obtains that $J^*$   acts pointwise as $(J^{\ast}\hat{\phi})_{\omega}(x):=\phi_{\rr{t}_{-x}(\omega)}(x)$, namely as the inverse of $J$. 

\medskip

The von Neumann algebras  $\rr{R}$, $\rr{L}$ and $\rr{M}_{B,\Omega}$ are related by the conditions
\[
J\rr{L}J^*\;=\;\rr{M}_{B,\Omega}\;,\qquad \bb{V}'_\Omega\cap\bb{D}(\s{H}_\Omega)\;\subseteq \;J\rr{R}'J^*\;,
\]
proved in Lemmas \ref{lemma_str1} and \ref{lemma_str2},
respectively. Using the fact that $\rr{R}'=\rr{L}$, one immediately gets
\begin{corollary}\label{corol:commut}
It holds true $\bb{V}'_\Omega \cap \bb{D}(\s{H}_\Omega)\subseteq \rr{M}_{B,\Omega}$.
\end{corollary}

\medskip

We are now in position to use the Hilbert algebra structure in order to endowed
the von Neumann algebra $\rr{M}_{B,\Omega}$
 with  a
distinguished trace $\tau_{\n{P}}$ called \emph{trace per unit
  volume}. We will follow here the construction of $\tau_{\n{P}}$
proposed in~\cite{lenz-99} and \cite{kervin-69}. However, we need to generalize
this construction in order to include  the twist
induced by the 2-cocycle $\Theta_B(x,y)$.

\medskip

For the sake of notation clarity, let us introduce the space $\rr{K}_{B,\Omega}=\pi(\s{K}_0)$. Then every element $K_F\in \rr{K}_{B,\Omega}$ has a kernel $F\in \s{K}_0$ such that $K_F=\pi(F)$.
By means of the Hilbert algebra structure of $\s{K}_0$ one can associate to $K_F$ the number
\begin{equation}\label{eq:tra_un_vol_01}
\tau_{\n{P}}(K_F^*K_F)\;:=\;\langle\langle F,F\rangle\rangle_{0}
\end{equation}
where the  right-hand side  is defined by \eqref{eq:inner_prod}. We will show that the prescription \eqref{eq:tra_un_vol_01} uniquely
specifies a densely-defined trace on the von Neumann algebra $\rr{M}_{B,\Omega}$. For that, we need one more definition.
\begin{definition}[Operator with $L^2$-kernel]\label{def:oK}
Let $G\in L^2(\Omega\times\R^2)$ and consider the linear operator on $\s{H}_\Omega$ defined by
\begin{equation}\label{eq:L2-kern01}
 (K_G\hat{\phi})_{\omega}(x)\;:=\;\frac{1}{2\pi\ell^2}\int_{\R^2}\dd y\; G\left(\rr{t}_{x}(\omega),
    x-y\right)\; \Theta_B(x,y)\;\phi_\omega(y)\;
\end{equation}
for suitable $\hat{\phi}\in\s{H}_\Omega$. The function $G$ is called the \emph{kernel} of $K_G$ and the set of operators with kernel will be denoted by $\rr{Z}_{B,\Omega}$.
\end{definition}

\medskip

The subset of bounded operators with $L^2$-kernel will be denoted with
\begin{equation}\label{N-id}
\rr{N}_{B,\Omega}\;:=\;\rr{Z}_{B,\Omega}\;\cap\; \bb{B}(\s{H}_\Omega)\;.
\end{equation}
\begin{proposition}
It holds true that
\[
\rr{K}_{B,\Omega}\;\subset\; \rr{N}_{B,\Omega}\;\subset\; \rr{M}_{B,\Omega}\;.
\]
Moreover, $\rr{N}_{B,\Omega}$ is
a weakly dense two-sided self-adjoint  ideal in $\rr{M}_{B,\Omega}$.  
\end{proposition}
\proof
The first inclusion is obvious. 
For the rest of the proof we will follow the arguments of \cite[Proposition 2.1.6 (a)]{lenz-99}.
Let  $G\in L^2(\Omega\times\R^2)$ 
such that $K_G\in \rr{N}_{B,\Omega}$. 
Observe that $\s{K}_0=C_c(\Omega\times\R^2)$ can be identified with a dense subspace of $\s{H}_\Omega\simeq L^2(\Omega\times\R^2)$ via the map $\s{K}_0\ni \Phi\mapsto \hat{\phi}\in \s{H}_\Omega$ given by ${\phi}_\omega(x):=\Phi(\omega,x)$. Similarly $\psi^G_\omega(x):=G(\omega,x)$
identifies $G$ with an element of $\s{H}_\Omega$.
Following \cite[Part I, Chapter 5, Definition 1]{dixmier-81}  let us recall that $G$ is \emph{left-bounded} if it exists a bounded operator ${Y}_G\in \bb{B}(\s{H}_\Omega)$ such that $Y_G \hat{\phi}=R_\Phi\hat{\psi}^G$ for every $\Phi\in\s{K}_0$. The set of left-bounded elements is denoted with $\rr{L}_{\rm b}$. In view of \cite[Part I, Chapter 5, Propositions 2 \& Theorem 3]{dixmier-81} one has that $\rr{L}_{\rm b}\subset \rr{L}$ is a two-sided self-adjoint  ideal. The same calculation in the proof of Lemma \ref{lemma_str1} shows that
\[
J^*K_GJ\hat{\phi}\;=\;G\star\Phi\;=\;R_\Phi\hat{\psi}^G\;,
\]
namely $Y_G:=J^*K_GJ$. Then one gets that $\rr{N}_{B,\Omega}=J\rr{L}_{\rm b}J^*\subset J\rr{L}J^*=\rr{M}_{B,\Omega}$ where the last equality is proved in Lemma \ref{lemma_str1}. Finally, $\rr{N}_{B,\Omega}$ is 
 weakly dense since it contains $\rr{K}_{B,\Omega}$ which is weakly dense.
\qed

\begin{remark}\label{rk:L2-ker}
It is worth 
focusing one the inclusion in definition \eqref{N-id}. Depending on the nature of the measure space $(\Omega,\n{P})$ there are operators with $L^2$-kernels that are unbounded. For instance, let us consider the function $G(\omega, x)=g(\omega)\psi_{j,k}(x)$ with $g\in L^2(\Omega)$ and $\psi_{n,m}$ the Laguerre function \eqref{eq:lag_pol}.
One has that $G\in L^2(\Omega\times\R^2)$.
With the argument in Proposition \ref{prop:com_teo1} and Remark \ref{rk:one-point} one infers that the associated operator $K_G$ acts fiberwise on $\s{H}_\Omega$ as the product $M_{g,\omega}\Upsilon_{j\mapsto k}$. More explicitly let us consider vectors  
$\hat{\phi}\in \s{H}_\Omega$  of the form $\phi_\omega(x)=\gamma(\omega)\psi(x)$ with $\gamma\in L^2(\Omega)$ and $\psi\in L^2(\R^2)$. Then a direct computation shows that
\[
(K_G\hat{\phi})_\omega(x)\;:=\;(M_{g,\omega}\Upsilon_{j\mapsto k}\phi_\omega)(x)\;=\;
g(\rr{t}_{x}(\omega))\gamma(\omega)(\Upsilon_{j\mapsto k}\psi)(x)\;.
\]
It follows that, if $g$ is unbounded in the sense that there are   $\gamma\in L^2(\Omega)$ 
such that $\gamma g \notin L^2(\Omega)$, then the operator $K_G$ is unbounded. On the ther hand, in the case $\Omega=\{\ast\}$ and $B\neq0$ corresponding to purely magnetic operators,   the \emph{magnetic Young’s inequalities} 
imply that $\rr{N}_{B}=\rr{Z}_{B}$ (we just omitted the redundancy $\{\ast\}$ from the symbols), namely all the magnetic operators with $L^2$-kernel are automatically bounded. This is discused in detail in \cite[Section 2.4]{denittis-sandoval-00}. 
\hfill $\blacktriangleleft$
\end{remark}

We have now all the ingredients to introduce the relevant trace 
on $\rr{M}_{B,\Omega}$.

\begin{theorem}\label{theo_trac}
There exists a unique normal trace $\tau_{\n{P}}$ on $\rr{M}_{B,\Omega}$ such that 
\begin{equation}\label{eq:def_trace}
\tau_{\n{P}}(A^*A)\;:=\;\langle\langle F_A,F_A\rangle\rangle_{0}\;,\qquad\quad A\in \rr{I}_{B,\Omega}
\end{equation}
where $F_A\in L^2(\Omega\times\R^2)$ denotes the $L^2$-kernel of $A$ and the domain of 
 $\tau_{\n{P}}$ is the dense two-sided self-adjoint  ideal
 $$
\rr{I}_{B,\Omega}\;:=\;\{C=A^*B\;|\; A,B\in \rr{N}_{B,\Omega}\}\;\subset\;\rr{N}_{B,\Omega}\;.
 $$
Moreover
 $$
\tau_{\n{P}}(C)\;=\;\langle\langle F_A,F_B\rangle\rangle_{0}\;<\;\infty \;,\qquad\quad C\;=\;A^*B\;\in\; \rr{I}_{B,\Omega}
 $$
 where $F_A$ and $F_B$ are the kernels of $A$ and $B$, respectively. The trace $\tau_{\n{P}}$ is faithful and semi-finite.
\end{theorem}
\proof
The result follows directly from \cite[Part I, Chapter 6, Theorem 1]{dixmier-81} which establishes the existence of a  
faithful, semi-finite, and normal trace $\theta$ on the von Neumann algebra $\rr{L}$. Then, by using the isomorphism $\rr{M}_{B,\Omega}=J\rr{L}J^*$
one can define the trace $\tau_{\n{P}}$ by $\tau_{\n{P}}(A):=\theta(J^*AJ)$
for $A$ in the corresponding domain. The uniqueness of $\tau_{\n{P}}$  is proved in \cite[Lemma 2.2.1]{lenz-99}. The inclusion $\rr{I}_{B,\Omega}\subset\rr{N}_{B,\Omega}$ follows since $\rr{N}_{B,\Omega}$ is an ideal.
\qed

\medskip

The trace $\tau_{\n{P}}$ is usually called \emph{trace per unit volume}. The reason of the name is justified in Appendix \ref{sec:tra_UV}.

\begin{remark}
It is worth observing that the trace $\tau_{\n{P}}$, as specified by the 
prescription
\eqref{eq:def_trace}, is completely determined by continuity by the  the prescription \eqref{eq:tra_un_vol_01} on $\rr{K}_{B,\Omega}$. 
By \cite[Part I, Chapter 5, Proposition 4]{dixmier-81} and in view of the isomorphism $\rr{N}_{B,\Omega}=J\rr{L}_{\rm b}J^*$  for any $A\in \rr{I}_{B,\Omega}$ there is a sequence $\{K_{F_n}\}_{n\in\N}\subseteq \rr{K}_{B,\Omega}$ such that $K_{F_n}\to A$ strongly and $\|K_{f_n}\|<M$
for some positive constant $M>0$. Moreover
$$
\lim_{n\to\infty}\tau_{\n{P}}\big((K_{F_n}-A)^*(K_{F_n}-A)\big)\;=\;\lim_{n\to\infty}\langle\langle F_n-F_A,F_n-F_A\rangle\rangle_{0}\;=\;0\;
$$
where $F_A$ is the $L^2$-kernel of $A$. From 
$$
\begin{aligned}
\tau_{\n{P}}(K_{F_n}^*K_{F_n}-A^*A)\;=\;&\tau_{\n{P}}\big((K_{F_n}-A)^*(K_{F_n}-A)\big)\\
&+\;\tau_{\n{P}}\big(A^*(K_{F_n}-A)\big)\;+\;
\tau_{\n{P}}\big((K_{F_n}^*-A^*)A\big)\;,
\end{aligned}
$$
and the application of  \cite[Lemma 3.2.14]{denittis-lein-book}, one obtains
$$
\tau_{\n{P}}(A^*A)\;=\;\lim_{n\to\infty}\tau_{\n{P}}(K_{F_n}^*K_{F_n})\;.
$$
Then, the prescription \eqref{eq:tra_un_vol_01}  fixes uniquely by continuity the prescription \eqref{eq:def_trace} which in turn fixes uniquely the normal trace $\tau_{\n{P}}$.
\hfill $\blacktriangleleft$
\end{remark}

The computation of the trace $\tau_{\n{P}}$ has a simple expression in terms of the integration of the integral kernel of regular element. Let $A\in \rr{K}_{B,\Omega}$ be a positive element with kernel $F_A\in\s{K}_0$. In view of the positivity, one has that $A=B^*B$ with $B\in\rr{K}_{B,\Omega}$ and kernel $F_B\in\s{K}_0$. From \eqref{eq:def_trace} one deduces that
\[
\tau_{\n{P}}(A)\;=\;\langle\langle F_B,F_B\rangle\rangle_{0}\;=\;\frac{1}{2\pi\ell^2}\int_{\R^2}\dd x\int_{\Omega}\dd\n{P}(\omega)\;
 |F_B(\omega,x)|^2\;.
\]
On the other hand the integral kernel of $A$ is given by
\[
F_A(\omega,x)\;:=\;\frac{1}{2\pi\ell^2}\int_{\R^2}\dd y\; \overline{F_B(\rr{t}_{-y}(\omega),-y)}\; F_B\left(\rr{t}_{-y}(\omega), x-y\right)\;\Theta_B(y,x)  \;.
\]
A comparison between the last two expressions shows that
\begin{equation}\label{eq:sim_form}
\tau_{\n{P}}(A)\;=\;\int_{\Omega}\dd\n{P}(\omega)\; F_A(\omega,0)\;.
\end{equation}
Equation \eqref{eq:sim_form} can be then extended by linearity to generic elements of $\rr{K}_{B,\Omega}$. The same argument can be generalized to all the elements $A\in \rr{I}_{B,\Omega}$ such that $F_A\in C_0(\R^2,L^1(\Omega))$, where $C_0$ denotes the continuous functions which vanish at infinity (see \cite[Remark 2.16]{denittis-sandoval-00}). In particular, let $g\in \bb{A}_\Omega$ and $f\in S(\R^2)$ a Schwarz function. As showed in Proposition \ref{prop:com_teo1}, the function $f\odot g \in \s{S}$ is the kernel of the operator   
$M_g K_f\in \rr{A}_{B,\Omega}$. 
\begin{proposition}
It holds true that $\pi(\s{S})\subset \rr{I}_{B,\Omega}
$ and $\pi(\s{S}^\star)\subset \rr{I}_{B,\Omega}
$. Morever
\[
\tau_{\n{P}}(M_g K_f)\;=\;\left(\int_{\Omega}\dd\n{P}(\omega)\;g(\omega)\right)f(0)\;=\;\tau_{\n{P}}(K_fM_g)
\]
for every $g\in \bb{A}_\Omega$ and $f\in S(\R^2)$.
\end{proposition}
\proof
In view of \cite[Proposition 2.14]{denittis-sandoval-00}, there are $f_1,f_2\in S(\R^2)$ such that
$f=f_1\star f_2$, and in turn
$K_f=K_{f_1}K_{f_2}$. Therefore $M_gK_f=A_1^* A_2$ with $A_1^*,A_2\in \rr{A}_{B,\Omega}$ are the operators with kernels
$F_{A_1^*}(\omega,x):=g(\omega)f_1(x)$ and $F_{A_2}(\omega,x):=f_2(x)$, respectively. 
Looking at the integrability of the kernels one gets that  $A_1^*, A_2\in \rr{N}_{B,\Omega}$. Since 
$\rr{N}_{B,\Omega}$ is a self-adjoint ideal one has  also $A_1\in \rr{N}_{B,\Omega}$ and in turn $M_gK_f\in \rr{I}_{B,\Omega}$. This proves $\pi(\s{S})\subset \rr{I}_{B,\Omega}
$. Since $\pi(\s{S}^\star)=\pi(\s{S})^*$ and $\rr{I}_{B,\Omega}$ is a self-adjoint ideal one also gets the second inclusion. The definition of the trace implies
\[
\begin{aligned}
\tau_{\n{P}}(M_g K_f)\;&=\;\langle\langle
  F_{A_1},F_{A_2}\rangle\rangle_{0}\\&=\; \frac{1}{2\pi\ell^2}\int_{\R^2}\dd x\; f_1(-x)f_2(x)\int_{\Omega}\dd\n{P}(\omega)\;
  g(\rr{t}_{-x}(\omega))\\
  &=\;\left(\int_{\Omega}\dd\n{P}(\omega)\;g(\omega)\right)(f_1\star f_2)(0)
  \end{aligned}
\]
where the last equality follows from the invariance of the measure $\n{P}$
and the explicit expression of the magnetic convolution $f=f_1\star f_2$. From the definition of the trace it follows that $\tau_{\n{P}}(A^*)=\overline{\tau_{\n{P}}(A)}$, for every $A\in\rr{I}_{B,\Omega}$. Since $K_fM_g=(M_{\overline{g}}K_{f^\star})^*$ one gets that
\[
\tau_{\n{P}}(K_fM_g)\;=\;\overline{\left(\int_{\Omega}\dd\n{P}(\omega)\;\overline{g(\omega)}\right)\overline{f(0)}}
\]
and this completes the proof.
\qed

\medskip

From the las result one infers that
\begin{equation}
\tau_{\n{P}}(M_g \Upsilon_{n\mapsto m})\;=\;\tau_{\n{P}}(\Upsilon_{n\mapsto m} M_g )\;=\;\delta_{n,m}
\int_{\Omega}\dd\n{P}(\omega)\;g(\omega)
\end{equation}
for every $g\in\bb{A}_\Omega$, in view of the fact that the kernel of  $\Upsilon_{n\mapsto m}$ is proportional to the Laguerre function $\psi_{n,m}$ and $\sqrt{2\pi} \ell\psi_{n,m}(0)=\delta_{n,m}$
(\cf Remark \ref{rk:one-point}). Finally, by using the fact that $\rr{I}_{B,\Omega}$ is an ideal, one gets that $\Upsilon_{j\mapsto k} M_g \Upsilon_{n\mapsto m}\in \rr{I}_{B,\Omega}$. Then, the cyclicity of the trace and the equality  $\Upsilon_{n\mapsto m}\Upsilon_{j\mapsto k}=\delta_{n,k}\Upsilon_{j\mapsto m}$ imply
\begin{equation}
\tau_{\n{P}}(\Upsilon_{j\mapsto k}M_g \Upsilon_{n\mapsto m})\;=\;\delta_{n,k}\delta_{j,m}
\int_{\Omega}\dd\n{P}(\omega)\;g(\omega)\;.
\end{equation}
%

 %------------------------------------------------%
\section{Relevant subspaces}
\label{sect:rel_subspaces}
In this section we will introduce the relevant subspace of $\rr{M}_{B,\Omega}$ where the main result described in Theorem \ref{theo:main_dix_eq_pot} applies.
Let us consider the spaces $\rr{L}^q_{B,\Omega}$ described in \eqref{eq:exp_op}.
In a similar way one can also  define $\rr{S}_{B,\Omega}$ as the space of formal sums associated with rapidly decaying (Schwartz) sequence $\{g_{n,m}\}\in S(\N_0^2, \bb{A}_\Omega)$. We will focus our attention on the spaces $\rr{S}_{B,\Omega}$, $\rr{L}^1_{B,\Omega}$ and $\rr{L}^2_{B,\Omega}$.

\medskip

We will start with a simple result.
\begin{proposition}
One has that 
\[
\rr{S}_{B,\Omega}\;\subset\;\rr{L}^1_{B,\Omega}\;\subset\;\rr{A}_{B,\Omega}\;\subset\; \rr{M}_{B,\Omega}\;.
\]
Moreover, $\rr{S}_{B,\Omega}$ and $\rr{L}^1_{B,\Omega}$ are norm-dense in $\rr{A}_{B,\Omega}$ and weakly dense in $\rr{M}_{B,\Omega}$.
\end{proposition}
\proof
The first inclusion follows from $S(\N_0^2)\subset \ell^1(\N_0^2)$.
The density follows from the inclusion $\pi(\s{F}^\star)\subset \rr{S}_{B,\Omega}$ and the proof of Proposition \ref{prop:com_teo1}. The boundedness of elements $A\in\rr{L}^1_{B,\Omega}$ follows from the inequality
\[
\|A\|\;\leqslant\;\sum_{(n,m)\in\N_0^2}\|\Upsilon_{n\mapsto m}M_{g_{n,m}}\|
\;\leqslant\;\sum_{(n,m)\in\N_0^2}\|g_{n,m}\|_\infty\;<\;+\infty
\]
obtained by using $\|\Upsilon_{n\mapsto m}\|=1$ and  $\|M_{g_{n,m}}\|=\|g_{n,m}\|_\infty$.
\qed

\medskip

Let us focus now on the space $\rr{L}^2_{B,\Omega}$. From 
$\ell^1(\N_0^2)\subset \ell^2(\N_0^2)$ it follows that 
$\rr{L}^1_{B,\Omega}\subset \rr{L}^2_{B,\Omega}$. However, the relevant point here is to prove that  elements of $\rr{L}^2_{B,\Omega}$ are bounded.  Let $A\in \rr{L}^2_{B,\Omega}$ and consider its kernel given by
\begin{equation}\label{eq:Ker_AA}
F_A(\omega,x)\;:=\;\sqrt{2\pi} \ell\sum_{(n,m)\in\N_0^2}g_{n,m}(\omega)\psi_{n,m}(x)\;.
\end{equation}
By using  the fact that the Laguerre functions are an orthonormal basis of $L^2(\R^2)$, one obtains
\[
\|F_A\|^2_{L^2(\Omega\times \R^2)}\;=\;\sum_{(n,m)\in\N_0^2}\|g_{n,m}\|^2_{L^2(\Omega)}\;\leqslant\;\sum_{(n,m)\in\N_0^2}\|g_{n,m}\|^2_\infty\;<\;+\infty,
\]
where  the inequality $\|g_{n,m}\|_{L^2(\Omega)}\leqslant \|g_{n,m}\|_\infty$
follows from the fact that $(\Omega,\n{P})$ is a probability space. The last computation shows that $F_{A}\in L^2(\Omega\times \R^2)$, and in turn  
one infers that $\rr{L}^2_{B,\Omega}\subset \rr{Z}_{B,\Omega}$, \ie
the elemnts of $\rr{L}^2_{B,\Omega}$ have  $L^2$-kernel. However, this fact by itself is not yet sufficient to prove the boundedness, and a finer analysis is needed.
\begin{lemma}\label{lem_cruc_trac}
It holds true that $\rr{L}^2_{B,\Omega}\subset\rr{N}_{B,\Omega}$.
\end{lemma}
\proof
Let $A=\sum_{(n,m)\in\N_0^2}\Upsilon_{n\mapsto m}M_{g_{n,m}}$
an element in $\rr{L}^2_{B,\Omega}$ associated with the $\ell^2$-sequence $\{g_{n,m}\}$. 
The kernel of $A$ is given by \eqref{eq:Ker_AA} and we already know that $A\in \rr{Z}_{B,\Omega}$. Therefore, we only need to show that $A\in\bb{B}(\s{H}_\Omega)$.
Let us consider the function $G_A:\Omega\to[0,\infty)$ defined by
$G_A(\omega):=\|F_{A}(\omega,\cdot)\|^2_{L^2(\R^2)}$.
Using the fact that the Laguerre functions are an orthonormal basis of $L^2(\R^2)$ one immediately gets
\[
G_A(\omega)\;=\;2\pi \ell^2\sum_{(n,m)\in\N_0^2}|g_{n,m}(\omega)|^2\;.
\]
It follows form the  $\ell^2$-assumption for the  $g_{n,m}$ that $\|G_A\|_\infty<\infty$
and in turn one has that $G_A\in C(\Omega)$. This shows that $F_{A}$ can be interpreted as a continuous function on $\Omega$ with values of  $L^2(\R^2)$, \ie
\[
F_{A}\;\in\; C(\Omega,L^2(\R^2))\;\simeq\;C(\Omega)\otimes_\varepsilon L^2(\R^2)\;\subset\;L^2(\Omega\times \R^2)\;.
\]
where the  isomorphism is is in the sense of \cite[Theorem 44.1]{treves-67}. Let us consider the operator $L_{F_A}=J^*\pi(F_A)J=J^*AJ$, where we are abusing of the notations introduced in  
Section \ref{sec:tr_u_v} and using the computation provided in  Lemma \ref{lemma_str1}. Since $J$ is a unitary operator on $\s{H}_\Omega\simeq L^2(\Omega\times\R^2)$ it follows that 
$\|A\|=\|AJ\|=\|JL_{F_A}\|$. 
Looking at the definition of $L_{F_A}$ one realizes that, for each fixed $\omega\in\Omega$, the function $(JL_{F_A}\hat{\phi})_\omega:\R^2\to\C$ is the magnetic convolution between
the two functions $f_\omega,h_\omega\in L^2(\R)$ defined by $f_\omega(x):=F_A(\omega,x)$ 
and $h_\omega(x):=\phi_{\rr{t}_x(\omega)}(x)$, respectively.
In view of the \emph{magnetic Young’s inequalities} (see \cite[Section 2.4]{denittis-sandoval-00}), one obtains 
\[
\begin{aligned}
\|(JL_{F_A}\hat{\phi})_\omega\|_{L^2(\R^2)}\;&\leqslant\;\frac{1}{\sqrt{2\pi}\ell}\|f_{\omega}\|_{L^2(\R^2)}\;\|h_\omega\|_{L^2(\R^2)}\;\\
&=\;\frac{\sqrt{G_A(\omega)}}{\sqrt{2\pi}\ell}\;\|h_\omega\|_{L^2(\R^2)}\;.
\end{aligned}
\]
After integrating over $\Omega$ one obtains the following bound for the norm in $\s{H}_\Omega$. 
\[
\|JL_{F_A}\hat{\phi}\|\;\leqslant \sqrt{\frac{{\|G_A\|_\infty}}{{2\pi}\ell^2}}\;\|J\hat{\phi}\|\;=\;\sqrt{\frac{{\|G_A\|_\infty}}{{2\pi}\ell^2}}\;\|\hat{\phi}\|.
\]
This shows that $JL_{F_A}$, and in turn $A$ are bounded operators.
 \qed

\medskip

Lemma \ref{lem_cruc_trac} provides the crucial ingredient for studying the good properties of elements in   $\rr{L}^1_{B,\Omega}$ with respect the trace $\tau_{\n{P}}$.

\begin{proposition}\label{prop:_pre_trac_dix}
It holds true that   $\rr{L}^1_{B,\Omega}\subset \rr{I}_{B,\Omega}$. Moreover, if $A\in \rr{L}^1_{B,\Omega}$ is the operator associated with the sequence $\{g_{n,m}\}\in \ell^1(\N_0^2, \bb{A}_\Omega)$ then
\begin{equation}\label{eq_tra_SS}
\tau_{\n{P}}(A)\;=\;\sum_{n\in\N_0}\int_{\Omega}\dd\n{P}(\omega)\;g_{n,n}(\omega)\;.
\end{equation}
\end{proposition}
\proof
Let $A=\sum_{(n,m)\in\N_0^2}\Upsilon_{n\mapsto m}M_{g_{n,m}}$
an element in $\rr{L}^1_{B,\Omega}$. 
To prove that $A\in \rr{I}_{B,\Omega}$ one can use Lemma \ref{lem_cruc_trac} and the fact that 
 $A=S_1S_2$ for a pair of operators 
$S_1,S_2\in \rr{L}^2_{B,\Omega}$.
For that, let us set
\[
d_r\;:=\;\sup_{n\in\N_0}\left\{\sqrt{\|g_{n,r}\|_\infty}\right\}\;,\qquad h_r\;:=\;\left\{
\begin{aligned}
&d_r^{-1}&&\text{if}\;\; d_r>0\\
&0&&\text{if}\;\; d_r=0\;.
\end{aligned}
\right.
\]
 Consider the two expressions
\[
S_1\;:=\;\sum_{(r,m)\in\N_0^2}\delta_{r,m} d_r\Upsilon_{r\mapsto m}\;,\qquad S_2\;:=\;\sum_{(n,s)\in\N_0^2}\Upsilon_{n\mapsto s} h_sM_{g_{n,s}}\;.
\]
The operator $S_1$ is defined by the  sequence $\{\delta_{r,m} d_r\}\in \ell^2(\N_0^2)\hookrightarrow \ell^2(\N_0^2,\bb{A}_\Omega)$, where the last inclusion is justified by the fact that $\bb{A}_\Omega$ is unital. Therefore $S_1\in \rr{L}^2_{B,\Omega}$.
By definition on has that $\|h_rM_{g_{n,r}}\|^2\leqslant \|M_{g_{n,r}}\|=\|g_{n,r}\|_\infty$.
Therefore also $S_2\in \rr{L}^2_{B,\Omega}$. A direct computation shows that 
\[
\begin{aligned}
S_1S_2\;&=\;\left(\sum_{(r,m)\in\N_0^2}\delta_{r,m} d_r\Upsilon_{r\mapsto m}\right)\left(\sum_{(n,s)\in\N_0^2}\Upsilon_{n\mapsto s} h_s M_{g_{n,s}}\right)\\
&=\;\sum_{(n,m,s)\in\N_0^3} d_m\Upsilon_{m\mapsto m}\Upsilon_{n\mapsto s} h_s M_{g_{n,s}}\\
&=\;\sum_{(n,m,s)\in\N_0^3} d_m\delta_{m,s}\Upsilon_{n\mapsto m} h_s M_{g_{n,s}}\\
&=\;\sum_{(n,m)\in\N_0^2} d_m \Upsilon_{n\mapsto m} h_m M_{g_{n,m}}\;=\;A\;,\\
\end{aligned}
\]
proving the fist part of the claim. Formula \eqref{eq_tra_SS} follows from the definition
$\tau_{\n{P}}(A)=\langle\langle F_{S_1}^*,F_{S_2}\rangle\rangle_{0}$ where $F_{S_1}$ and $F_{S_2}$
are the kernels of $S_1$ and $S_2$, respectively. By observing that
\[
F_{S_1}(\omega,x)\;:=\;\sqrt{2\pi} \ell \sum_{(r,m)\in\N_0^2}\delta_{r,m} d_r \psi_{r,m}(x)
\]
and
\[
F_{S_2}(\omega,x)\;:=\;\sqrt{2\pi} \ell\sum_{(n,s)\in\N_0^2}\frac{g_{n,s}(\omega)}{d_s}\psi_{n,s}(x)\;,
\]
and making use of the fact that the Laguerre functions $\psi_{n,m}$ are an orthonormal basis of $L^2(\R^2)$, one directly obtains \eqref{eq_tra_SS}.
It is worth to observe that 
\[
\left|\sum_{n\in\N_0}\int_{\Omega}\dd\n{P}(\omega)\;g_{n,n}(\omega)\right|\;\leqslant\;
\sum_{n\in\N_0}\|g_{n,n}\|_\infty\;<\;+\infty
\]
which shows that the right-hand side of \eqref{eq_tra_SS} is always well defined for an element in $A\in \rr{L}^1_{B,\Omega}$.
\qed

%------------------------------------------------%
\section{Relation with the Dixmier trace}
\label{sect:dixmier_gen_elem}
In this section, we will provide the proof Theorem \ref{theo:main_dix_eq_pot}
which relates the trace $\tau_{\n{P}}$ with  the {Dixmier trace} 
${\Tr}_{{\rm Dix}}$ defined by \eqref{eq:recip_Dix_Tr}
 weighted by the operator $Q^{-1}_\lambda$ defined by \eqref{eq:op-Q-2}.
 The relation between the trace $\tau_{\n{P}}$ and the trace per unit volume 
 $\s{T}_{\rm u.v.}$ as defined in \eqref{eq:TUV} is discussed in Appendix \ref{sec:tra_UV}. We assume here some familiarity of the reader with the theory of the Dixmier trace, and we refer for more details to the specific literature mentioned in Section \ref{sec:Intr0}.

\medskip

Let us start by mentioning that  Theorem \ref{theo:main_dix_eq_pot} for the special case $\Omega=\{\ast\}$
has been proved in \cite[Proposition 2.27]{denittis-sandoval-00} and the proof is based on the preliminary result described in the following. 
Let $\rr{S}^{1^+}_{\rm m}\subset \rr{S}^{1^+}$ denotes the space of \emph{measurable} operators (\ie the space of operators for which the values of the Dixmier trace is independent of the choice of the generalized limit). Let
$$
T_{j\mapsto k}\;\in\;\left\{Q_{\lambda}^{-1}\Upsilon_{j\mapsto k},\Upsilon_{j\mapsto k}Q_{\lambda}^{-1},Q_{\lambda}^{-\frac{1}{2}}\Upsilon_{j\mapsto k}Q_{\lambda'}^{-\frac{1}{2}}\right\}\;.
$$
Then $T_{j\mapsto k}\in \rr{S}^{1^+}_{\rm m}$
and
\begin{equation}\label{eq:inic_eq}
{\Tr}_{\rm Dix}\big(T_{j\mapsto k}\big)\;=\;\delta_{j,k}\;=\;\tau_{\n{P}}(\Upsilon_{j\mapsto k})\;.
\end{equation}
independently of $\lambda,\lambda'>-1$. The proof of this result is contained in \cite[Corollary 2.26 \& Lemma 3.10]{denittis-sandoval-00}.
The next step is to generalize the equality \eqref{eq:inic_eq} replacing 
$\Upsilon_{j\mapsto k}$ with $\Upsilon_{j\mapsto k}M_g$ for some $g\in \bb{A}_\Omega$. For that, we will need a series of preliminary results.

\medskip

We say that a measurable subset $\Sigma\subset\R^2$ has \emph{finite density} if
\begin{equation}\label{eq:def_dens}
 \text{\upshape dens}[\Sigma] \;:=\; \lim_{\rho\to+\infty}\frac{1}{|B_0(\rho)|}
\int_{B_0(\rho)}\dd x\; \chi_\Sigma(x)\;=\;C\;<\;+\infty\;,
\end{equation}
where $B_0(\rho)\subset\R^2$ is the  ball of radius $\rho > 0$ centered at zero,
$|B_0(\rho)|$ is its volume,   and $\chi_\Sigma$ is the characteristic function of $\Sigma$. 
We will use the symbol $P_\Sigma$ for the projection associated to $\Sigma$, \ie
$(P_\Sigma\varphi)(x):=\chi_\Sigma(x)\varphi(x)$ for every $\varphi\in L^2(\R^2)$.  
\begin{lemma}\label{lem:densiti_set_trace_0}
Let $\Sigma\subset\R^2$ be a measurable subset with {finite density} and 
\begin{equation}\label{eq:D-seq}
\s{D}_N^{(i,j)}[\Sigma]\;:=\;\frac{1}{N}\sum_{m=0}^{N-1}\langle\psi_{i,m},P_\Sigma\psi_{j,m}\rangle_{L^2},
\end{equation}
where 
$\psi_{j,m}$ are the Laguerre functions defined by \eqref{eq:lag_pol}. Then
$$
\begin{aligned}
&\lim_{N\to+\infty}\s{D}_N^{(i,j)}[\Sigma]\;=\;\delta_{i,j}\;  \text{\upshape dens}[\Sigma],
\end{aligned}
$$
for every $i,j\in\N_0$.
\end{lemma}
\proof
After introducing the change of coordinates
$$
x_1\;:=\;\ell\sqrt{2 r}\; \cos\theta\;,\qquad x_2\;:=\;\ell\sqrt{2 r}\; \sin\theta\;,
$$
which implies $\dd x_1\dd x_2=\ell^2\dd r\dd\theta$,
a direct computation provides
$$
\langle\psi_{i,m},P_\Sigma\psi_{j,m}\rangle_{L^2}\;=\;\frac{1}{2\pi}\int_0^{2\pi}\dd\theta\; \expo{\ii (i-j)\theta}\int_0^{+\infty}\dd r\;\tilde{\chi}_{\Sigma}(r,\theta)\; {R}_m^{(i,j)}(r),
$$
where
$$
{R}_m^{(i,j)}(r)\;:=\:\frac{\sqrt{i!j!}}{m!}\expo{-r}r^{m-\frac{i+j}{2}}L^{m-i}_i(r)L^{m-j}_j(r)
$$
and
$$
\tilde{\chi}_{\Sigma}(r,\theta)\;:=\;{\chi}_{\Sigma}\big(x_1(r,\theta),x_2(r,\theta)\big)\;.
$$
Therefore
$$
\s{D}_N^{(i,j)}[\Sigma]\;=\;\frac{1}{2\pi}\int_0^{2\pi}\dd\theta\; \expo{\ii (i-j)\theta}\int_0^{+\infty}\dd \xi\;\tilde{\chi}_{\Sigma}(N\xi,\theta)\; \s{G}_N^{(i,j)}(\xi),
$$
where $r=N\xi$ and
$$
\s{G}_N^{(i,j)}(\xi)\;:=\;\sum_{m=0}^{N-1}{R}_m^{(i,j)}(N\xi)\;.
$$
Thus, as a consequence of Corollary \ref{rk:GDCT}, one obtains that
$$
\lim_{N\to+\infty}\s{D}_N^{(i,j)}[\Sigma]\;=\;0\;,\qquad \text{if}\quad i\neq j\;.
$$

To deal with the case  $i=j$ let us start by observing that
$$
\begin{aligned}
\frac{1}{2\pi}\int_0^{2\pi}\dd\theta\
\int_0^{1}\dd \xi\;\tilde{\chi}_{\Sigma}(N\xi,\theta)\;&=\;\frac{1}{2\pi N}\int_0^{2\pi}\dd\theta\
\int_0^{N}\dd r\;\tilde{\chi}_{\Sigma}(r,\theta)\\
&=\;\frac{1}{2\pi N\ell^2}\int_{B_0(\rho_N)}\dd x\;{\chi}_{\Sigma}(x)
\end{aligned}
$$
with $\rho_N=\ell\sqrt{2N}$. Since $|B_0(\rho_N)|=2\pi N\ell^2$, one gets
$$
\lim_{N\to+\infty}\frac{1}{2\pi}\int_0^{2\pi}\dd\theta\
\int_0^{1}\dd \xi\;\tilde{\chi}_{\Sigma}(N\xi,\theta)
\;=\;{\rm dens}[\Sigma]\;.
$$
Therefore one has that
$$
\begin{aligned}
\s{D}_N^{(j,j)}[\Sigma]- {\rm dens}[\Sigma]
\;=\;\frac{1}{2\pi}\int_0^{2\pi}\dd\theta\
\int_{0}^{\infty}\dd \xi\; \tilde{\chi}_{\Sigma}(N\xi,\theta)\big[ 
\s{G}_N^{(j,j)}(\xi)-\chi_{[0,1]}(\xi)\big],\;
\end{aligned}
$$
where $\chi_{[0,1]}$ is the characteristic function of the interval $[0,1]$.
By using again Corollary \ref{rk:GDCT},
one obtains that
$$
\lim_{N\to+\infty}\left(\s{D}_N^{(j,j)}[\Sigma]- {\rm dens}[\Sigma]\right)\;=\;0
$$
for every $j\in\N_0$, and this concludes the proof.
\qed

\medskip

For the next result let us introduce the sequence
\begin{equation}\label{eq:W-seq}
\s{W}^{(j,k)}_N[\Sigma]\;:=\;\frac{1}{\log(N)}\sum_{m=0}^{N-1}\frac{\s{D}_{m+1}^{(j,k)}[\Sigma]}{m+\zeta}
\end{equation}
where $\s{D}_N^{(j,k)}[\Sigma]$ is defined by \eqref{eq:D-seq}.

\begin{lemma}\label{lem:densiti_set_trace_1}
Let $\Sigma\subset\R^2$ be a measurable subset with {finite density}. Then 
$$
\begin{aligned}
&\lim_{N\to+\infty}\s{W}_N^{(i,j)}[\Sigma]\;=\;\delta_{i,j}\;  \text{\upshape dens}[\Sigma],
\end{aligned}
$$
for every $i,j\in\N_0$.
\end{lemma}
\proof
The proof is a consequence of the {Stolz-Ces\'{a}ro Theorem} 
\cite[Theorem 1.22]{muresan-09}. Indeed, after rewriting 
$$
\s{D}_N^{(i,j)}[\Sigma]\;=\;\frac{a_N-a_{N-1}}{b_N-b_{N-1}}
$$
with
$$
a_N\;:=\;\log(N)\s{W}^{(i,j)}_N[\Sigma]\;, \qquad b_N\;:=\;\sum_{n=0}^{N-1}\frac{1}{n+\zeta}\;,
$$
one gets that 
$$
\lim_{N\to+\infty}\s{D}_N^{(i,j)}[\Sigma]\;=\;\lim_{N\to+\infty}\frac{a_N}{b_N}\;=\;\lim_{N\to+\infty}\frac{\log(N)}{b_N}\s{W}^{(i,j)}_N[\Sigma]\;=\;\lim_{N\to+\infty}\s{W}^{(i,j)}_N[\Sigma]\;.
$$
The result of Lemma \ref{lem:densiti_set_trace_0} concludes the proof.
\qed

\begin{proposition}\label{lem:densiti_set_trace}
Let $\Sigma\subset\R^2$ be a measurable subset with {finite density}. Then $Q_\lambda^{-1}\Upsilon_{j\mapsto k} P_\Sigma$ is an element of $\rr{S}^{1^+}_{\rm m}$
and
\begin{equation}\label{eq:dix_tr_fog}
\begin{aligned}
&\text{\upshape Tr}_{\text{\upshape Dix}}\left(Q_\lambda^{-1}\Upsilon_{j\mapsto k} P_\Sigma \right)%\; =\;\text{\upshape Tr}_{\text{\upshape Dix}}\left(\Upsilon_{j\mapsto k}Q_\lambda^{-1}P_\Sigma \right)
\;=\;\delta_{j,k}\;  \text{\upshape dens}\big[\Sigma\big]
\end{aligned}
 \end{equation}
for every $j,k\in\N$ and independently of $\lambda>-1$.
\end{proposition}
\proof
The relation $\Pi_k\Upsilon_{j\mapsto k}=\Upsilon_{j\mapsto k}$ allows to write
$$
Q_\lambda^{-1}\Upsilon_{j\mapsto k} P_\Sigma \;=\;(Q_\lambda^{-1}\Pi_k)(\Upsilon_{j\mapsto k} P_\Sigma)\;.
$$
Since $Q_\lambda^{-1}\Pi_k$ is a measurable operator \cite[Corollary 2.26]{denittis-sandoval-00} and $\Upsilon_{j\mapsto k} P_\Sigma$ is bounded, it follows that $Q_\lambda^{-1}\Upsilon_{j\mapsto k} P_\Sigma\in \rr{S}^{1^+}$. For the computation of the Dixmier trace we can use the formula in \cite[Lemma 7.17]{gracia-varilly-figueroa-01} which provides 
\[
\text{\upshape Tr}_{\text{\upshape Dix},{\rm Lim}}\left(Q_\lambda^{-1}\Upsilon_{j\mapsto k} P_\Sigma \right)\; =\;{\rm Lim}
\left[
\frac{1}{\log(N)}{\rm Tr}(E_N Q_\lambda^{-1}\Upsilon_{j\mapsto k} P_\Sigma)
\right],
\]
where 
\[
E_N:=\left(\sum_{m=0}^{N-1}\sum_{r\in\N_0}\ketbra{\psi_{r,m}}{\psi_{r,m}}\right)\Pi_k\;=\;\sum_{m=0}^{N-1}\ketbra{\psi_{k,m}}{\psi_{k,m}}\;.
\]
Then, a direct computation provides
\begin{equation}\label{eq:trac_omegg_01}
\text{\upshape Tr}_{\text{\upshape Dix},{\rm Lim}}\left(Q_\lambda^{-1}\Upsilon_{j\mapsto k} P_\Sigma \right)\; =\;{\rm Lim}\left[\s{A}^{(j,k)}_N[\Sigma]\right]
\end{equation}
where
$$
\s{A}^{(j,k)}_N[\Sigma]\;:=\; \frac{1}{\log(N)}\sum_{m=0}^{N-1}\frac{\langle\psi_{j,m},P_\Sigma\psi_{k,m}\rangle_{L^2}}{m+\zeta}
$$
and $\zeta\equiv\zeta(k,\lambda):=k+1+\lambda>0$. Let $\s{D}_N^{(j,k)}[\Sigma]$ be the sequence defined by \eqref{eq:D-seq}. 
The recurrence relation 
$$
\frac{\langle\psi_{j,m},P_\Sigma\psi_{k,m}\rangle_{L^2}}{m+\zeta}\;=\;\frac{m}{m+\zeta}\left(\s{D}_{m+1}^{(j,k)}[\Sigma]- \s{D}_m^{(j,k)}[\Sigma]\right)\;+\;\frac{\s{D}_{m+1}^{(j,k)}[\Sigma]}{m+\zeta}\;, 
$$
with the convention $\s{D}_0^{(j,k)}[\Sigma]=0$,
leads to
$$
\begin{aligned}
\s{A}^{(j,k)}_N[\Sigma]\;:=\;\frac{1}{\log(N)}\frac{N-1}{N-1+\zeta}\s{D}_N^{(j,k)}[\Sigma]\;&-\;\frac{\zeta}{\log(N)}\sum_{m=1}^{N-1}\frac{\s{D}_m^{(j,k)}[\Sigma]}{(m+\zeta)(m-1+\zeta)}\\
&+\;\s{W}^{(j,k)}_N[\Sigma]\;,
\end{aligned}
$$
where $\s{W}^{(j,k)}_N[\Sigma]$ is defined by \eqref{eq:W-seq}.
Since $\s{D}_N^{(j,k)}[\Sigma]$ is a convergent sequence in view of Lemma \ref{lem:densiti_set_trace_0}, one infers that the 
the first two terms of the right-hand side of the above equation go to zero in the limit $N\to+\infty$, and in turn Lemma \ref{lem:densiti_set_trace_1} implies that
$$
\lim_{N\to+\infty}\s{A}^{(j,k)}_N[\Sigma]\;=\;\lim_{N\to+\infty}\s{W}_N^{(j,k)}[\Sigma]\;=\;\delta_{j,k}\;  \text{\upshape dens}[\Sigma].
$$
As a result, and in view of \eqref{eq:trac_omegg_01}, one obtains that 
$$
\text{\upshape Tr}_{\text{\upshape Dix},{\rm Lim}}\left(Q_\lambda^{-1}\Upsilon_{j\mapsto k} P_\Sigma \right)\; =\;\delta_{j,k}\;  \text{\upshape dens}[\Sigma]
$$
independently of the generalized limit ${\rm Lim}$. The latter fact  ensures that $Q_\lambda^{-1}\Upsilon_{j\mapsto k} P_\Sigma\in\rr{S}^{1^+}_{\rm m}$ and  this concludes the proof.
\qed

\begin{remark}%\label{rk:no-inv}
In Lemma \cite[Lemma A.2]{denittis-sandoval-21} it is proved that
$$
[Q_\lambda^{-1},\Upsilon_{j\mapsto k}]\;\in\;\rr{S}^{1}\;\subset\; \rr{S}^{1^+}_{0},
$$
where $\rr{S}^{1}$ is the ideal of trace class operators.
As a consequence, one has that 
$$
\Upsilon_{j\mapsto k}Q_\lambda^{-1} P_\Sigma\;=\;Q_\lambda^{-1}\Upsilon_{j\mapsto k} P_\Sigma\;-\;[Q_\lambda^{-1},\Upsilon_{j\mapsto k}]P_\Sigma\;\in\; \rr{S}^{1^+}_{\rm m}
$$
and 
$$
\text{\upshape Tr}_{\text{\upshape Dix}}\left(Q_\lambda^{-1}\Upsilon_{j\mapsto k} P_\Sigma \right)\; =\;\text{\upshape Tr}_{\text{\upshape Dix}}\left(\Upsilon_{j\mapsto k}Q_\lambda^{-1} P_\Sigma \right)\;
$$
in view of the vanishing under the  trace of the term containing the commutator.
 \hfill $\blacktriangleleft$
\end{remark}

\medskip

Finally, the next result provides the 
announced generalization of \eqref{eq:inic_eq}.
This result is inspired by \cite[Lemma 4]{bellissard-elst-schulz-baldes-94}.
\begin{proposition}\label{prop:densiti_set_trace_02}
Let $g\in \bb{A}_\Omega$. Then $Q_\lambda^{-1}\Upsilon_{j\mapsto k} M_{g,\omega} $ is an element of $\rr{S}^{1^+}_{\rm m}$ and
\begin{equation}\label{eq:dix_tr_fog_m7}
\begin{aligned}
&\text{\upshape Tr}_{\text{\upshape Dix}}\left(Q_\lambda^{-1}\Upsilon_{j\mapsto k} M_{g,\omega} \right)\; =\;%\text{\upshape Tr}_{\text{\upshape Dix}}\left(Q_\lambda^{-1}M_{g,\omega}\Upsilon_{j\mapsto k} \right)\;=\;
\delta_{j,k}\;  \int_\Omega\dd\n{P}(\nu)\;g(\nu),
\end{aligned}
 \end{equation}
for almost every $\omega\in\Omega$, for every $j,k\in\N$ and independently of $\lambda>-1$. In particular, it follows that 
\begin{equation}\label{eq:dix_tr_fog_m99}
 \n{E}\left[\text{\upshape Tr}_{\text{\upshape Dix}}\left(Q_\lambda^{-1}\Upsilon_{j\mapsto k} M_{g,\omega} \right)\right]\;=\; \tau_{\n{P}}(\Upsilon_{j\mapsto k} M_{g})
 \end{equation}
independently of $\lambda>-1$.
\end{proposition}
\proof
It is straightforward to check that \eqref{eq:dix_tr_fog_m99} follows from 
\eqref{eq:dix_tr_fog_m7} and Proposition \ref{prop:_pre_trac_dix}. For this reason, we will focus on the proof of equality \eqref{eq:dix_tr_fog_m7}. First of all, since $Q_\lambda^{-1}\Upsilon_{j\mapsto k}$ is a measurable operator, it follows that $Q_\lambda^{-1}\Upsilon_{j\mapsto k} M_{g,\omega}\in \rr{S}^{1^+}$ is an element of the Dixmier ideal.  
There is no loss of generality in assuming that $g\geqslant 0$ is a non negative function. Let $\Omega^{\delta}_{r}\subset\Omega$ be the set of $\omega$ for which 
$r\delta\leqslant g(\omega)\leqslant  (r+1 )\delta$ where 
$\delta>0$ is a suitable (small enough) constant and 
$r\in\N_0$. Since $g$ is continuous and 
$\Omega$ is compact, $\Omega^{\delta}_{r}\neq\emptyset$ only for a finite number of indices 
$0\leqslant r\leqslant r_\ast$ (with $r_\ast$ depending on $\delta$). Moreover, $\{\Omega^{\delta}_{r}\}$ provides a (finite) partition of $\Omega$. For a fixed $\omega\in\Omega$, let $\Sigma^\delta_r(\omega)\subset\R^2$ be the set of $x\in\R^2$ such that $\rr{t}_{-x}(\omega)\in\Omega^{\delta}_{r}$.  Using the Birkhoff's ergodic theorem one has that
\begin{equation}\label{eq:equi_dens_vol}
\begin{aligned}
\text{dens}\big[\Sigma^\delta_r(\omega)\big]\; :&=\;\lim_{\Lambda_n\nearrow\R^2}\frac{1}{|\Lambda_n|}\int_{\Lambda_n}\ \dd x\;\chi_{\Omega^{\delta}_{r}}\big(\rr{t}_{-x}(\omega)\big)\\
&=\;\int_\Omega \ \dd\n{P}(\nu)\;\chi_{\Omega^{\delta}_{r}}(\nu)\;=:\;\n{P}\big[\Omega^{\delta}_{r}\big],
\end{aligned}
\end{equation}
where $\{\Lambda_n\}$ is some exhaustive F{\o}lner sequence (for instance the balls used in definition \eqref{eq:def_dens}) and the equality holds for $\n{P}$-almost every $\omega\in\Omega$. 
Let $P_{\Sigma^\delta_r(\omega)}$ be the projection on the subset
$\Sigma^\delta_r(\omega)\subset\R^2$, then
 $$
 \left\|M_{g,\omega}-\sum_{r=0}^{r_\ast}\delta r P_{\Sigma^\delta_r(\omega)}\right\|\;\leqslant\; \delta\;.
 $$
The continuity and the linearity of the Dixmier trace imply that
 $$
 \left|
  \text{\upshape Tr}_{\text{\upshape Dix}}\left(Q_\lambda^{-1}
   \Upsilon_{j\mapsto k} M_{g,\omega}\right)
- \text{\upshape Tr}_{\text{\upshape Dix}}\left(Q_\lambda^{-1}
    \Upsilon_{j\mapsto k} \sum_{r=0}^{r_\ast}\delta rP_{\Sigma^\delta_r(\omega)} \right)
 \right|\;\leqslant\;\delta\left\|Q_\lambda^{-1}
  \  \Upsilon_{j\mapsto k}\right\|_{1+}\;,
 $$
 where for simplicity we are omitting the dependence of the Dixmier trace on the choice of a generalized limit. It follows that 
 $$
  \text{\upshape Tr}_{\text{\upshape Dix}}\left(Q_\lambda^{-1}
   \Upsilon_{j\mapsto k} M_{g,\omega}\right)\;=\; 
 \delta_{j,k}  \sum_{r=0}^{r_\ast}\delta r \n{P}\big[\Omega^{\delta}_{r}\big]\;+\;O(\delta)\;.
 $$
 where we used Proposition \ref{lem:densiti_set_trace} and equality \eqref{eq:equi_dens_vol}.
 Taking the limit $\delta\to0$ (and consequently $r_\ast\to\infty$) one obtains
$$
\lim_{\delta\to0}   \sum_{r=0}^{r_\ast(\delta)}\delta r \n{P}\big[\Omega^{\delta}_{r}\big]\;=\;\int_\Omega \dd\n{P}(\nu)\;g(\nu)\; 
$$   
 and in turn
 $$
 \text{\upshape Tr}_{\text{\upshape Dix}}\left(Q_\lambda^{-1}
  \Upsilon_{j\mapsto k}M_{g,\omega}\right)\;=\;\delta_{j,k}\: \int_\Omega \dd\n{P}(\nu)\;g(\nu) \;.
$$
Since the result doesn't depend on the particular definition of the Dixmier trace, it follows that 
$Q_\lambda^{-1}
  \Upsilon_{j\mapsto k}M_{g,\omega}\in \rr{S}^{1^+}_{\rm m}$ is a measurable operator.
\qed

\begin{remark}%\label{rk:no-inv}
The reader can check that the result of Proposition \ref{prop:densiti_set_trace_02} can be extended to functions $g$ in the von Neumann algebra $L^\infty(\Omega)$ by repeating verbatim the same proof. Indeed, the essential boundedness guarantees that one can still construct the 
finite family
$\{\Omega^{\delta}_{r}\}$ which covers $\Omega$ up to a set of zero measure.
 \hfill $\blacktriangleleft$
\end{remark}

\medskip

We are now in position to provide the proof of the main result enounced in the introduction.

\proof[{\it Proof of Theorem \ref{theo:main_dix_eq_pot}}]
The result follows from a generalization of the argument used in the proof of 
\cite[Proposition 2.27]{denittis-sandoval-00}. The crucial  ingredient of the proof is the estimate of the Dixmier norm of  $\|Q_\lambda^{-1}
  \Upsilon_{j\mapsto k}\|_{1^+}\leqslant 1$. This implies that
  \[
  \|Q_\lambda^{-1}
  \Upsilon_{j\mapsto k}M_{g,\omega}\|_{1^+}\;\leqslant\;\|M_{g,\omega}\|
  \;\leqslant\;\|g\|_{\infty},
  \]
for every $g\in \bb{A}_\Omega$. Now, let  $S=\sum_{(n,m)\in\N_0^2}\Upsilon_{n\mapsto m}M_{g_{n,m}}\in\rr{L}^1_{B,\Omega}$ and denote with $S_\omega$ its restriction on the fiber over $\omega\in\Omega$ of the direct integral.
One infers that
\[
\|Q_\lambda^{-1}
 S_\omega\|_{1^+}\;\leqslant\;\sum_{(n,m)\in\N_0^2}\|M_{g_{n,m},\omega}\|
  \;\leqslant\;\sum_{(n,m)\in\N_0^2}\|{g_{n,m}}\|_\infty\;<\;\infty
\]
as a consequence of the fact that $S\in\rr{L}^1_{B,\Omega}$.
This implies that $Q_\lambda^{-1}
 S_\omega\in  \rr{S}^{1^+}$ and the continuity of the Dixmier trace with respect to the Dixmier norm provides
\[
\begin{aligned}
 \text{\upshape Tr}_{\text{\upshape Dix}}\left(Q_\lambda^{-1}
 S_\omega\right)\;&=\;\sum_{(n,m)\in\N_0^2} \text{\upshape Tr}_{\text{\upshape Dix}}\left(Q_\lambda^{-1}\Upsilon_{n\mapsto m}M_{g_{n,m},\omega}\right)\\
&=\; \sum_{(n,m)\in\N_0^2} \delta_{n,m}\: \int_\Omega \dd\n{P}(\nu)\;g_{n,m}(\nu)\\
&=\; \tau_{\n{P}}(S),
\end{aligned}
\]
where the second equality follows from Proposition \ref{prop:densiti_set_trace_02} and the last equality follows from
Proposition \ref{prop:_pre_trac_dix}.
Since the result holds for almost every $\omega\in\Omega$ (and independently of $\lambda>-1$), one immediately concludes that
\[
\n{E}\left[ \text{\upshape Tr}_{\text{\upshape Dix}}\left(Q_\lambda^{-1}
 S_\omega\right)\right]\;=\;\tau_{\n{P}}(S)\;.
\]
Finally, since the evaluation of the Dixmier trace of $Q_\lambda^{-1}
 S_\omega$ doesn't depend on the choice of a particular realization of the Dixmier trace, one deduces that $Q_\lambda^{-1}
 S_\omega\in  \rr{S}^{1^+}_{\rm m}$, for almost every $\omega\in\Omega$.
\qed

\appendix

\section{Technical results}
In this appendix, we collect certain technical results that complement the material presented in the main body of the text. Certain results like Lemma
\ref{lemm_scal_lim} and the construction in Appendix \ref{sec_dif_pot} can have a future independent interest.

%-------------%
\subsection{The Hilbert algebra structure}\label{sec_Hilb_str}
This section contains certain technical results necessary to complement the material provided in Section \ref{sec:tr_u_v}.

\medskip

Let us start by providing a \emph{sketch of the proof of Proposition \ref{prop:hilb_alg_ext}}.
 The  proof of (i) and (ii)
amounts to a simple direct calculation based on the invariance of the measures $\n{P}$ and $\dd x$. We leave the details to the reader. Property (iii) follows from the density of $\s{K}_0$ in the Hilbert space $L^2(\Omega\times\R^2)$ and the density of $\{F\star G\ |\ F,G\in \s{K}_0\}$ in $\s{K}_0$. The latter can be proved by observing that $\s{K}_0$ contains an approximated identity (Remark \ref{rk:approx_ident}). Property (iv) follows from
\begin{equation}\label{eq:Hib_Al_01}
 \langle\langle F\star G,F\star G\rangle\rangle_{0}\;\leqslant\;\nnnorm{G}_1^2\;   \langle\langle F,F\rangle\rangle_{0}\end{equation}
which proves the continuity of the map $F\mapsto F\star G$. The continuity of the map $F\mapsto G\star F$ follows in turn by observing that  $G\star F=(F^\star\star G^\star)^\star$ and using the property (i) which ensures the continuity of the adjoint map $F\mapsto F^\star$.
Let us prove \eqref{eq:Hib_Al_01} by starting from
$$
 \langle\langle
  F\star G,F\star G\rangle\rangle_{0}\;:=\;\frac{1}{(2\pi\ell^2)^3}\int_{\R^2}\dd x\int_{\Omega}\dd\n{P}(\omega)\; |(F\star G)(\omega,x)|^2
  $$
where
\[
|(F\star G)(\omega,x)|^2\;=\;\int_{\R^2\times \R^2}\dd y\dd y'\; \Theta_B(y'-y,x)\;H_{\omega,x}(y,y')
\]
with
$$
H_{\omega,x}(y,y')\;:=\;\overline{G\left(\rr{t}_{-y}(\omega), x-y\right)F(\omega,y)}F(\omega,y')G\left(\rr{t}_{-y'}(\omega), x-y'\right)\;.
$$
Therefore one gets
 $$
  |H_{\omega,x}(y,y')|\;\leqslant\;C_{\omega,x}(y,y')\;|F(\omega,y)F(\omega,y')|
 $$ 
where 
$$
\begin{aligned}
C_{\omega,x}(y,y')\;:&=\; |
G\left(\rr{t}_{-y'}(\omega), x-y'\right)G\left(\rr{t}_{-y}(\omega), x-y\right)^*|\;.
\end{aligned}
$$  
Since 
$$
C_{\omega,x}(y,y')\;\leqslant\;\sup_{\omega\in\Omega}C_{\omega,x}(y,y')\;\leqslant\;
\lVert
G(x-y)\rVert_{\infty}\; \lVert
G\left(
x-y'\right)\rVert_{\infty}
$$ 
one gets
$$
\begin{aligned}
|H_{\omega,x}(y,y')|\;\leqslant\;&\big(\lVert
G(x-y)\rVert_{\infty}|F\left(\omega,y\right)|\big)\big(\lVert
G(x-y')\rVert_{\infty}|f\left(\omega,y'\right)|\big)\;.
\end{aligned}
$$
and in turn one obtains
  $$
  \langle\langle
  F\star G,F\star G\rangle\rangle_{0}\;\leqslant\;\frac{1}{(2\pi\ell^2)^3}\int_{\R^2}\dd x\int_{\Omega}\dd\n{P}(\omega)\; \left(\int_{\R^2}\dd y\; \lVert
G(x-y)\rVert_{\infty}|F\left(\omega,y\right)|\right)^2\;.
$$ 
Let us observe that  
  $$\lVert
G(x-y)\rVert_{\infty}|F\left(\omega,y\right)|\;=\;
\big(\lVert
    G(\cdot)\rVert_\infty\ast | F(\omega,\cdot)|\big)(x)
  $$
  is the convolution between the functions $x\mapsto\lVert G(x)\rVert_\infty$
  and $x\mapsto | F(\omega,x) |$. 
   Thus
   $$
   \begin{aligned}
   \int_{\R^2}\dd x\; \left(\int_{\R^2}\dd y\; \lVert
G(x-y)\rVert_{\infty}|F\left(\omega, y\right)|\right)^2\;&=\;\big\|\lVert
    G(\cdot)\rVert_\infty\ast | F(\omega,\cdot)|\big\|_{L^2(\R^2)}^2\\
    &\leqslant\;\big\lVert\rVert G(\cdot)\lVert_\infty\big\rVert_{L^1(\R^2)}^2\;\big\lVert
                                                                              |F(\omega,\cdot)|\big\rVert_{L^2(\R^2)}^2\\
               &=\;(2\pi\ell^2)^2 \nnnorm{G}_1^2\;\big\lVert     |F(\omega,\cdot)|\big\rVert_{L^2(\R^2)}^2\;.
\end{aligned}
   $$
where the second line follows from the Young's convolution inequality and in the last line we used \eqref{eq_L1-norm}.   
Consequently   
   $$
   \begin{aligned}
  \langle\langle
  F\star G,F\star G\rangle\rangle_{0}\;&\leqslant\;\nnnorm{G}_1^2\;\left(\frac{1}{2\pi\ell^2}\int_{\Omega}\dd\n{P}(\omega)\;\big\lVert
                                                                         |F(\omega,\cdot)|\big\rVert_{L^2(\R^2)}^2
                       \right)\\
                       &=\;\nnnorm{G}_1^2\; \langle\langle
  F,F\rangle\rangle_{0}\;
                       \end{aligned}
$$ 
  which is exactly inequality \eqref{eq:Hib_Al_01}.

\medskip

We complete this section with two structural results repeatedly used in Sections \ref{sec:tr_u_v} and \ref{sect:rel_subspaces}.

\begin{lemma}\label{lemma_str1}
It holds true that $J\rr{L}J^*=\rr{M}_{B,\Omega}$.
\end{lemma}
\proof
After the  change of variables $y\mapsto x-y$ one gets that
\[
\begin{aligned}
(L_{F}\hat{\phi})_\omega(x)\;&=\;\frac{1}{2\pi\ell^2}\int_{\R^2}\dd y\; F(\omega,x-y)\; \phi_{\rr{t}_{y-x}(\omega)}\left(y\right)\;  \Theta_B(x,y)\\
&=\;\frac{1}{2\pi\ell^2}\int_{\R^2}\dd y\; F(\omega,x-y)\; \phi_{\rr{t}_{y}(\rr{t}_{-x}(\omega))}\left(y\right)\;  \Theta_B(x,y)\\
&=\;\frac{1}{2\pi\ell^2}\int_{\R^2}\dd y\; F(\omega,x-y)\; 
(J\hat{\phi})_{\rr{t}_{-x}(\omega)}(y)
\;  \Theta_B(x,y)\;.
\end{aligned}
\]
Therefore
\[
(L_{F}\hat{\phi})_{\rr{t}_{x}(\omega)}(x)\;=\;\frac{1}{2\pi\ell^2}\int_{\R^2}\dd y\; F(\rr{t}_{x}(\omega),x-y)\; 
(J\hat{\phi})_{ \omega}(y)
\;  \Theta_B(x,y)\;,
\]
and a comparison with \eqref{eq:c01-ññ_MM}  shows that
\[
(L_{F}\hat{\phi})_{\rr{t}_{x}(\omega)}(x)\;=\;(\pi(F)J\hat{\phi})_\omega(x)\;,
\] 
\ie $JL_FJ^*=\pi(F)$. By density one concludes that $J\rr{L}J^*=\pi(\s{K}_0)''$. Since 
$\pi(\s{K}_0)''$ is closed with  respect to  topologies weaker than the operator topology,
one has the inclusions $\pi(\s{K}_0)\subset \rr{A}_{B,\Omega}\subset \pi(\s{K}_0)''$ which implies that $\rr{A}_{B,\Omega}'= \pi(\s{K}_0)'$
and in turn $\pi(\s{K}_0)''=\rr{M}_{B,\Omega}$. This concludes the proof. \qed
\medskip

\begin{lemma}\label{lemma_str2}
It holds true that $\bb{V}'_\Omega\cap\bb{D}(\s{H}_\Omega)\subseteq J\rr{R}'J^*$.
\end{lemma}
\proof
Again the   change of variables $y\mapsto x-y$ provides after some computations
\[
\begin{aligned}
(R_{F}\hat{\phi})_\omega(x)\;&=\;\frac{1}{2\pi\ell^2}\int_{\R^2}\dd y\; 
F(\rr{t}_{y}(\rr{t}_{-x}(\omega)),y)\;  \Theta_B(x,y)\phi_\omega(x-y)\;.
\end{aligned}
\]
By inserting the action of the unitary $\rr{V}(y)$ defined by \eqref{eq:magn_tral-LL} one can rewrite the above expression as
\[
\begin{aligned}
(R_{F}\hat{\phi})_{\rr{t}_x(\omega)}(x)\;&=\;\frac{1}{2\pi\ell^2}\int_{\R^2}\dd y\; 
F(\rr{t}_{y}(\omega),y)\;  (\rr{V}(y)\hat{\phi})_{\rr{t}_x(\omega)}(x)\;.
\end{aligned}
\]
which, in view of the arbitrarily of $\hat{\phi}$, implies
\[
JR_{F}J^*\;=\;\int_{\R^2}\dd y\; 
N_F(y)\; \rr{V}(y)
\]
where, for every $y$, the operator $N_F(y)$ acts on the elements of the direct integral $\hat{\phi}\in\s{H}_\Omega$  as follows:
\[
(N_F(y)\hat{\phi})_\omega(x)\;=\;F(\rr{t}_{y}(\omega),y)\;\phi_\omega(x)\;.
\]
Since, for every fixed $y\in\R^2$ the operator $N_F(y)$ acts on the fiber over $\omega$ as the multiplication by the constant $F(\rr{t}_{y}(\omega),y)$, it follows that $N_F(y)\in \bb{D}(\s{H}_\Omega)'$. Then, from the integral representation above, one infers that  $JR_{F}J^*$ is an element of the algebra generated by $\bb{D}(\s{H}_\Omega)'$ and $\bb{V}_\Omega$.
As a consequence, if $T\in \bb{V}'_\Omega\cap\bb{D}(\s{H}_\Omega)$, then $T$ commute with $JR_{F}J^*$. This shows that $\bb{V}'_\Omega\cap\bb{D}(\s{H}_\Omega)\subseteq (J\rr{R}J^*)'$.
The proof is completed by the equality $(J\rr{R}J^*)'=J\rr{R}'J^*$.
\qed

%-------------%
\subsection{A scaling-limit formula for Laguerre polynomials}\label{app_scal_lim}
We will provide here a limit formula for a special scaling of certain combinations of the Laguerre polynomials. This formula extends a result initially proved in  \cite[Lemma 3]{hupfer-leschke-warzel-01}.

\medskip

Let  $L_n^{(\alpha)}$ be the   generalized Laguerre polynomial of degree $n\in\N_0$ and parameter $\alpha\in\R$ defined in \eqref{eq:lag_pol-00}. 

\begin{lemma}\label{mIN}
For each $\alpha,\beta\in\R$ and $n,m\in\mathbb{N}_0$ we have that
$$
\int_0^\infty \dd\xi\; \expo{-\xi}\xi^{\frac{\alpha+\beta}{2}}\left|L_n^{(\alpha)}(\xi)L_m^{(\beta)}(\xi)\right|\;\leqslant\; \sqrt{\frac{\Gamma(\alpha+n+1)}{n!}}\sqrt{\frac{\Gamma(\beta+m+1)}{m!}}\;,
$$
where $\Gamma$ denotes the Gamma function.
\end{lemma}
\begin{proof}
Let $\omega_\alpha$ the weight given by $\omega_\alpha(\xi)=\expo{-\xi}\xi^\alpha$. It is well known that $L^{(\alpha)}_n$ belongs to $L^2(\R_+,\dd\omega_\alpha)$ and 
\begin{equation}\label{eq:norm_L}
\left\|L^{(\alpha)}_n\right\|_\alpha\;=\;\sqrt{\frac{\Gamma(\alpha+n+1)}{n!}}\;,
\end{equation}
 where $\|\cdot\|_\alpha$ is the canonical norm on 
the  corresponding weighted $L^2$ space \cite[eq. 8.980]{gradshteyn-ryzhik-07}. Let 
$$
F_n^{(\alpha)}(\xi)\;:=\;\xi^{\frac{\alpha}{2}}\left|L^{(\alpha)}_n(\xi)\right|\;.
$$
 Clearly, $F^{(\alpha)}_n$ belongs to $L^2(\R_+,\dd\omega_0)$ and $\|F^{(\alpha)}_n\|_0=\|L^{(\alpha)}_n\|_\alpha$. Since 
 $$
 \xi^{\frac{\alpha+\beta}{2}}\left|L_n^{(\alpha)}(\xi)L_m^{(\beta)}(\xi)\right|\;=\;F^{(\alpha)}_n(\xi)F^{(\beta)}_m(\xi)\;$$ the result follows from the Cauchy-Schwarz inequality.
\end{proof}

 Consider the family of functions
\begin{equation}
{R}_m^{(i,j)}(\xi)\;:=\:\frac{\sqrt{i!j!}}{m!}\expo{-\xi}\xi^{m-\frac{i+j}{2}}L^{(m-i)}_i(\xi)L^{(m-j)}_j(\xi)\;,\qquad \xi\geqslant 0
\end{equation}
with $i,j,m\in\N_0$, and the \emph{scaled} partial sums
\begin{equation}
\s{G}_N^{(i,j)}(\xi)\;:=\;\sum_{m=0}^{N-1}{R}_m^{(i,j)}(N\xi)\;
\end{equation}
indexed by $N\in\N$.
\begin{lemma}[Scaling-limit formula: pointwise convergence]\label{lemm_scal_lim}
The following scaling-limit relation holds pointwise
$$
\lim_{N\to+\infty}\s{G}_N^{(i,j)}(\xi)\;=\;\delta_{i,j}\;\left\{
\begin{aligned}
&1&\quad&\text{if}\;\;0\leqslant\xi<1\\
&0&\quad&\text{if}\;\;1<\xi<+\infty\;.
\end{aligned}
\right.
$$
\end{lemma}
\proof
%Let us start with the scaling limit when $\xi\neq 0,1$.
The proof of the cases $i=j$  is provided in \cite[Lemma 3]{hupfer-leschke-warzel-01}. Moreover, the special case $\xi=0$ is discussed in Remark \ref{rk:xi=0}. For the other cases, in view of the symmetry  $\s{G}_N^{(i,j)}=\s{G}_N^{(j,i)}$, we can restrict to  the situation $i>j$. Let
$$
g^{(i,j)}_N(\xi)\;:=\;\sum_{m=0}^{N-1}{R}_m^{(i,j)}(\xi)\;
$$
such that $\s{G}_N^{(i,j)}(\xi)=g_N^{(i,j)}(N\xi)$.
Since $iL_{i}^{(m-i)}(\xi)=mL_{i-1}^{(m-i)}(\xi)-\xi L_{i-1}^{(m-i+1)}(\xi)$ \cite[eq. 8.971(4)]{gradshteyn-ryzhik-07}, one obtains the formula
\begin{equation}\label{recur}
ig^{(i,j)}_N(\xi)\;=\;\frac{\sqrt{i!j!}}{\expo{\xi}\xi^{\frac{i+j}{2}}}\;\sum_{m=0}^{N-1}\frac{\xi^m}{m!}\big(mL_{i-1}^{(m-i)}(\xi)-\xi L_{i-1}^{(m-i+1)}(\xi)\big)L^{(m-j)}_j(\xi)\;.
\end{equation}
 Let us consider the case $i>j=0$. Since $L_0^{(k)}(\xi)=1$ \cite[eq. 8.973(1)]{gradshteyn-ryzhik-07}, one gets
%\begin{align*}
%
\begin{equation}\label{eq:g-form}
\begin{aligned}
ig^{(i,0)}_N(\xi)\;&=\;\frac{\sqrt{i!}}{\expo{\xi}\xi^{\frac{i}{2}}}\sum_{m=0}^{N-1}\left(\frac{m}{m!}\xi^mL_{i-1}^{(m-i)}(\xi)-\frac{m+1}{(m+1)!}\xi^{m+1}L_{i-1}^{(m+1-i)}(\xi)\right)\\ 
&=\;-\sqrt{i!}\expo{-\xi}\frac{\xi^{N-\frac{i}{2}}}{(N-1)!}L_{i-1}^{(N-i)}(\xi).
\end{aligned}
\end{equation}
Since $|L_l^{(k)}(\xi)|\leq (l+k)^l\expo{\frac{\xi}{l+k}}$ for $k\geqslant 1-l$ \cite[eq. (42)]{hupfer-leschke-warzel-01}, one obtains that 
\begin{equation}
\begin{aligned}
i\left|\s{G}_N^{(i,0)}(\xi)\right| \;&=\;i\left|g^{(i,0)}_N(N\xi)\right|\;\leqslant\; \sqrt{i!}\expo{-N\xi}\frac{(N\xi)^{N-\frac{i}{2}}}{(N-1)!}(N-1)^{i-1}\expo{\frac{N\xi}{N-1}}\\
&=\;\sqrt{i!}\xi^{-\frac{i}{2}}\expo{-N(\xi-1-\log\xi)}\frac{N^{N-\frac{1}{2}}\expo{-N}}{(N-1)!}\frac{(N-1)^{i-1}}{N^{\frac{i-1}{2}}}
\expo{\frac{N\xi}{N-1}}\;.
\end{aligned}
\end{equation}
%\end{align*}
%By using the relations $L^{(m)}_0(\xi)=1$ and  $L^{(m)}_1(\xi)=1+m-x$ valid  for every $m\in\N_0$ one gets
%$$
%\begin{aligned}
%\s{G}_N^{(0,1)}(\xi)\;&=\;\frac{1}{\sqrt{N\xi}}\expo{-N\xi}\sum_{m=0}^{N-1}\frac{(N\xi)^m}{m!}(m-N\xi)
%\;=\;\expo{-N\xi}\frac{(N\xi)^{N-\frac{1}{2}}}{(N-1)!}\;.
%\end{aligned}
%$$
Using the Stirling's estimate 
$$
\frac{N^{N+\frac{1}{2}}\expo{-N}}{N!}\;=\;\frac{N^{N-\frac{1}{2}}\expo{-N}}{(N-1)!}\;\leqslant\;\frac{1}{\sqrt{2\pi}}\;,
$$
along with other obvious estimates,
one gets
\begin{equation}\label{eq_ineqGN}
i\left|\s{G}_N^{(i,0)}(\xi)\right|\leq  \sqrt{\frac{i!}{2\pi}}\xi^{-\frac{i}{2}}\expo{2\xi}N^{\frac{i-1}{2}}\expo{-N(\xi-1-\log\xi)}\;.
\end{equation}
Therefore, using the elementary inequality $\xi-1\geqslant\log\xi$, one  infers
$$
\lim_{N\to+\infty}\s{G}_N^{(i,0)}(\xi)\;=\;0\;, \qquad \forall \xi\geqslant0\;,\;\; \xi\neq1\;.
$$
For $j>0$, using the identity $L_j^{(m-j)}=L_j^{(m-j+1)}-L_{j-1}^{(m-j+1)}$ \cite[eq. 8.971(5)]{gradshteyn-ryzhik-07} and equation \eqref{recur}, we have that
$$
\begin{aligned}
 ig^{(i,j)}_N(\xi)\;&=\;\frac{\sqrt{i!j!}}{\expo{\xi}\xi^{\frac{i+j}{2}}}\sum_{m=0}^{N-1}\frac{\xi^{m+1}}{m!}L_{i-1}^{(m-i+1)}(\xi)L^{(m-j+1)}_{j-1}(\xi)\\
+\frac{\sqrt{i!j!}}{\expo{\xi}\xi^{\frac{i+j}{2}}}&\sum_{m=0}^{N-1}\left(\frac{m\xi^m}{m!}L_{i-1}^{(m-i)}(\xi)L^{(m-j)}_j(\xi)-\frac{(m+1)\xi^{m+1}}{(m+1)!}L_{i-1}^{(m-i+1)}(\xi)L^{(m-j+1)}_j(\xi)\right)\\
&\;=\;\sqrt{ij}g_N^{(i-1,j-1)}(\xi)\;-\;\sqrt{i!j!}\expo{-\xi}\frac{\xi^{N-\frac{i+j}{2}}}{(N-1)!}L_{i-1}^{(N-i)}(\xi)L_j^{(N-j)}(\xi).
\end{aligned}
$$
%$$
%0\;\leqslant\;\s{G}_N^{(0,1)}(\xi)\;\leqslant\;\expo{-N(\xi-1)}\;\frac{\xi^{N-\frac{1}{2}}}{\sqrt{2\pi}}\;=\;\expo{-N(\xi-1-\log\xi)}\;\frac{\expo{-\frac{1}{2}\log\xi}}{\sqrt{2\pi}}\;.
%$$
Let $D^{(i,j)}_N(\xi)=-\sqrt{i!j!}\expo{-\xi}\frac{\xi^{N-\frac{i+j}{2}}}{(N-1)!}L_{i-1}^{(N-i)}(\xi)L_j^{(N-j)}(\xi)$. Applying the same estimates used for the case $j=0$, we get 
\begin{equation}\label{eqD-ineq}
\left|D^{(i,j)}_N(N\xi)\right|\;\leqslant\;  \sqrt{\frac{i!j!}{2\pi}}\xi^{-\frac{i+j}{2}}\expo{3\xi}N^{\frac{i+j-1}{2}}\expo{-N(\xi-1-\log\xi)}\;.
\end{equation}
Thus, $\lim_{N\to\infty} D^{(i,j)}_N(N\xi)=0$ whenever $\xi\geqslant 0$, $\xi\neq 1$. The case $j=0$ and induction
given by
\begin{equation}\label{eqG-D}
i\;\s{G}_N^{(i,j)}(\xi)\;-\;\sqrt{ij}\;\s{G}_N^{(i-1,j-1)}(\xi)\;=\;D^{(i,j)}_N(N\xi)
\end{equation}
 imply that
$$
\lim_{N\to+\infty}\s{G}_N^{(i,j)}(\xi)\;=\;0\;, \qquad \forall \xi\geqslant0\;,\;\;\xi\neq1\;
$$
as claimed.
\qed

\begin{remark}[The case $\xi=0$]\label{rk:xi=0}
By using the identity
$$
\frac{(-\xi)^m}{m!}L^{(m-i)}_i(\xi)\;=\;\frac{(-\xi)^i}{i!}L^{(i-m)}_m(\xi)
$$
and the fact that $L^{(m-i)}_i(0)=\frac{m!}{i!(m-i)!}$ if $m\geqslant i$ and $L^{(0)}_i(0)=1$ for all $i$ it follows that
$$
{R}_m^{(i,j)}(0)\;:=\:\delta_{i,j}\; \delta_{i,m}\;, \qquad \forall\; i,j,m\in\N_0\;.
$$
From that the relation
$$
\lim_{N\to+\infty}\s{G}_N^{(i,j)}(0)\;=\;\delta_{i,j}\;
$$
follows immediately.
 \hfill $\blacktriangleleft$
\end{remark}
\begin{lemma}[Scaling-limit formula: integrability]\label{lemm_scal_lim2}
It holds true that  $\s{G}_N^{(i,j)}\in L^1(\R_+)$ for every $i,j\in\N_0$ and $N\in\N$. Moreover
$$
\lim_{N\to+\infty}\int_0^\infty \dd\xi\; \left|\s{G}^{(i,j)}_N(\xi)\right|\;=\;\delta_{i,j}\;.
$$
\end{lemma}
\proof
After the change of variables $x:=N\xi$ one has that 
$$
\int_0^\infty \dd\xi\; \left|\s{G}^{(i,j)}_N(\xi)\right|\;=\;\frac{1}{N}\int_0^\infty \dd x\; \left|g^{(i,j)}_N(x)\right|
$$
Let us start from the case $i=j$. From its very definition it follows that  $g^{(i,i)}_N$ is a positive functions and 
$$
\begin{aligned}
\int_0^\infty \dd x\, g^{(i,i)}_N(x)\;&= \;i!\sum_{m=0}^{N-1}\frac{1}{m!}\int_0^\infty \dd x\;\expo{-x} x^{m-i}\left(L^{(m-i)}_i(x)\right)^2\\
&\;=\;i!\sum_{m=0}^{N-1}\frac{\|L_i^{(m-i)}\|^2_{m-i}}{m!}\;=\;N
\end{aligned}
$$
in view of \eqref{eq:norm_L}. This proves that 
$$
\int_0^\infty \dd\xi\; \left|\s{G}^{(i,i)}_N(\xi)\right|\;=\;1
$$
constantly in $N$. For the remaining cases let us consider first  $i>j=0$. According to equation \eqref{eq:g-form} one has that 
$$
\begin{aligned}
\int_0^\infty \dd x\; \left|g^{(i,0)}_N(x)\right|\;= &\;\frac{\sqrt{i!}}{(N-1)!}\int_0^\infty \dd x\;\expo{-x} x^{N-\frac{i}{2}}\left|L_{i-1}^{(N-i)}( x)\right|\\
&\leq \frac{\sqrt{i!}}{(N-1)!}\sqrt{\frac{(N-1)!}{(i-1)!}}\sqrt{N!}\;=\;\sqrt{iN},
\end{aligned}
$$
where the inequality follows from  Lemma \ref{mIN} along with the identity $1=L_{0}^{(N)}( x)$. This implies that
$$
\lim_{N\to+\infty}\int_0^\infty \dd\xi\; \left|\s{G}^{(i,0)}_N(\xi)\right|\;\leqslant\;\lim_{N\to+\infty}\sqrt{\frac{i}{N}}\;=\;0\;.
$$
With a similar argument one can shows that
$$
\begin{aligned}
\int_0^\infty\dd\xi\; \left|D^{(i,j)}_N(N\xi)\right|\;&=\;
\frac{\sqrt{i!j!}}{N(N-1)!}\int_0^\infty\dd x\;\expo{-x}x^{N-\frac{i+j}{2}}\left|L^{(N-i)}_{i-1}(x)L^{(N-j)}_{j}(x)\right|\\
&\leqslant\;\frac{\sqrt{i!j!}}{N(N-1)!}\sqrt{\frac{(N-1)!}{(i-1)!}}\sqrt{\frac{N!}{j!}}\;=\;\sqrt{\frac{i}{N}}
\end{aligned}
$$
which implies that 
$$
\lim_{N\to+\infty}\int_0^\infty\dd\xi\; \left|D^{(i,j)}_N(N\xi)\right|\;=\;0\;.
$$
In view of \eqref{eqG-D} and using an inductive argument
one concludes the proof for the general case.
\qed

\begin{corollary}[Application of the Generalized Dominated Convergence Theorem]\label{rk:GDCT}
 Let $\chi_{[0,1]}$ be the characteristic function of the interval $[0,1]$.
If  $(f_N)$ is sequence in $L^\infty(\R_+)$ such that
$\|f_N\|_\infty\leqslant C$, with $C\leqslant 0$ a positive constant, then

$$
\lim_{N\to+\infty}\int_0^{+\infty}\dd \xi\;f_N(\xi)\left(\s{G}^{(i,j)}_N(\xi)-\delta_{i,j}\chi_{[0,1]}(\xi)\right)\;=\;0.
$$
\end{corollary}
\begin{proof}
Clearly we have that
$$
\left|f_N(\xi)\left(\s{G}^{(i,j)}_N(\xi)-\delta_{i,j}\chi_{[0,1]}(\xi)\right)\right|\;\leqslant\;C\left(\left|\s{G}^{(i,j)}_N(\xi)\right|\;+\; \chi_{[0,1]}(\xi)\right)
$$
and the right-hand side is an integrable function in view of Lemma \ref{lemm_scal_lim2}. Moreover
$$
\lim_{N\to+\infty}f_N(\xi)\left(\s{G}^{(i,j)}_N(\xi)-\delta_{i,j}\chi_{[0,1]}(\xi)\right)\;=\;
\lim_{N\to+\infty}\left(\s{G}^{(i,j)}_N(\xi)-\delta_{i,j}\chi_{[0,1]}(\xi)\right)\;=\;0
$$
pointwise (almost everywhere) in view of Lemma \ref{lemm_scal_lim}. Our result follows from the Generalized Lebesgue Dominated Convergence Theorem \cite[Proposition 11.18]{royden}.
\end{proof}

 %\hfill $\blacktriangleleft$ 

%

\subsection{Thermodynamic interpretation of the canonical trace}\label{sec:tra_UV}
We will justify in this section the interpretation of the canonical trace $\tau_{\n{P}}$ on the von Neumann algebra $\rr{M}_{B,\Omega}$
 of perturbed magnetic operators described in Section \ref{sec:tr_u_v} as the the trace per unit volume $\s{T}_{\rm u.v.}$ described in \eqref{eq:TUV}

\medskip

The group
 $\R^2$ is  locally compact and abelian, hence \emph{amenable}. This means 
 that the von Neumann algebra $L^\infty(\R^2)$
admits a left invariant mean, or equivalently that $\R^2$ meets the F{\o}lner condition (see \cite{greenleaf-69,greenleaf-73} for more details). In particular, $\R^2$ possesses several \emph{monotone exhausting F{\o}lner  sequence} $\{\Lambda_n\}_{n\in\N}$ such that: i) the  $\Lambda_n\subset\R^2$ are compact; ii)
 $\Lambda_n\subset \Lambda_{n+1}$; iii) $\Lambda_n\nearrow\R^2$ and iv)
 $$
 \lim_{n\to\infty}\;\frac{|(\Lambda_n+x)\triangle \Lambda_n|}{|\Lambda_n|}\;=\;0\;,\qquad\quad\forall\; x\in\R^2
 $$
where $\triangle$ denotes the symmetric difference of the set  $\Lambda_n$ and its translate $\Lambda_n+x$ and $|\Lambda|$ is used for the volume (Lebesgue measure) of the set $\Lambda|$.
 Therefore, under the conditions stipulated in Section~\ref{Sec:pot}, the \emph{mean ergodic theorem} \cite[Corollary 3.5]{greenleaf-73} holds:
\begin{align*}
	\lim_{n \to \infty} \frac{1}{|\Lambda_n|} \int_{\Lambda_n} \dd x \; f \bigl ( \rr{t}_{-x}(\omega) \bigr ) 
\;	=\; \int_\Omega\dd \mathbb{P}(\omega') \; f(\omega') 
	&&
	f \in L^1(\Omega)\;,
%	, 
\end{align*}
where  $\omega$ is any point in a set $\Omega_f \subseteq \Omega$ of full measure. 
It is worth to notice that this result is independent of the particular choice of the
monotone exhausting F{\o}lner  sequence.
Let $A \in \rr{N}_{B,\Omega}$ with $L^2$-kernel $F_A$. Then, the function
\begin{align*}
	Z_A(\omega)\; :=\;  \frac{1}{2\pi\ell^2} \int_{\R^2} \dd x \; \left|F_A(\omega,x)\right|^2
\end{align*}
is an element of $L^1(\Omega)$ and the mean ergodic theorem implies 
\begin{align*}
	\lim_{n \to \infty} \frac{1}{|\Lambda_n|} \int_{\Lambda_n} \dd x \; Z_A \big ( \rr{t}_{-x}(\omega) \big) 
	\;=\; \int_{\Omega} \dd \mathbb{P}(\omega') \, Z_A(\omega') 
	\;=\; \tau_{\n{P}} \bigl ( A^* A \bigr ),\;
\end{align*}
where the second equality is just the definition of $\tau_{\n{P}}$.
On the other hand, one can check directly  that 
$$
	\frac{1}{{2\pi\ell^2}}\int_{\Lambda_n} \dd x\; Z_A \big ( \rr{t}_{-x}(\omega) \big ) \;=\;
	\mathrm{Tr}_{{L^2(\R^2)}} \big ( \chi_{\Lambda_n} \, {A}_{\omega}^* \, {A}_{\omega} \, \chi_{\Lambda_n} \big ),
$$
where $\chi_{\Lambda_n}$ denotes the multiplication operator by the characteristic function for $\Lambda_n$ (in fact, a projection). This follows by observing that ${A}_{\omega} \chi_{\Lambda_n}$ is a
 Hilbert-Schmidt operator for almost all $\omega$ \cite[Proposition 2.1.6 (c)]{lenz-99}, and computing the trace by integrating the related kernel.
After putting all the pieces together one obtains
\begin{align}
	\frac{1}{{2\pi\ell^2}}\;\tau_{\n{P}} (A^*A) \;=\; \lim_{n \to \infty} \frac{1}{|{\Lambda_n}|}  \mathrm{Tr}_{L^2(\R^2)}\big (\chi_{\Lambda_n} \, {A}_{\omega}^* \, {A}_{\omega} \, \chi_{\Lambda_n} \big )\;,
%	\bigl ( P_{\Lambda_n} \, A_{\omega} \, P_{\Lambda_n} \bigr )
%	\notag \\
%	&= \lim_{n \to \infty} \frac{1}{\sabs{\Lambda_n}} \, \mathrm{Tr}_{\Hil_{\omega}}
%	\bigl ( P_{\Lambda_n} \, A_{\omega} \, P_{\Lambda_n} \bigr ) 
\label{unified:eqn:trace_per_unit_volume_as_infinite_volume_limit}
\end{align}
for any $A \in  \rr{N}_{B,\Omega}$ and for
 $\mathbb{P}$-almost all $\omega\in\Omega$. 
 A comparison between the right-hand side of 
equation~\eqref{unified:eqn:trace_per_unit_volume_as_infinite_volume_limit}  and the definition \ref{eq:TUV} provides the interpretation of 
$\tau_{\n{P}} $ as the
 \emph{trace per unit volume}. 
 Finally, the definition \ref{eq:tr_un_vol} is justified by the 
  independence of the equality \eqref{unified:eqn:trace_per_unit_volume_as_infinite_volume_limit} by the
  the particular choice of $\omega\in\Omega$ in a set of full measure.
 For more details we refer to \cite[pp.~208-209]{lenz-99} and references therein.

\medskip

The next result proved in \cite[Lemma 2.2.6 \& Theorem 2.2.7]{lenz-99} (see also
\cite[Proposition 4.2.1]{denittis-lein-book}) provides a further recipe to calculate $\tau_{\n{P}}$.
\begin{proposition}
Let $\lambda\in L^{\infty}(\R^2)\cap L^{2}(\R^2)$ be any positive function with normalization $\|\lambda\|_{L^2}=1$. Let $\Sigma_\lambda\in\bb{B}(\s{H}_\Omega)$ be the operator that acts on each fiber of the direct integral $\s{H}_\Omega$ as the multiplication by $\lambda$.
Then
$$
\tau_{\n{P}}(A)\;=\;\int_\Omega\dd\n{P}(\omega)\; {\rm Tr}_{{L^2(\R^2)}}(\Sigma_\lambda A_\omega \Sigma_\lambda)\;,\qquad A\in \rr{I}_{B,\Omega}\;.
$$
\end{proposition}

 %-------%
 \subsection{Differential structure on the algebra of potentials}\label{sec_dif_pot}
 In this section we will construct a dense subalgebra of $\bb{A}_\Omega$ made of differentiable elements. This construction complements the material contained in Section \ref{Sec:pot}.
 
 \medskip

The manifold structure of $\R^2$ can be used for the definition of the {directional derivatives} on  $(\bb{A}_\Omega,\R^2,T)$. 
 An element $g\in\bb{A}_\Omega$ is said to be \emph{(G{\^a}teaux) differentiable} if there are $\partial_1g$ and $\partial_2g$ in $\bb{A}_\Omega$ such that
\begin{equation}\label{eq:gat_deriv_disord}
(\partial_jg)(\omega)\;:=\;\lim_{s\to 0}\frac{T_{s e_j}(g)\big(\omega\big)-g(\omega)}{s},\qquad \quad \forall\ \ \omega\in\Omega
\end{equation}
where $e_1:=(1,0)$ and $e_2:=(0,1)$ provide the canonical basis of $\R^2$. The set of all differentiable elements in $\bb{A}_\Omega$ will be denoted with
 $\text{Diff}(\bb{A}_\Omega)$. If $g$ is a differentiable element one can build the \emph{directional derivatives} 
$$
\partial_ag\;:=\; a_1\; (\partial_1g) \;+\;a_2\; (\partial_2g) $$
for all $a=(a_1,a_2)\in\R^2$.
It turns out that $g\in \text{Diff}(\bb{A}_\Omega)$ implies $\partial_ag\in\bb{A}_\Omega$, for all $a\in\R^2$. Usual computations show that
$\text{Diff}(\bb{A}_\Omega)$ is a unital self-adjoint sub-algebra of $\bb{A}_\Omega$. 
 In particular, the Leibniz's rule holds true, \ie
$$
\partial_a(gh)\;=\;g\;(\partial_ah)(\omega)\;+\;h\;(\partial_ag)
$$
for all $g,h\in\text{Diff}(\bb{A}_\Omega)$ and every $a\in\R^2$.

\medskip

Non-trivial elements of $\text{Diff}(\bb{A}_\Omega)$ can be realized with the following procedure: for
given $\phi\in C^\infty_c(\R^2)$ and $g\in\bb{A}_\Omega$ one define
\begin{equation}\label{eq:mollif}
g_\phi(\omega)\;:=\;\int_{\R^2}\dd y\ \phi(y)\ (T_yg) (\omega )\;=\;\int_{\R^2}\dd y\ \phi(y)\ g (\rr{t}_{-y}(\omega) ).
\end{equation}
Since the map $\R^2\ni y\mapsto\phi(y) (T_y g)\in\bb{A}_\Omega$ is norm continuous, one can use  standard results from the theory of the integration on $C^\ast$-algebras 
\cite[Appendix B]{williams-07}  
to prove that  $g_\phi\in\bb{A}_\Omega$ and $\|g_\phi\|_\infty\leqslant\|\phi\|_{L^1(\R^2)}\ \|g\|_\infty$. Moreover, for for $j=1,2$ the elementary computation 
$$
\begin{aligned}
(\partial_jg_\phi)(\omega)&\;=\;\lim_{s\to 0}\frac{1}{s}\left(\int_{\R^2}\dd y\ \phi(y+s e_j)\ (T_yg)(\omega)-\int_{\R^2}\dd y\ \phi(y)\ (T_yg)(\omega)\right)\\
%\\&=\int_{\R^2}\dd y\ ({\partial_j}\phi)(y)\ f\big(\tau_{-y}(\omega)\big)
&\;=\;g_{{\partial_j}\phi}(\omega)\;,
\end{aligned}
$$
 shows that $g_\phi\in\text{Diff}(\bb{A}_\Omega)$. Interestingly, elements like \eqref{eq:mollif} are sufficiently many in $\bb{A}_\Omega$, in the following sense:
\begin{lemma}\label{lemma:dens01}
 ${\rm Diff}(\bb{A}_\Omega)$ is a dense $\ast$-subalgebra of $\bb{A}_\Omega$.
\end{lemma}
\proof
We need to prove only the density and this can be done with a standard technique (\cf \cite[Example 3, pg. 251]{reed-simon-I}). 
Let $\phi\in{C}_c^\infty(\R^2)$ be a positive function supported in  $B_1:=\{x\in\R^2 \: |x|\leqslant1\}$ and with $\int_{\R^2}\phi(y)\ \dd y=1$. 
For each $n\in\N$ we can define the normalized function $\phi_n(x):=n^2\phi(nx)$ with support in
$B_{1/n}:=\{x\in\R^2 \: |x|\leqslant1/n\}$. The sequence  $\{\phi_n\}_{n\in\N}$ is an {approximate identity}. For a given 
$g\in\bb{A}_\Omega$ one considers the sequence of  $g_{\phi_n}\in {\rm Diff}(\bb{A}_\Omega)$ defined as in \eqref{eq:mollif}. 
 Observe that
$$
\begin{aligned}
|g_{\phi_n}(\omega)-g(\omega)|&\;\leqslant\;\int_{\R^2}\dd y\ \phi_n(y)\ \left|(T_yg) (\omega )-g(\omega)\right|\\
&\;\leqslant\;\text{sup}_{y\in B_{1/n}}\left|(T_yg) (\omega )-g(\omega)\right|\;\to\;0\;\quad\; \text{if}\quad n\to\infty\;.
\end{aligned}
$$
The map $B_1\times\Omega\to\C$ defined by $(y,\omega)\mapsto (T_yg)\big(\omega\big)=g\big(\rr{t}_{-y}(\omega)\big)$ is continuous (by definition) since it is the composition of continuous maps.
Moreover, this map is also uniformly continuous since it is defined on a compact metrizable space $B_1\times\Omega$ \cite[Theorem 1-31]{hocking-young-61}. 
Then, for each $\varepsilon>0$ there exists a $n>0$ such that  $\left|(T_yg)\big(\omega\big)-g(\omega)\right|\leqslant\varepsilon$
if $|y|<1/n$ and $n$ does not depend on $\omega$.
This  shows that  $g_{\phi_n}$ converges to $g$ in the norm of $\bb{A}_\Omega$. 
\qed

\medskip

Let us to point out that in the proof of  Lemma \ref{lemma:dens01} we used in an essential way the fact that the map $(y,\omega)\mapsto \rr{t}_{-y}(\omega)$ is  {jointly} continuous.

%------------------------------% 

\end{document}